\documentclass{lmcs}
\usepackage[utf8]{inputenc}
\pdfoutput=1

\usepackage{lastpage}
\lmcsdoi{18}{3}{11}
\lmcsheading{}{\pageref{LastPage}}{}{}%
{Aug.~13,~2021}{Jul.~29,~2022}{}


\usepackage[inline]{enumitem}
\usepackage{amsmath}
\usepackage{amssymb}
\usepackage{wasysym}
\usepackage{mathrsfs}
\usepackage{array}
\usepackage{wrapfig}

\usepackage{tikz}
\usetikzlibrary{arrows,automata}
\usetikzlibrary{matrix}
\usetikzlibrary{decorations.pathmorphing}
\usetikzlibrary{quotes}

\let\doteqrel\doteq
\renewcommand*{\doteq}{\mathbin{\doteqrel}}

\newcommand*{\lnext}{\ocircle}
\newcommand*{\lanext}{\ocircle_{\chi}}
\newcommand*{\lgnext}[1]{\ocircle^{#1}}
\newcommand*{\lganext}[1]{\chi_{F}^{#1}}
\newcommand*{\lback}{\circleddash}
\newcommand*{\laback}{\circleddash_{\chi}}
\newcommand*{\lgback}[1]{\circleddash^{#1}}
\newcommand*{\lgaback}[1]{\chi_{P}^{#1}}

\newcommand*{\lthrnext}{\mathit{CallThr}}

\newcommand*{\lguntil}[4]{#3 \mathbin{\mathcal{U}^{#1}_{#2}} #4}
\newcommand*{\luntil}[3]{#2 \mathbin{\mathcal{U}^{#1}} #3}
\newcommand*{\lluntil}[2]{\luntil{}{#1}{#2}}
\newcommand*{\lgsince}[4]{#3 \mathbin{\mathcal{S}^{#1}_{#2}} #4}
\newcommand*{\lsince}[3]{#2 \mathbin{\mathcal{S}^{#1}} #3}
\newcommand*{\llsince}[2]{\lsince{}{#1}{#2}}
\newcommand*{\lhyuntil}[2]{\luntil{\uparrow}{#1}{#2}}
\newcommand*{\lhysince}[2]{\lsince{\downarrow}{#1}{#2}}
\newcommand*{\lhtuntil}[2]{\luntil{\downarrow}{#1}{#2}}
\newcommand*{\lhtsince}[2]{\lsince{\uparrow}{#1}{#2}}
\newcommand*{\lcallsince}{\mathit{Scall}}
\newcommand*{\lfuntil}[3]{#2 \mathbin{\mathcal{U}(#1)} #3}

\newcommand*{\lglob}[1]{\square^{#1}}

\newcommand*{\llglob}{\square}

\newcommand*{\leven}[1]{\Diamond^{#1}}

\newcommand*{\ldnext}{\lnext^d}
\newcommand*{\lunext}{\lnext^u}
\newcommand*{\ldback}{\lback^d}
\newcommand*{\luback}{\lback^u}
\newcommand*{\lcnext}[1]{\chi_F^{#1}}
\newcommand*{\lcdnext}{\lcnext{d}}
\newcommand*{\lcunext}{\lcnext{u}}
\newcommand*{\lcback}[1]{\chi_P^{#1}}
\newcommand*{\lcdback}{\lcback{d}}
\newcommand*{\lcuback}{\lcback{u}}
\newcommand*{\lcduntil}[2]{{#1} \mathbin{\mathcal{U}_\chi^d} {#2}}
\newcommand*{\lcdsince}[2]{{#1} \mathbin{\mathcal{S}_\chi^d} {#2}}
\newcommand*{\lcuuntil}[2]{{#1} \mathbin{\mathcal{U}_\chi^u} {#2}}
\newcommand*{\lcusince}[2]{{#1} \mathbin{\mathcal{S}_\chi^u} {#2}}
\newcommand*{\lhnext}[1]{\lnext_H^{#1}}
\newcommand*{\lhback}[1]{\lback_H^{#1}}
\newcommand*{\lhdnext}{\lhnext{d}}
\newcommand*{\lhdback}{\lhback{d}}
\newcommand*{\lhduntil}[2]{{#1} \mathbin{\mathcal{U}_H^d} {#2}}
\newcommand*{\lhdsince}[2]{{#1} \mathbin{\mathcal{S}_H^d} {#2}}
\newcommand*{\lhunext}{\lhnext{u}}
\newcommand*{\lhuback}{\lhback{u}}
\newcommand*{\lhuuntil}[2]{{#1} \mathbin{\mathcal{U}_H^u} {#2}}
\newcommand*{\lhusince}[2]{{#1} \mathbin{\mathcal{S}_H^u} {#2}}

\newcommand*{\chain}{\chi}

\newcommand*{\powset}[1]{{\mathcal{P}(#1)}}
\newcommand*{\lsucc}{\operatorname{succ}}
\newcommand*{\prf}{\pi}
\newcommand*{\pr}{\mathrel{\prf}}

\DeclareMathOperator{\Theight}{h_\lthrow}

\newcommand*{\clos}[1]{\operatorname{Cl}({#1})}

\newcommand*{\ra}{\operatorname{Rc}}
\newcommand*{\invtau}{\tau^{-1}_{AP}}
\newcommand*{\first}{\operatorname{first}}
\newcommand*{\last}{\operatorname{last}}
\newcommand*{\rchild}{R_\Downarrow}
\newcommand*{\rsibl}{R_\Rightarrow}
\renewcommand*{\rparen}{R_\Uparrow}
\newcommand*{\rlsibl}{R_\Leftarrow}
\newcommand*{\uotrees}{\mathcal{T}}

\newcommand*{\xuntil}{$\mathcal{X}_\mathit{until}$}
\newcommand*{\xsnext}{\mathop{\ocircle_{\mathord{\Rightarrow}}}}
\newcommand*{\xsback}{\mathop{\ocircle_{\mathord{\Leftarrow}}}}
\newcommand*{\xsuntil}[2]{{\mathop{\Rightarrow}\left({#1}, {#2}\right)}}
\newcommand*{\xssince}[2]{{\mathop{\Leftarrow}\left({#1}, {#2}\right)}}
\newcommand*{\xcnext}{\mathop{\ocircle_{\mathord{\Downarrow}}}}
\newcommand*{\xcback}{\mathop{\ocircle_{\mathord{\Uparrow}}}}
\newcommand*{\xcuntil}[2]{{\mathop{\Downarrow}\left({#1}, {#2}\right)}}
\newcommand*{\xcsince}[2]{{\mathop{\Uparrow}\left({#1}, {#2}\right)}}
\newcommand*{\xsuntilo}{\mathop{\Rightarrow}}
\newcommand*{\xssinceo}{\mathop{\Leftarrow}}
\newcommand*{\xcuntilo}{\mathop{\Downarrow}}
\newcommand*{\xcsinceo}{\mathop{\Uparrow}}

\newcommand*{\vra}{\overrightarrow{v}}
\newcommand*{\vla}{\overleftarrow{v}}

\newcommand*{\lcall}{\mathbf{call}}
\newcommand*{\lret}{\mathbf{ret}}
\newcommand*{\lhandle}{\mathbf{han}}
\newcommand*{\lthrow}{\mathbf{exc}}

\newcommand{\mrk}[1]{{#1}'}
\newcommand{\oldstack}[3]{%
{\ifthenelse{\equal{#1}{1}}{%
\mrk{#2}
}%
{#2}}_{#3}%
}
\newcommand{\stack}[3]{%
[%
{\ifthenelse{\equal{#1}{1}}{%
\mrk{#2}
}%
{#2}}
{\ifthenelse{\equal{#1}{0}}{\ }{} }
{#3}%
]%
}

\newcommand{\tstack}[2]{%
[#1,\ #2]%
}
\newcommand{\tconfig}[3]{\langle #1, \ #2, \ #3 \rangle}

\newcommand{\transition}[1]{\stackrel {{#1}} \vdash}

\newcommand{\va}[1]{\stackrel{#1}{\longrightarrow}}
\newcommand{\vshift}[1]{\stackrel{#1}{\dashrightarrow}}
\newcommand{\ourpath}[1]{\stackrel{#1}{\leadsto}}
\newcommand{\flush}[1]{\stackrel{#1}{\Longrightarrow}}

\newcommand{\ochain}[3]{{}^{#1}[ #2 ]{}^{#3}}

\newcommand{\symb}[1]{\mathop{smb}(#1)}
\newcommand{\state}[1]{\mathop{st}(#1)}

\begin{document}

\title[A First-Order Complete Temporal Logic for Structured CFLs]%
{A First-Order Complete Temporal Logic for Structured Context-Free Languages}

\author[M. Chiari]{Michele Chiari\lmcsorcid{0000-0001-7742-9233}\rsuper{a}} 
\author[D. Mandrioli]{Dino Mandrioli\lmcsorcid{0000-0002-0945-5947}\rsuper{a}} 
\author[M. Pradella]{Matteo Pradella\lmcsorcid{0000-0003-3039-1084}\rsuper{a,b}} 

\address{\lsuper{a}DEIB, Politecnico di Milano, Italy}
\email{\{michele.chiari, dino.mandrioli, matteo.pradella\}@polimi.it}
\address{\lsuper{b}IEIIT, Consiglio Nazionale delle Ricerche}

\begin{abstract}
  The problem of model checking procedural programs has fostered much research
  towards the definition of temporal logics for reasoning on context-free
  structures. The most notable of such results are temporal logics on Nested
  Words, such as CaRet and NWTL\@. Recently, the logic OPTL was introduced, based
  on the class of Operator Precedence Languages (OPLs), more powerful than Nested
  Words. We define the new OPL-based logic POTL and prove its FO-completeness.
   POTL improves on
  NWTL by enabling the formulation of requirements involving
  pre/post-conditions, stack inspection, and others in the presence of
  exception-like constructs. It improves on OPTL too, which instead we show not to be FO-complete;
  it also allows to
  express more easily stack inspection and function-local properties.
 In a companion paper we report a model checking procedure for POTL and experimental results based on a prototype tool developed therefor. For completeness a short summary of this complementary result is provided in this paper too.
\end{abstract}

\maketitle

\section{Introduction}%
\label{sec:intro}

Linear Temporal Logic (LTL) is one of the most successful languages
for the specification and verification of system requirements.  Being defined on a
linearly ordered algebraic structure, it is ideal
to express safety and liveness properties on a linear flow of
events.  In particular, LTL is equivalent to First-Order Logic
(FOL) on this structure~\cite{Kamp68} and captures the first-order definable
fragment of ($\omega$-)regular languages.
Its satisfiability and model checking are PSPACE-complete
with respect to formula length but polynomial in model size.

The need for model checking software programs lead to several attempts
at widening the expressive power of specifications beyond LTL\@.
Indeed, procedural programs are often modeled with operational formalisms
such as Boolean programs~\cite{BallR00}, Pushdown Systems, and Recursive State
Machines~\cite{AlurBEGRY05}.  Reachability and LTL model checking have
been extensively studied for such systems~\cite{BouajjaniEM97,EsparzaHRS00,AlurBEGRY05}.
However, they present behaviors typical of pushdown automata.
Thus, many properties thereof cannot be expressed by a formalism
whose expressive power is limited to regular languages.
Some early attempts at extending the expressive power of specifications
beyond regular languages in model checking include
equipping LTL with Presburger arithmetic~\cite{BouajjaniH96},
using directly a class of pushdown automata for the specification~\cite{KupfermanPV02b},
and a variant of Propositional Dynamic Logic that can express a restricted class
of Context-Free Languages (CFLs)~\cite{HarelKT02}.

A natural way of generalizing LTL is defining temporal logics on more complex algebraic structures.
The introduction of logics based on \emph{Nested Words}~\cite{jacm/AlurM09},
such as CaRet~\cite{AlurEM04}, has been one of the first attempts in this direction.
Nested Words consist of a discrete, linear order with the addition of a one-to-one
nesting relation. They are equivalent to Visibly Pushdown Languages (VPLs)~\cite{AluMad04},
a.k.a.\ Input-Driven Languages~\cite{DBLP:conf/icalp/Mehlhorn80},
a strict subclass of deterministic CFLs.
The parenthetical nature of VPLs results in the nesting relation on words,
which naturally models the one-to-one correspondence between
function calls and returns in procedural programs.
CaRet is the first temporal logic that extends LTL through explicit modalities
that refer to the underlying model's context-free structure.
Thus, CaRet can easily express Hoare-style pre/post-conditions, stack inspection, and more.

The applications of Nested Words in model checking lead to the
question of which theoretical properties of LTL could be extended to them.
One of the most prominent is First-Order (FO) completeness.
CaRet was devised with the specification of procedural programs in
mind, but its stance with respect to FOL remains unknown, although it
is conjectured not to be FO-complete~\cite{lmcs/AlurABEIL08}.
FO-completeness on Nested Words is thoroughly studied in~\cite{lmcs/AlurABEIL08},
and is achieved with Nested Words Temporal Logic (NWTL) and other logics.
NWTL satisfiability and model checking are EXPTIME-complete
and retain polinomiality on system size.

One natural further question is whether it is possible to define
temporal logics that capture a fragment of CFLs wider than VPLs.
Several practical applications motivate this question: the nesting
relation of Nested Words is one-to-one~\cite{MP18}, preventing its ability to model
behaviors in which a single event is related to multiple ones.
Such behaviors occur, e.g., in widespread programming constructs
such as exceptions, interrupts and continuations.

OPTL~\cite{ChiariMP20a} is a recent temporal logic that can express specifications concerning,
e.g., whether a function is terminated by an exception, or throws one, and pre/post-conditions.
OPTL is based on Operator Precedence Languages (OPLs),
which are a subclass of deterministic CFLs initially introduced
with the purpose of efficient parsing~\cite{Floyd1963},
a field in which they continue to offer useful applications~\cite{GruneJacobs:08,BarenghiEtAl2015}.
They strictly contain VPLs~\cite{CrespiMandrioli12} and retain all complexity, closure
and decidability properties that make regular languages and VPLs well-suited for temporal logics:
they are closed under Boolean operations, concatenation, Kleene~*,
and language emptiness and inclusion are decidable~\cite{CrespiMandrioli12}.
They have also been characterized through Monadic Second-Order Logic (MSO)~\cite{LonatiEtAl2015}.
This greater expressive power results in OPTL being based on a structure called
\emph{OP word}, made of a linear order plus a nesting relation that can be
one-to-many and many-to-one, and considerably generalizes that of Nested Words.
Indeed, OPTL is strictly more expressive than NWTL, but retains the
same satisfiability and model checking complexity~\cite{ChiariMP20a}.
However, the relationship between OPTL and FOL remained unknown.

In this paper, we close this gap thanks to the novel temporal logic POTL
(Precedence-Oriented Temporal Logic).
POTL too is based on OPLs and is devised to easily navigate a word's underlying syntax tree (ST).
It is also based on OP words, enabling reasoning
on constructs such as exceptions, but has a better ability to
express properties on the context-free structure of words.
It features new until and since modalities based on \emph{summary} paths,
which can navigate a word's ST up or down while skipping
subtrees, and \emph{hierarchical} paths, which move between word
positions where the nesting relation is many-to-one or one-to-many.

The gap in expressive power between OPTL and POTL is both practical and theoretical.
In fact, in POTL it is easier to express stack inspection properties
in the presence of uncaught exceptions, as well as function-frame local properties.
We show that one of such properties is not expressible at all in OPTL,
but it is in POTL, and hence POTL is strictly more expressive than OPTL\@.

Moreover, we prove that POTL is equivalent to FOL on both OP finite and $\omega$-words.
We obtain the result on finite words by translating an FO-complete
logic for finite trees~\cite{Marx2004} to POTL\@.
We then extend the proof to $\omega$-words by employing a composition argument on trees
proved with Ehrenfeucht-Fra\"{\i}ss\'e games.
This last proof is rather involved, as it needs to distinguish
between two classes of trees that arise from OP $\omega$-words.

As a corollary, we show that FOL has the three-variable property on OP words.
I.e., every FO formula on OP words is equivalent to another
FOL formula in which only three distinct variables appear.
This property is related to FO-completeness in temporal
logics~\cite{Gabbay81} and holds both on simple Dedekind-complete
linear orders and Nested Words~\cite{lmcs/AlurABEIL08}.

Recently, we also showed~\cite{MPC20,MPC20-arXiv} that FO-definable OPLs coincide with aperiodic or noncounting ones,
according to a definition of aperiodic structured CFLs which naturally extends the classic one for regular languages~\cite{McNaughtonPapert71}%
\footnote{Noticeably, this equivalence does not hold for tree-languages and their FO-logic~\cite{DBLP:conf/caap/Thomas84,DBLP:journals/ita/Heuter91}.}.
Whereas various regular languages of practical usage are counting ---e.g., many hardware devices are just counters modulo some integer--- we are not aware of programming languages or other CFLs of practical interest, such as structured data description languages, that exhibit counting properties:
e.g., no language imposes to write programs with an even number of nested loops. Thus, the breadth of coverage in terms of practical applications of model-checking algorithms for an FO-complete logic is even larger for the class of CFLs than that of classic model checkers for regular languages.

POTL's expressiveness gains do not come at the cost of higher
model checking complexity, which remains EXPTIME-complete, as we show in~\cite{ChiariMP21}.

The paper is organized as follows: Section~\ref{sec:opl} provides some background on OPLs;
Section~\ref{sec:potl} presents the syntax and semantics of POTL,
with some qualitative demonstration of its expressive power and comparison with similar logics,
and the proof of the expressive \emph{in}completeness of OPTL\@;
Section~\ref{sec:fo-completeness} proves the equivalence of POTL to FOL on finite words;
Section~\ref{sec:omega-fo-completeness} extends this result to $\omega$-words;
Section~\ref{sec:mc} summarizes the results of the companion paper~\cite{ChiariMP21} where a model checker for OPLs based on POTL is provided, its complexity evaluated, and experimental results are also reported;
Section~\ref{sec:conclusion} concludes and proposes future research lines.
Appendices~\ref{sec:chain-prop-proofs} and~\ref{sec:expansion-proofs} contain proofs
that did not fit into the main text.

We assume familiarity with classic topics of theoretical computer science, specifically context-free grammars and languages, pushdown automata, parsing, syntax tree (ST) (cf.~\cite{Harrison78}), $\omega$-languages, i.e., languages with infinite words, and temporal logic ---LTL in particular---~\cite{Emerson90}.

\section{Operator Precedence Languages}%
\label{sec:opl}

Operator Precedence Languages (OPLs) were originally defined through their generating grammars~\cite{Floyd1963}: \emph{operator precedence grammars} (OPGs) are a special class of context-free grammars (CFGs) in \emph{operator normal form} ---i.e., grammars in which right-hand sides (rhs) of production rules contain no consecutive non-terminals\footnote{Every CFG can be effectively transformed into an equivalent one in operator form~\cite{Harrison78}.}---. As a consequence, the syntax trees generated by such grammars never exhibit two consecutive internal nodes.

The distinguishing feature of OPGs is that they define three \emph{precedence relations} (PR) between pairs of input symbols which drive the deterministic parsing and therefore the construction of a unique ST, if any, associated with an input string. For this reason we consider OPLs a kind of input-driven languages, but larger then the original VPLs.
The three PRs are denoted by the symbols $\lessdot, \doteq, \gtrdot$ and are respectively named
\emph{yields precedence, equal in precedence}, and \emph{takes precedence}.
They graphically resemble the traditional arithmetic relations but do not share their typical ordering and equivalence properties;  we kept them for ``historical reasons'', but we recommend the reader not to be confused by the similarity.

Intuitively, given two input characters $a, b$ belonging to a grammar's \emph{terminal alphabet} separated by at most one non-terminal,
$a \lessdot b$ iff in some grammar derivation $b$ is the first terminal character of a grammar's rhs following $a$ whether the grammar's rule contains a non-terminal character before $b$ or not
(for this reasons we also say that non-terminal characters are ``transparent'' in OPL parsing);
 $a \doteq b$ iff $a$ and $b$ occur consecutively in some rhs, possibly separated by one non-terminal;
 $a \gtrdot b$ iff $a$ is the last terminal in a rhs ---whether followed or not by a non-terminal---, and $b$ follows that rhs in some derivation.
The following example provides a first intuition of how a set of \emph{unique} PRs drives the parsing of a string of terminal characters in a deterministic way; subsequently the above concepts are formalized.

\begin{exa}
Consider the alphabet of terminal symbols $\Sigma = \{ \lcall, \lret, \lhandle, \lthrow   \}$: as the chosen identifiers suggest, $\lcall$ represents the fact that a procedure call occurs, $\lret$ represents the fact that a procedure terminates normally and returns to its caller, $\lthrow$  that an exception is raised and $\lhandle$ that an exception handler is installed.
We want to implement a policy such that an exception aborts all the pending calls up to the point where an appropriate handler is found in the stack, \emph{if any}; after that, execution is resumed normally.
Calls and returns, as well as possible pairing of handlers and exceptions are managed according to the usual LIFO policy. The alphabet symbols are written in boldface for reasons that will be explained later but are irrelevant for this example.

The above policy is implemented by the PRs described in Figure~\ref{fig:opm-mcall} which displays the PRs through a square matrix, called \emph{operator precedence matrix (OPM)}, where the element of row $i$ and column $j$ is the PR between the symbol labeling row $i$ and that of column $j$.
We also add the special symbol $\#$ which is used as a string delimiter and state the convention that all symbols of $\Sigma$ yield precedence to, and take precedence over it.

\begin{figure}[tb]
\[
\begin{array}{r | c c c c} 
         & \lcall   & \lret   & \lhandle & \lthrow \\
\hline 
\lcall   & \lessdot & \doteq  & \lessdot & \gtrdot \\
\lret    & \gtrdot  & \gtrdot & \gtrdot  & \gtrdot \\
\lhandle & \lessdot & \gtrdot & \lessdot & \doteq \\
\lthrow  & \gtrdot  & \gtrdot & \gtrdot  & \gtrdot \\
\end{array}
\]
\caption{The OPM $M_\lcall$.}%
\label{fig:opm-mcall}
\end{figure}

Let us now see how the OPM of Figure~\ref{fig:opm-mcall}, named $M_\lcall$, drives the construction of a unique ST associated to a string on the alphabet $\Sigma$ through a typical bottom-up parsing algorithm.
We will see that the shape of the obtained ST depends only on the OPM and not on the particular grammar exhibiting the OPM\@.
Consider the sample word
$w_\mathit{ex} =
 \lcall \ \allowbreak
 \lhandle \ \allowbreak
 \lcall \ \allowbreak
 \lcall \ \allowbreak
 \lcall \ \allowbreak
 \lthrow \ \allowbreak
 \lcall \
 \lret \
 \lcall \
 \lret \
 \lret$.
First, add the delimiter $\#$ at its boundaries and write all precedence relations between consecutive characters,
according to $M_{\lcall}$. The result is row $0$ of Figure~\ref{fig:parsing}.

Then, select all innermost patterns of the form
$a \lessdot c_1 \doteq \dots \doteq c_\ell \gtrdot b$.
In row $0$ of Figure~\ref{fig:parsing} the only such pattern is the underscored $\lcall$ enclosed within the pair ($\mathord{\lessdot}, \mathord{\gtrdot}$).
This means that the ST we are going to build, if it exists, must contain an internal node with the terminal character $\lcall$ as its only child.
We mark this fact by replacing the pattern $\lessdot \underline{\lcall} \gtrdot$ with a dummy non-terminal character, say $N$ ---i.e., we \emph{reduce} $\underline{\lcall}$ to $N$---.
The result is row $1$ of Figure~\ref{fig:parsing}.

Next, we apply the same labeling to row $1$ by simply ignoring the presence of the dummy symbol $N$ and we find a new candidate for reduction, namely the pattern $\lessdot \underline{\lcall} \ N \gtrdot$.
Notice that there is no doubt on building the candidate rhs as $\lessdot \underline{\lcall} \ N \gtrdot$:
if we reduced just the $\underline{\lcall}$ and replaced it by a new $N$, we would produce two adjacent internal nodes, which is impossible since the ST must be generated by a grammar in operator normal form.

By skipping the obvious reduction of row $2$, we come to row $3$. This time the terminal characters to be reduced, again, underscored, are two, with an $\doteq$ and an $N$ in between.
This means that they embrace a subtree of the whole ST whose root is the node represented by the dummy symbol $N$.
By executing the new reduction leading from row $3$ to $4$ we produce a new $N$ immediately to the left of a $\lcall$ which is matched by an equal in precedence $\lret$.
Then, the procedure is repeated until the final row $7$ is obtained, where, by convention we state the $\doteq$ relation between the two delimiters.

Given that each reduction applied in Figure~\ref{fig:parsing} corresponds to a derivation step of a grammar and to the expansion of an internal node of the corresponding ST, it is immediate to realize that the ST of $w_\mathit{ex}$ is the one depicted in Figure~\ref{fig:sttree-example}, where the terminal symbols have been numbered according to their occurrence ---including the conventional numbering of the delimiters--- for future convenience, and labeling internal nodes has been omitted as useless.

\begin{figure}
\[
\begin{array}{l | l} 
0 & \# \lessdot \lcall \lessdot \lhandle \lessdot \lcall \lessdot \lcall \lessdot \underline{\lcall} \gtrdot \lthrow \gtrdot \lcall \doteq \lret \gtrdot \lcall \doteq \lret \gtrdot \lret \gtrdot \# \\
1 & \# \lessdot \lcall \lessdot \lhandle \lessdot \lcall \lessdot \underline{\lcall} \ N \gtrdot \lthrow \gtrdot \lcall \doteq \lret \gtrdot \lcall \doteq \lret \gtrdot \lret \gtrdot \# \\
2 & \# \lessdot \lcall \lessdot \lhandle \lessdot \underline{\lcall} \ N \gtrdot \lthrow \gtrdot \lcall \doteq \lret \gtrdot \lcall \doteq \lret \gtrdot \lret \gtrdot \# \\
3 & \# \lessdot \lcall \lessdot \underline{\lhandle} \doteq N \ \underline{\lthrow} \gtrdot \lcall \doteq \lret \gtrdot \lcall \doteq \lret \gtrdot \lret \gtrdot \# \\
4 & \# \lessdot \lcall \lessdot N \ \underline{\lcall} \doteq  \underline{\lret} \gtrdot \lcall \doteq \lret \gtrdot \lret \gtrdot \# \\
5 & \# \lessdot \lcall \lessdot N \ \underline{\lcall} \doteq  \underline{\lret} \gtrdot \lret \gtrdot \# \\
6 & \# \lessdot \underline{\lcall} \doteq N \ \underline{\lret} \gtrdot \# \\
7 & \# \doteq N \# \\
\end{array}
\]
\caption{The sequence of bottom-up reductions during the parsing of $w_\mathit{ex}$.}%
\label{fig:parsing}
\end{figure}

\begin{figure}
\includegraphics{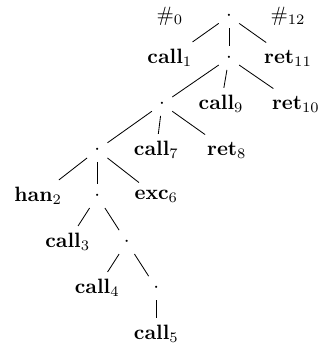}
\caption{The ST corresponding to word $w_\mathit{ex}$.
  Dots represent non-terminals.}%
\label{fig:sttree-example}
\end{figure}
\end{exa}

\begin{rems}
The tree of Figure~\ref{fig:sttree-example} emphasizes the main difference between various types of parenthesis-like languages, such as VPLs, and OPLs: whereas in the former ones every open parenthesis is consumed by the only corresponding closed one\footnote{To be precise, VPLs allow for unmatched closed parentheses but only at the beginning of a string and unmatched open ones at the end.}, in our example a $\lcall$ can be matched by the appropriate $\lret$ but can also be ``aborted'' by an $\lthrow$ which in turn aborts all pending $\lcall$s until its corresponding $\lhandle$ ---\emph{if any}--- is found.
\end{rems}

Thus, an OPM defines a \emph{universe} of strings on the given alphabet that can be parsed according to it and assigns a unique ST ---with unlabeled internal nodes--- to each one of them.
Such a universe is the whole $\Sigma^*$ iff \emph{the OPM is complete}, i.e.\ it has no empty cells, including those of the implicit row and column referring to the delimiters.
In the early literature about OPLs, e.g.,~\cite{Floyd1963,Crespi-ReghizziMM1978} OPGs sharing a given OPM were used to define restricted languages w.r.t.\ the universe defined by the OPM and their algebraic properties have been investigated.
Later on the same operation has been defined by using different formalisms such as pushdown automata, monadic second order logic, and suitable extensions of regular expressions. In this paper we refer to the use automata and temporal logic, which are typical of model checking.
As a side remark we mention that, in general, it may happen that in the same string there are several patterns ready to be reduced;
this could enable the implementation of parallel parsing algorithms (see e.g.,~\cite{BarenghiEtAl2015}) which however is not an issue of interest in this paper.

We now state the basics of OPLs needed for this paper in a formal way.
Let $\Sigma$ be a finite alphabet, and $\varepsilon$ the empty string.
We use the special symbol $\# \not\in \Sigma$ to mark the beginning and
the end of any string.

\begin{defi}\label{def:opm}
  An \textit{operator precedence matrix} (OPM) $M$ over $\Sigma$ is a partial function
  $(\Sigma \cup \{\#\})^2 \to \{\lessdot, \allowbreak \doteq, \allowbreak \gtrdot\}$,
  that, for each ordered pair $(a,b)$, defines the \emph{precedence relation} $M(a,b)$
  holding between $a$ and $b$. If the function is total we say that $M$ is \emph{complete}.
  We call the pair $(\Sigma, M)$ an \emph{operator precedence alphabet}.
  By convention, the initial $\#$ can only yield precedence to other symbols, and other
  symbols can only take precedence on the ending $\#$.

  If $M(a,b) = \prf$, where $\prf \in \{\lessdot, \doteq, \gtrdot\}$,
  we write $a \pr b$.  For $u,v \in (\Sigma \cup \{\#\})^+$ we write $u \pr v$ if
  $u = xa$ and $v = by$ with $a \pr b$.
\end{defi}

The next concept of \emph{chain} makes the connection between OP relations and
ST structure explicit, through brackets.

\begin{defi}%
\label{def:chain}
A \emph{simple chain}
$
\ochain {c_0} {c_1 c_2 \dots c_\ell} {c_{\ell+1}}
$,
with $\ell \geq 1$,
is a string $c_0 c_1 c_2 \dots c_\ell c_{\ell+1}$,
such that:
$c_0, \allowbreak c_{\ell+1} \in \Sigma \cup \{\#\}$,
$c_i \in \Sigma$ for every $i = 1,2, \dots \ell$,
and $c_0 \lessdot c_1 \doteq c_2 \dots c_{\ell-1} \doteq c_\ell
\gtrdot c_{\ell+1}$.

A \emph{composed chain} is a string
$c_0 s_0 c_1 s_1 c_2  \dots c_\ell s_\ell c_{\ell+1}$,
where
$\ochain {c_0}{c_1 c_2 \dots c_\ell}{c_{\ell+1}}$ is a simple chain, and
$s_i \in \Sigma^*$ is either the empty string
or is such that $\ochain {c_i} {s_i} {c_{i+1}}$ is a chain (simple or composed),
for every $i = 0,1, \dots, \ell$ ($\ell \geq 1$).
Such a composed chain will be written as
$\ochain {c_0} {s_0 c_1 s_1 c_2 \dots c_\ell s_\ell} {c_{\ell+1}}$.

In a chain, simple or composed, $c_0$ (resp.\ $c_{\ell+1}$) is called its \emph{left} (resp.\ \emph{right}) \emph{context};
all terminals between them are called its \emph{body}.

A finite word $w$ over $\Sigma$ is \emph{compatible} with an OPM $M$ iff
for each pair of letters $c, d$, consecutive in $w$, $M(c, d)$ is defined and,
for each substring $x$ of $\# w \#$ which is a chain of the form $^a[y]^b$,
$M(a, b)$ is defined.
For a given operator precedence alphabet $(\Sigma, M)$ the set of all words  compatible with $M$ is called the \emph{universe} of the operator precedence alphabet.
\end{defi}

The chain below is the chain defined by the OPM $M_\lcall$ of
Figure~\ref{fig:opm-mcall} for the word $w_\mathit{ex}$. It shows the natural isomorphism between STs with unlabeled internal nodes (see Figure~\ref{fig:sttree-example}) and chains.
\[
\# [ \lcall [ [ [ \lhandle
[ \lcall [ \lcall [ \lcall ] ] ]
\lthrow ] \lcall \; \lret ] \lcall \; \lret ] \lret ] \#
\]

Note that, in composed chains, consecutive inner chains are always separated by an input symbol: this is due to the fact that OPL strings are generated by grammars in operator normal form.

Next we introduce operator precedence automata as pushdown machines suitable to carve specific OPLs within the universe of an operator precedence alphabet.

\begin{defi}\label{def:OPA}
An \emph{operator precedence automaton (OPA)} is a tuple
$\mathcal A = (\Sigma, \allowbreak M, \allowbreak Q, \allowbreak I, \allowbreak F, \allowbreak \delta) $ where:
$(\Sigma, M)$ is an operator precedence alphabet,
$Q$ is a finite set of states (disjoint from $\Sigma$),
$I \subseteq Q$ is the set of initial states,
$F \subseteq Q$ is the set of final states,
$\delta \subseteq Q \times (\Sigma \cup Q) \times Q$ is the transition relation,
which is the union of the three disjoint relations
$\delta_{\mathit{shift}}\subseteq Q \times \Sigma \times Q$,
$\delta_{\mathit{push}}\subseteq Q \times \Sigma \times Q$,
and
$\delta_{\mathit{pop}}\subseteq Q \times Q \times Q$.
An OPA is deterministic iff $I$ is a singleton,
and all three components of $\delta$ are ---possibly partial--- functions.
\end{defi}

To define the semantics of OPA, we need some new notations.
Letters $p, q, p_i, \allowbreak q_i, \dots$ denote states in $Q$.
We use
$q_0 \va{a}{q_1}$ for $(q_0, a, q_1) \in \delta_{\mathit{push}}$,
$q_0 \vshift{a}{q_1}$ for $(q_0, a, q_1) \in \delta_{\mathit{shift}}$,
$q_0 \flush{q_2}{q_1}$  for $(q_0, q_2, q_1) \in  \delta_{\mathit{pop}}$,
and ${q_0} \ourpath{w} {q_1}$, if the automaton can read $w \in \Sigma^*$ going from $q_0$ to $q_1$.
Let $\Gamma$ be $\Sigma \times Q$ and $\Gamma' = \Gamma \cup \{\bot\}$
be the \textit{stack alphabet};
we denote symbols in $\Gamma$ as $\tstack aq$.
We set $\symb {\tstack aq} = a$, $\symb {\bot}=\#$, and
$\state {\tstack aq} = q$.
For a stack content $\gamma = \gamma_n \dots \gamma_1 \bot$,
with $\gamma_i \in \Gamma$, $n \geq 0$,
we set $\symb \gamma = \symb{\gamma_n}$ if $n \geq 1$, and $\symb \gamma = \#$ if $n = 0$.

A \emph{configuration} of an OPA is a triple $c = \tconfig w q \gamma$,
where $w \in \Sigma^*\#$, $q \in Q$, and $\gamma \in \Gamma^*\bot$.
A \emph{computation} or \emph{run} is a finite sequence
$c_0 \transition{} c_1 \transition{} \dots \transition{} c_n$
of \emph{moves} or \emph{transitions}
$c_i \transition{} c_{i+1}$.
There are three kinds of moves, depending on the PR between the symbol
on top of the stack and the next input symbol:

\begin{description}
\item[push move] if $\symb \gamma \lessdot a$ then
$\tconfig {ax} p \gamma \transition{} \tconfig {x} q {\tstack a p \gamma}$,
with $(p,a, q) \in \delta_{\mathit{push}}$;

\item[shift move] if $a \doteq b$ then
$\tconfig {bx} q {\tstack a p \gamma} \transition{} \tconfig x r {\tstack b p \gamma}$,
with $(q,b,r) \in \delta_{\mathit{shift}}$;

\item[pop move] if $a \gtrdot b$ then
$\tconfig {bx} q {\tstack a p \gamma} \transition{} \tconfig {bx} r \gamma$,
with $(q, p, r) \in \delta_{\mathit{pop}}$.
\end{description}

Shift and pop moves are not performed when the stack contains only $\bot$.
Push moves put a new element on top of the stack consisting of the input symbol together with the current state of the OPA\@.
Shift moves update the top element of the stack by \textit{changing its input symbol only}.
Pop moves remove the element on top of the stack,
and update the state of the OPA according to $\delta_{\mathit{pop}}$ on the basis of the current state and the state in the removed stack symbol.
They do not consume the input symbol, which is used only as a look-ahead to establish the $\gtrdot$ relation.
The OPA accepts the language
\(
L(\mathcal A) = \left\{ x \in \Sigma^* \mid  \tconfig {x\#} {q_I} {\bot} \vdash ^*
\tconfig {\#} {q_F}{\bot} , \allowbreak q_I \in I, \allowbreak q_F \in F \right\}.
\)

\begin{defi}
Let $\mathcal A$ be an OPA\@.
We call a \emph{support} for the simple chain
$\ochain {c_0} {c_1 c_2 \dots c_\ell} {c_{\ell+1}}$
any path in $\mathcal A$ of the form
$q_0
\va{c_1}{q_1}
\vshift{}{}
\dots
\vshift{}q_{\ell-1}
\vshift{c_{\ell}}{q_\ell}
\flush{q_0} {q_{\ell+1}}$.
The label of the last (and only) pop is exactly $q_0$, i.e.\ the first state of the path;
this pop is executed because of relation $c_\ell \gtrdot c_{\ell+1}$.

We call a \emph{support for the composed chain}
$\ochain {c_0} {s_0 c_1 s_1 c_2 \dots c_\ell s_\ell} {c_{\ell+1}}$
any path in $\mathcal A$ of the form
\(
q_0
\ourpath{s_0}{q'_0}
\va{c_1}{q_1}
\ourpath{s_1}{q'_1}
\vshift{c_2}{}
\dots
\vshift{c_\ell} {q_\ell}
\ourpath{s_\ell}{q'_\ell}
\flush{q'_0}{q_{\ell+1}}
\)
where, for every $i = 0, 1, \dots, \ell$:
if $s_i \neq \varepsilon$, then $q_i \ourpath{s_i}{q'_i} $
is a support for the chain $\ochain {c_i} {s_i} {c_{i+1}}$, else $q'_i = q_i$.
\end{defi}

Chains fully determine the parsing structure of any
OPA over $(\Sigma, M)$. If the OPA performs the computation
$
\langle sb, q_i, [a, q_j] \gamma \rangle \vdash^*
\langle b,  q_k, \gamma \rangle
$,
then $\ochain asb$
is necessarily a chain over $(\Sigma, \allowbreak M)$, and there exists a support
like the one above with $s = s_0 c_1 \dots c_\ell s_\ell$ and $q_{\ell+1} = q_k$.
This corresponds to the parsing of the string $s_0 c_1 \dots c_\ell s_{\ell}$ within the
context $a, b$, which contains all
information needed to build the subtree whose frontier is that string.

Consider the OPA
$\mathcal{A}(\Sigma, M) = (\Sigma, M, \{q\}, \{q\}, \{q\}, \delta_{max})$
where $\delta_{max}(q,q) = q$, and $\delta_{max}(q,c) = q$, $\forall c \in \Sigma$.
We call it the \emph{OP Max-Automaton} over $(\Sigma, M)$.
For a max-automaton, each chain has a support;
thus, a max-automaton accepts exactly the universe of the operator precedence alphabet.
If $M$ is complete, the language accepted by $\mathcal A(\Sigma, M)$ is $\Sigma^*$.
With reference to the OPM $M_{\lcall}$ of Figure~\ref{fig:opm-mcall}, the string
$
\lret \
\lcall \
\lhandle
$ is accepted by the max-automaton with structure defined by the chain
$
\#[[\lret]
\lcall
[\lhandle]]
\#.
$

In conclusion, given an OP alphabet, the OPM $M$ assigns a unique structure
to any compatible string in $\Sigma^*$;
unlike VPLs, such a structure is not visible in the string,
and must be built by means of a non-trivial parsing algorithm.
An OPA defined on the OP alphabet selects an appropriate subset within the
universe of the OP alphabet.
OPAs form a Boolean algebra whose universal element is the max-automaton.
The language classes recognized by deterministic and non-deterministic OPAs coincide.
For a more complete description of the OPL family and of its relations with other CFLs
we refer the reader to~\cite{MP18}.

\begin{figure}[tb]
\centering
\begin{tabular}{p{0.1\textwidth} p{0.2\textwidth} p{0.2\textwidth} p{0.2\textwidth} p{0.1\textwidth}}
&
\begin{minipage}[t]{0.2\textwidth}
\footnotesize
\begin{verbatim}
    pA() {
A0:   try {
A1:     pB();
A2:   } catch {
A3:     pErr();
A4:     pErr();
      }
Ar: }
\end{verbatim}
\end{minipage}
&
\begin{minipage}[t]{0.2\textwidth}
\footnotesize
\begin{verbatim}
    pB() {
B0:   pC();
Br: }
\end{verbatim}
\end{minipage}
&
\begin{minipage}[t]{0.2\textwidth}
\footnotesize
\begin{verbatim}
    pC() {
C0:   if (*) {
C1:     throw;
C2:   } else {
C3:     pC();
      }
Cr: }
\end{verbatim}
\end{minipage}
&
\\
\multicolumn{5}{c}{%
\begin{tikzpicture}
  [node distance=40pt, every state/.style={minimum size=0pt, inner sep=2pt}, >=latex, font=\footnotesize]
\node[state, initial by arrow, initial text=] (m0) {M0};
\node[state] (a0) [right of=m0] {A0};
\node[state] (a1) [right of=a0] {A1};
\node[state] (b0) [right of=a1] {B0};
\node[state] (c0) [right of=b0] {C0};
\node[state] (a2) [below of=c0] {A2};
\node[state] (a3) [left of=a2] {A3};
\node[state] (er) [left of=a3] {Er};
\node[state] (ar) [left of=er] {Ar};
\node[state] (arp) [left of=ar] {Ar'};
\node[state] (mr) [accepting, left of=arp] {Mr};
\node[state] (a4) [below of=er] {A4};
\path[->] (m0) edge[above] node[label=below:$\mathrm{p}_A$] {$\lcall$} (a0)
          (a0) edge[above] node[label=below:\texttt{try}] {$\lhandle$} (a1)
          (a1) edge[above] node[label=below:$\mathrm{p}_B$] {$\lcall$} (b0)
          (b0) edge[above] node[label=below:$\mathrm{p}_C$] {$\lcall$} (c0)
          (c0) edge[out=30, in=0, loop, right] node {$\lcall \, \mathrm{p}_C$} (c0)
          (c0) edge[out=-20, in=-50, loop, double, right] node {A1, B0, C0} (c0)
          (c0) edge[dashed, left] node {$\lthrow$} (a2)
          (a2) edge[double, above] node {A0} (a3)
          (a3) edge[above] node[label=below:$\mathrm{p}_{\mathit{Err}}$] {$\lcall$} (er)
          (er) edge[loop above, dashed, right] node {$\lret \, \mathrm{p}_{\mathit{Err}}$} (er)
          (er) edge[bend left, double, right] node {A3} (a4)
          (a4) edge[bend left, left] node[label=below:$\mathrm{p}_{\mathit{Err}}$] {$\lcall$} (er)
          (er) edge[double, above] node {A4} (ar)
          (ar) edge[dashed, above] node[label=below:$\mathrm{p}_{A}$] {$\lret$} (arp)
          (arp) edge[double, above] node {M0} (mr);
\end{tikzpicture}
}
\end{tabular}
  \caption{Example procedural program (top) and the derived OPA (bottom).
    Push, shift, pop moves are shown by, resp., solid, dashed and double arrows.}%
  \label{fig:example-prog}
\end{figure}

\begin{exa}%
\label{exa:opa-example}
For readers not familiar with OPLs, we show how OPAs can naturally model
programming languages such as Java and C++.
Given a set $AP$ of
atomic propositions describing events and states of the program, we
use $(\powset{AP}, M_{AP})$ as the OP alphabet.  For
convenience, we consider a partitioning of $AP$ into a set of normal
propositional labels (in round font), and \emph{structural labels} (in bold).
Structural labels define the OP structure of the word:
$M_{AP}$ is only defined for subsets of $AP$ containing exactly
one structural label, so that given two structural labels $\mathbf{l}_1, \mathbf{l}_2$, for any $a,
a', b, b' \in \powset{AP}$ s.t.\ $\mathbf{l}_1 \in a, a'$ and
$\mathbf{l}_2 \in b, b'$ we have $M_{AP}(a,b) =
M_{AP}(a',b')$.  In this way, it is possible to define an OPM on
the entire $\powset{AP}$ by only giving the relations between structural labels, as
we did for $M_\lcall$.
Figure~\ref{fig:example-prog} shows how to model a procedural program with
 an OPA\@.  The OPA simulates the program's behavior with respect to the stack, by expressing its
execution traces with four event kinds: $\lcall$ (resp.\ $\lret$)
marks a procedure call (resp.\ return), $\lhandle$ the
installation of an exception handler by a \texttt{try} statement, and
$\lthrow$ an exception being raised. OPM $M_\lcall$ defines the
context-free structure of the word, which is strictly linked with the
programming language semantics: the $\lessdot$ PR causes
nesting (e.g., $\lcall$s can be nested into other $\lcall$s), and the
$\doteq$ PR implies a one-to-one relation, e.g.\ between a $\lcall$
and the $\lret$ of the same function, and a $\lhandle$ and the
$\lthrow$ it catches. Each OPA state represents a line in
the source code. First, procedure $\mathrm{p}_A$ is called by the
program loader (M0), and $[\{\lcall, \mathrm{p}_A\}, \text{M0}]$ is pushed onto the stack, to track the
program state before the $\lcall$. Then, the \texttt{try} statement at
line \texttt{A0} of $\mathrm{p}_A$ installs a handler. All subsequent
calls to $\mathrm{p}_B$ and $\mathrm{p}_C$ push new stack symbols on
top of the one pushed with $\lhandle$. $\mathrm{p}_C$ may only call
itself recursively, or throw an exception, but never return
normally. This is reflected by $\lthrow$ being the only transition
leading from state C0 to the accepting state Mr, and $\mathrm{p}_B$
and $\mathrm{p}_C$ having no way to a normal $\lret$. The OPA has a
look-ahead of one input symbol, so when it encounters $\lthrow$, it
must pop all symbols in the stack, corresponding to active function
frames, until it finds the one with $\lhandle$ in it, which cannot be
popped because $\lhandle \doteq \lthrow$. Notice that such behavior
cannot be modeled by Visibly Pushdown Automata or Nested Word
Automata, because they need to read an input symbol for each pop
move. Thus, $\lhandle$ protects the parent function from the
exception. Since the state contained in $\lhandle$'s stack symbol is
A0, the execution resumes in the \texttt{catch} clause of
$\mathrm{p}_A$. $\mathrm{p}_A$ then calls twice the library error-handling
function $\mathrm{p}_{\mathit{Err}}$, which ends regularly both times, and returns.
The string of Figure~\ref{fig:opm-mcall} is accepted by this OPA\@.

In this example, we only model the stack behavior for simplicity, but
other statements, such as assignments, and other behaviors, such as
continuations, could be modeled by a different choice of the OPM, and
other aspects of the program's state by appropriate
abstractions~\cite{JhalaPR18}.
\end{exa}

\subsection{Operator Precedence
  \texorpdfstring{$\omega$-Languages}{Omega-Languages}}%
\label{sec:background-omega}

All definitions regarding OPLs are extended to infinite words in
the usual way, but with a few distinctions~\cite{LonatiEtAl2015}.  Given an
alphabet $(\Sigma, M)$, an $\omega$-word $w \in \Sigma^\omega$ is
compatible with $M$ if every prefix of $w$ is compatible with $M$.
OP $\omega$-words are not terminated by the delimiter $\#$.
An $\omega$-word may contain never-ending chains of the form $c_0
\lessdot c_1 \doteq c_2 \doteq \cdots$, where the $\lessdot$ relation
between $c_0$ and $c_1$ is never closed by a corresponding
$\gtrdot$. Such chains are called \emph{open chains} and may be simple
or composed. A composed open chain may contain both open and closed
subchains. Of course, a closed chain cannot contain an open one.
A terminal symbol $a \in \Sigma$ is \emph{pending} if it
is part of the body of an open chain and of no closed chains.

OPA classes accepting the whole class of $\omega$OPLs can be defined by augmenting
Definition~\ref{def:OPA} with B\"uchi or Muller acceptance conditions.  In this paper, we only consider the former one.
The semantics of configurations, moves and infinite
runs are defined as for finite OPAs.  For the acceptance condition,
let $\rho$ be a run on an $\omega$-word $w$. Define
\(
\operatorname{Inf}(\rho) = \{q \in Q \mid \text{there exist infinitely
many positions $i$ s.t.\ } \langle \beta_i, q, x_i \rangle \in \rho \}
\)
as the set of states that occur infinitely often in $\rho$.
$\rho$ is successful iff there exists a state $q_f \in F$ such that
$q_f \in \operatorname{Inf}(\rho)$. An $\omega$OPBA $\mathcal{A}$ accepts
$w \in \Sigma^\omega$ iff there is a successful run of $\mathcal{A}$ on $w$.
The $\omega$-language recognized by $\mathcal{A}$ is
$L(\mathcal{A}) = \{w \in \Sigma^\omega \mid \mathcal{A} \text{ accepts } w \}$.
Unlike OPAs, $\omega$OPBAs do not require the stack to be empty for word acceptance:
when reading an open chain, the stack symbol pushed when the first character
of the body of its underlying simple chain is read remains into the stack forever;
it is at most updated by shift moves.

The most important closure properties of OPLs are preserved by
$\omega$OPLs, which form a Boolean algebra and are closed under
concatenation of an OPL with an $\omega$OPL~\cite{LonatiEtAl2015}.
The equivalence between deterministic and nondeterministic automata is
lost in the infinite case, which is unsurprising, since it
also happens for regular $\omega$-languages and $\omega$VPLs.

A more complete treatment of OPLs properties and parsing algorithms can be found in~\cite{MP18, GruneJacobs:08}; $\omega$OPLs are described in some depth in~\cite{LonatiEtAl2015}.

\subsection{OPLs vs other structured language families}%
\label{sec:structured-langs}

The nice closure properties of OPLs come from the fact that a word's syntactic structure
is determined locally, once an OPM is given.
Other language classes enjoy similar properties.
The simplest (and earliest) ones are \emph{Parenthesis Languages}~\cite{McNaughton67}.
In Parenthesis Languages, two terminals disjoint from the input alphabet
are used as open and closed parentheses, and they surround all grammar rule rhs.
Thus, the syntactic structure is directly encoded in words
(just like we did in Figure~\ref{fig:opm-mcall} with chains).

Visibly Pushdown Languages (VPLs)~\cite{AluMad04},
first introduced as Input-Driven Languages~\cite{DBLP:conf/icalp/Mehlhorn80},
extend parenthesis languages; they also lead to applications in model checking.
In VPLs, the input alphabet $\Sigma$ is partitioned into three disjoint sets $\Sigma_c$, $\Sigma_i$,
and $\Sigma_r$, called respectively the \emph{call}, \emph{internal}, and \emph{return} alphabets.
\emph{Visibly Pushdown Automata} (VPA), the automata class recognizing VPLs,
always perform a push move when reading a call symbol, a pop move when reading a return symbol,
and, when reading an internal symbol, they perform a move that only updates the current state,
leaving the stack untouched.
Thus, a string's syntactic structure is fully determined by the alphabet partition,
and is clearly \emph{visible}: a symbol in $\Sigma_c$ is an open parenthesis,
and one in $\Sigma_r$ is a closed parenthesis.
The matching between such symbols is unambiguous.
Once the alphabet partition is fixed, VPLs form a Boolean algebra,
which enabled their success in model checking.
VPAs are \emph{real-time}, as they read exactly one input symbol with each move.
This limitation distinguishes them from OPAs, whose pop moves are so-called $\varepsilon$-moves.
OPLs strictly contain VPLs~\cite{CrespiMandrioli12}.

\emph{Nested Words}~\cite{jacm/AlurM09}, an algebraic characterization of VPLs,
were introduced to foster their logic and data-theoretic applications.
They consist of a linear sequence of positions, plus a \emph{matching relation}
encoding the pairing of call and return symbols.
As a result, the matching relation is a strictly one-to-one nesting relation
that never crosses itself (each ``open parenthesis'' is matched to a closed one, with a minor exception).
The class of Regular Languages of Nested Words is recognized by Nested Words Automata,
and is equivalent to VPLs.

The original motivation for VPLs was model checking procedural programs:
the matching between call and return symbols easily models the behavior of a program's stack
during function calls and returns, while internal symbols model other statements.
However, many programming languages manage the stack in more complex ways.
E.g., when an exception is raised, the stack may be unwound,
as multiple procedures are terminated by the exception.
In Example~\ref{exa:opa-example}, this is easily modeled
by an OPM where the $\lcall$ symbol takes precedence from $\lthrow$.
In contrast, a VPA would need to read a different symbol for each pop move,
so a single $\lthrow$ would not suffice.

An early attempt at overcoming such limitations was made with Colored Nested Words~\cite{AlurF16},
in which multiple calls can be matched with a return of a different color,
allowing many-to-one relations.
Colored Nested Words are subsumed by OPLs and, as we show in Section~\ref{sec:potl},
the nesting relation of OPLs can be also one-to-many, besides many-to-one.

\section{Precedence-Oriented Temporal Logic}%
\label{sec:potl}

\begin{figure}[tb]
\begin{tikzpicture}
  [edge/.style={}]
\matrix (m) [matrix of math nodes, column sep=-4, row sep=-4]
{
  \#
  & \lessdot & \lcall
  & \lessdot & \lhandle
  & \lessdot & \lcall
  & \lessdot & \lcall
  & \lessdot & \lcall
  & \gtrdot & \lthrow
  & \gtrdot & \lcall
  & \doteq & \lret
  & \gtrdot & \lcall
  & \doteq & \lret
  & \gtrdot & \lret
  & \gtrdot & \# \\
  & & \mathrm{p}_A
  & &
  & & \mathrm{p}_B
  & & \mathrm{p}_C
  & & \mathrm{p}_C
  & &
  & & \mathrm{p}_{\mathit{Err}}
  & & \mathrm{p}_{\mathit{Err}}
  & & \mathrm{p}_{\mathit{Err}}
  & & \mathrm{p}_{\mathit{Err}}
  & & \mathrm{p}_A
  & & \\
  0
  & & 1
  & & 2
  & & 3
  & & 4
  & & 5
  & & 6
  & & 7
  & & 8
  & & 9
  & & 10
  & & 11
  & & 12 \\
};
\draw[edge] (m-1-1) to [out=20, in=160] (m-1-25);
\draw[edge] (m-1-3) to [out=20, in=160] (m-1-23);
\draw[edge] (m-1-3) to [out=20, in=160] (m-1-15);
\draw[edge] (m-1-3) to [out=20, in=160] (m-1-19);
\draw[edge] (m-1-5) to [out=20, in=160] (m-1-13);
\draw[edge] (m-1-7) to [out=20, in=160] (m-1-13);
\draw[edge] (m-1-9) to [out=20, in=160] (m-1-13);
\end{tikzpicture}
\caption{The example string of Figure~\ref{fig:opm-mcall} as an OP word.
  Chains are highlighted by arrows joining their contexts;
  structural labels are in bold,
  and other atomic propositions are shown below them.
  $\mathrm{p}_l$ means a $\lcall$ or a $\lret$ is related to procedure $\mathrm{p}_l$.
  First, procedure $\mathrm{p}_A$ is called (pos.~1),
  and it installs an exception handler in pos.~2.
  Then, three nested procedures are called,
  and the innermost one ($\mathrm{p}_C$) throws an exception,
  which is caught by the handler.
  Two more functions are called and, finally, $\mathrm{p}_A$ returns.}%
\label{fig:potl-example-word}
\end{figure}

POTL is a linear-time temporal logic, which extends the classical LTL\@.
We recall that the semantics of LTL~\cite{Pnueli77} is defined on a Dedekind-complete set of word positions $U$
equipped with a total ordering and monadic relations, called \emph{atomic propositions} (AP).
In this paper, we consider a discrete timeline, hence
$U = \{0, 1, \dots, n\}$, with $n \in \mathbb{N}$, or $U = \mathbb{N}$.
Each LTL formula $\varphi$ is evaluated in a word position:
we write $(w, i) \models \varphi$ to state that $\varphi$ holds in position $i$ of word $w$.
Besides operators from propositional logic, LTL features modalities that enable
movement between positions; e.g.,
the \emph{Next} modality states that a formula holds in the subsequent position of the current one:
$(w, i) \models \lnext \varphi$ iff $(w, i+1) \models \varphi$;
the \emph{Until} modality states that there exists a \emph{linear path},
made of consecutive positions and starting from the current one, such that a formula $\psi$ holds in the last position of such path,
and another formula $\varphi$ holds in all previous positions.
Formally, $(w, i) \models \lluntil{\varphi}{\psi}$ iff there exists $j \geq i$
s.t.\ $(w, j) \models \psi$, and for all $j'$, with $i \leq j' < j$,
we have $(w, j') \models \varphi$.

The linear order, however, is not sufficient to express properties of more complex structures than the linear ones, typically the tree-shaped ones, which are the natural domain of context-free languages. The history of logic formalisms suitable to deal with CFLs somewhat parallels the path that led from regular languages to tree-languages~\cite{Tha67} or their equivalent counterpart in terms of strings, i.e.\ parenthesis languages~\cite{McNaughton67}.

A first logic mechanism aimed at ``walking through the structure of a context-free sentence'' was proposed in~\cite{Lautemann94} and consists in a \emph{matching condition} that relates the two extreme terminals of the rhs of a context-free grammar in so-called \emph{double Greibach normal form}, i.e.\ a grammar whose production rhs exhibit a terminal character at both ends: in a sense such terminal characters play the role of explicit parentheses.~\cite{Lautemann94} provides a logic language for general CFLs based on such a relation which however fails to extend the decidability properties of logics for regular languages due to lack of closure properties of CFLs. The matching condition was then resumed in~\cite{jacm/AlurM09} to define their MSO logic for VPLs and subsequently the temporal logics  CaRet~\cite{AlurEM04} and NWTL~\cite{lmcs/AlurABEIL08}.

OPLs are structured but not ``visibly structured'' as they lack explicit parentheses (see Section~\ref{sec:opl}). Nevertheless, a more sophisticated notion of matching relation has been introduced in~\cite{LonatiEtAl2015} for OPLs by exploiting the fact that OPLs remain input-driven thanks to the OPM\@. We name the new matching condition \emph{chain relation} and define it here below.
We fix a finite set of atomic propositions $AP$, and an OPM $M_{AP}$ on $\powset{AP}$.

A \emph{word structure}
---also called \emph{OP word} for short--- is the tuple
$\langle U, <, M_{AP}, P \rangle$, where $U$, $<$, and $M_{AP}$ are as above, and
$P \colon AP \to \powset{U}$ is a function associating each atomic proposition
with the set of positions where it holds, with $0, (n+1) \in P(\#)$.
For the time being, we consider just finite string languages; the necessary extensions needed to deal with $\omega$-languages will be introduced in Section~\ref{sec:optl-omega}.

\begin{defi}[Chain relation]
The \emph{chain relation} $\chain(i, j)$ holds between two positions $i,j \in U$ iff $i < j-1$,
and $i$ and $j$ are resp.\ the left and right contexts of the same chain
(cf.\ Definition~\ref{def:chain}), according to $M_{AP}$ and the labeling induced by $P$.
\end{defi}

In the following, given two positions $i, j$ and a PR $\prf$,
we write $i \pr j$ to say $a \pr b$,
where $a = \{\mathrm{p} \mid i \in P(\mathrm{p})\}$, and
$b = \{\mathrm{p} \mid j \in P(\mathrm{p})\}$.
For notational convenience, we partition $AP$ into structural labels, written in bold face,
which define a word's structure, and normal labels, in round face,
defining predicates holding in a position.
Thus, an OPM $M$ can be defined on structural labels only,
and $M_{AP}$ is obtained by inverse homomorphism of $M$ on subsets of $AP$ containing exactly one of them.

The chain relation augments the linear structure of a word with the tree-like structure of OPLs.
Figure~\ref{fig:potl-example-word} shows the word of Figure~\ref{fig:sttree-example} as an OP word and emphasizes the distinguishing feature of the relation, i.e.\ that, for composed chains, it may not be one-to-one, but also one-to-many or many-to-one.
Notice the correspondence between internal nodes in the ST and pairs of positions
in the $\chain$ relation.

In the ST, we say that the right context $j$ of a chain is at the \emph{same level}
as the left one $i$ when $i \doteq j$ (e.g., in Figure~\ref{fig:sttree-example}, pos.\ 1 and 11),
at a \emph{lower level} when $i \lessdot j$ (e.g., pos.\ 1 with 7, and 9),
at a \emph{higher level} if $i \gtrdot j$ (e.g., pos.\ 3 and 4 with 6).

Furthermore, given $i,j \in U$, relation $\chain$ has the following properties:
\begin{enumerate}
\item\label{item:never-cross}
  It never crosses itself: if $\chain(i,j)$ and $\chain(h,k)$, for any $h,k \in U$,
  then we have $i < h < j \implies k \leq j$
  and $i < k < j \implies i \leq h$.
\item\label{item:border-prop}
  If $\chain(i,j)$, then $i \lessdot i+1$ and $j-1 \gtrdot j$.
\item\label{item:downward-prop}
  Consider all positions (if any) $i_1 < i_2 < \dots < i_n$ s.t.\ $\chain(i_p, j)$
  for all $1 \leq p \leq n$.
  We have $i_1 \lessdot j$ or $i_1 \doteq j$ and, if $n > 1$,
  $i_q \gtrdot j$ for all $2 \leq q \leq n$.
\item\label{item:upward-prop}
  Consider all positions (if any) $j_1 < j_2 < \dots < j_n$ s.t.\ $\chain(i, j_p)$
  for all $1 \leq p \leq n$.
  We have $i \gtrdot j_n$ or $i \doteq j_n$ and, if $n > 1$,
  $i \lessdot j_q$ for all $1 \leq q \leq n-1$.
\end{enumerate}
Property~\ref{item:upward-prop} says that when the chain relation is one-to-many,
the contexts of the outermost chain are in the $\doteq$ or $\gtrdot$ relation,
while the inner ones are in the $\lessdot$ relation.
Property~\ref{item:downward-prop} says that contexts of outermost
many-to-one chains are in the $\doteq$ or $\lessdot$ relation,
and the inner ones are in the $\gtrdot$ relation.
Such properties are proved in Appendix~\ref{sec:chain-prop-proofs}
for readers unfamiliar with OPLs.

The $\chi$ relation is the core of the MSO logic characterization for OPLs given in~\cite{LonatiEtAl2015}
where it is also shown that the greater generality of OPLs and corresponding MSO logic, though requiring more technical proofs,
produces results in terms of closure properties, decidability and complexity of the constructions that are the same as the corresponding ones for VPLs. Similarly, in this paper we are going to show that the temporal logic POTL replicates the FO-completeness result of NWTL despite the greater complexity of the $\chi$ relation.%
\footnote{In~\cite{ChiariMP21} we produce model checking algorithms with the same order of complexity as those for NWTL.}

While LTL's linear paths only follow the ordering relation $<$,
paths in POTL may follow the $\chain$ relation too. As a result, a POTL path through a string can simulate paths through the corresponding ST\@.

We envisage two basic types of path. The first one is that of  \emph{summary paths}.
By following the chain relation, summary paths may skip chain bodies,
which correspond to the fringe of a subtree in the ST\@.
We distinguish between \emph{downward} and \emph{upward} summary paths (resp.\ DSP and USP).
Both kinds can follow both the $<$ and the $\chain$ relations; DSPs can enter a chain body but cannot exit it so that they can move only downward in a ST or remain at the same level; conversely, USPs cannot enter one but can move upward by exiting the current one.
In other words, if a position $k$ is part of a DSP, and there are two positions $i$ and $j$,
with $i < k < j$ and $\chain(i,j)$ holds, the next position in the DSP cannot be $\geq j$.
E.g., two of the DSPs starting from pos.~1 in Figure~\ref{fig:potl-example-word}
are 1-2-3, which enters chain $\chain(2, 6)$, and 1-2-6, which skips its body. 
USPs are symmetric, and some examples thereof are paths 3-6-7 and 4-6-7. 

Since the $\chain$ relation can be many-to-one or one-to-many, it makes sense to write formulas that consider only left contexts of chains that share their right context, or vice versa.
Thus, the paths of our second type, named \emph{hierarchical paths}, are made of such positions, but excluding outermost chains.
E.g., in Figure~\ref{fig:potl-example-word}, positions 2, 3 and 4 are all in the $\chain$ relation with 6, so 3-4 is a hierarchical path ($\chain(2, 6)$ is the outermost chain). 
Symmetrically, 7-9 is another hierarchical path. 
The reason for excluding the outermost chain is that, with most OPMs, such positions have a different semantic role than internal ones.
E.g., positions 3 and 4 are both calls terminated by the same exception, while 2 is the handler.
Positions 7 and 9 are both calls issued by the same function (the one called in position 1), while 11 is its return.
This is a consequence of properties~\ref{item:downward-prop} and~\ref{item:upward-prop} above.

In the next subsection, we describe in a complete and formal way POTL for finite string OPLs while in the subsequent subsection we briefly describe the necessary changes to deal with $\omega$-languages.

\subsection{POTL Syntax and Semantics}%
\label{sec:potl-syntax-semantics}

Given a finite set of atomic propositions $AP$, let $\mathrm{a} \in AP$, and $t \in \{d, u\}$.
The syntax of POTL is the following:
\begin{align*}
\varphi ::= &\; \mathrm{a}
\mid \neg \varphi
\mid \varphi \lor \varphi
\mid \lnext^t \varphi
\mid \lback^t \varphi
\mid \lcnext{t} \varphi
\mid \lcback{t} \varphi
\mid \lguntil{t}{\chi}{\varphi}{\varphi}
\mid \lgsince{t}{\chi}{\varphi}{\varphi} \\
&\mid \lhnext{t} \varphi
\mid \lhback{t} \varphi
\mid \lguntil{t}{H}{\varphi}{\varphi}
\mid \lgsince{t}{H}{\varphi}{\varphi}
\end{align*}

The truth of POTL formulas is defined w.r.t.\ a single word position.
Let $w$ be an OP word, and $\mathrm{a} \in AP$.
Then, for any position $i \in U$ of $w$,
we have $(w, i) \models \mathrm{a}$ iff $i \in P(\mathrm{a})$.
Operators such as $\land$ and $\neg$ have the usual semantics from propositional logic.
Next, while giving the formal semantics of POTL operators, we illustrate it by showing
how it can be used to express properties on program execution traces,
such as the one of Figure~\ref{fig:potl-example-word}.

\paragraph{\textbf{Next/back operators}}
The \emph{downward} next and back operators $\ldnext$ and $\ldback$
are like their LTL counterparts, except they are true only if the next
(resp.\ current) position is at a lower or equal ST level than the current (resp.\ preceding) one.
The \emph{upward} next and back, $\lunext$ and $\luback$, are symmetric.
Formally, $(w,i) \models \ldnext \varphi$ iff $(w,i+1) \models \varphi$
and $i \lessdot (i+1)$ or $i \doteq (i+1)$,
and $(w,i) \models \ldback \varphi$ iff $(w,i-1) \models \varphi$,
and $(i-1) \lessdot i$ or $(i-1) \doteq i$.
Substitute $\lessdot$ with $\gtrdot$ to obtain the semantics for $\lunext$ and $\luback$.

E.g., we can write $\ldnext \lcall$ to say that the next position is
an inner call (it holds in pos.\ 2, 3, 4 of
Figure~\ref{fig:potl-example-word}), $\ldback \lcall$ to say that the
previous position is a $\lcall$, and the
current is the first of the body of a function (pos.\ 2,
4, 5), or the $\lret$ of an empty one (pos.\ 8, 10), and $\luback
\lcall$ to say that the current position terminates an empty function
frame (holds in 6, 8, 10).  In pos.\ 2 $\ldnext \mathrm{p}_B$ holds, but
$\lunext \mathrm{p}_B$ does not.

\paragraph{\textbf{Chain Next/Back}}
The \emph{chain} next and back operators $\lcnext{t}$ and $\lcback{t}$
evaluate their argument resp.\ on future and past positions
in the chain relation with the current one.
The \emph{downward} (resp.\ \emph{upward}) variant only considers chains
whose right context goes down (resp.\ up) or remains at the same level in the ST\@.
Formally, $(w,i) \models \lcdnext \varphi$
iff there exists a position $j > i$ such that $\chain(i,j)$,
$i \lessdot j$ or $i \doteq j$, and $(w,j) \models \varphi$.
$(w,i) \models \lcdback \varphi$ iff there exists a position $j < i$
such that $\chain(j,i)$, $j \lessdot i$ or $j \doteq i$, and $(w,j) \models \varphi$.
Replace $\lessdot$ with $\gtrdot$ for the upward versions.

E.g., in pos.\ 1 of Figure~\ref{fig:potl-example-word}, $\lcdnext \mathrm{p}_{\mathit{Err}}$ holds
because $\chain(1,7)$ and $\chain(1,9)$,
meaning that $\mathrm{p}_A$ calls $\mathrm{p}_{\mathit{Err}}$ at least once.
Also, $\lcunext \lthrow$ is true in
$\lcall$ positions whose procedure is terminated by an exception thrown by
an inner procedure (e.g.\ pos.\ 3 and 4).
$\lcuback \lcall$ is true in $\lthrow$ statements that terminate at least one procedure
other than the one raising it, such as the one in pos.\ 6.
$\lcdnext \lret$ and $\lcunext \lret$ hold in $\lcall$s
to non-empty procedures that terminate normally, and not due to an uncaught exception
(e.g., pos.\ 1).

\paragraph{\textbf{(Summary) Until/Since operators}}
POTL has two kinds of until and since operators.
They express properties on paths, which are sequences of positions
obtained by iterating the different kinds of next or back operators.
In general, a \emph{path} of length $n \in \mathbb{N}$ between
$i, j \in U$ is a sequence of positions $i = i_1 < i_2 < \dots < i_n = j$.
The \emph{until} operator on a set of paths $\Gamma$ is defined as follows:
for any word $w$ and position $i \in U$,
and for any two POTL formulas $\varphi$ and $\psi$,
$(w, i) \models \lfuntil{\Gamma}{\varphi}{\psi}$
iff there exist a position $j \in U$, $j \geq i$,
and a path $i_1 < i_2 < \dots < i_n$ between $i$ and $j$ in
$\Gamma$
such that $(w, i_k) \models \varphi$ for any $1 \leq k < n$, and $(w, i_n) \models \psi$.
\emph{Since} operators are defined symmetrically.
Note that, depending on $\Gamma$, a path from $i$ to $j$ may not exist.
We define until/since operators by associating them with different sets of paths.

The \emph{summary} until $\lguntil{t}{\chi}{\psi}{\theta}$
(resp.\ since $\lgsince{t}{\chi}{\psi}{\theta}$) operator is obtained by inductively applying
the $\lnext^t$ and $\lcnext{t}$ (resp.\ $\lback^t$ and $\lcback{t}$) operators.
It holds in a position in which either $\theta$ holds,
or $\psi$ holds together with $\lnext^t (\lguntil{t}{\chi}{\psi}{\theta})$
(resp.\ $\lback^t (\lgsince{t}{\chi}{\psi}{\theta})$)
or $\lcnext{t} (\lguntil{t}{\chi}{\psi}{\theta})$
(resp.\ $\lcback{t} (\lgsince{t}{\chi}{\psi}{\theta})$).
It is an until operator on paths that can move not only between consecutive positions,
but also between contexts of a chain, skipping its body.
With the OPM of Figure~\ref{fig:opm-mcall}, this means skipping function bodies.
The downward variants can move between positions at the same level in the ST
(i.e., in the same simple chain body), or down in the nested chain structure.
The upward ones remain at the same level, or move to higher levels of the ST\@.

Formula $\lcuuntil{\top}{\lthrow}$ is true in positions contained in the frame
of a function that is terminated by an exception.
It is true in pos.\ 3 of Figure~\ref{fig:potl-example-word} because of path 3-6, 
and false in pos.\ 1, because no upward path can enter the chain whose contexts are pos.\ 1 and 11.
Formula $\lcduntil{\top}{\lthrow}$ is true in call positions whose function frame
contains $\lthrow$s, but that are not directly terminated by one of them,
such as the one in pos.\ 1 (with path 1-2-6) 

We formally define \emph{Downward Summary Paths} (DSP) as follows.
Given an OP word $w$, and two positions $i \leq j$ in $w$,
the DSP between $i$ and $j$, if it exists,
is a sequence of positions $i = i_1 < i_2 < \dots < i_n = j$ such that, for each $1 \leq p < n$,
\[
i_{p+1} =
\begin{cases}
  k & \text{if $k = \max\{ h \mid h \leq j \land \chain(i_p,h) \land (i_p \lessdot h \lor i_p \doteq h)\}$ exists;} \\
  i_p + 1 & \text{otherwise, if $i_p \lessdot (i_p + 1)$ or $i_p \doteq (i_p + 1)$.}
\end{cases}
\]
The Downward Summary (DS) until and since operators $\lcduntil{}{}$
and $\lcdsince{}{}$ use as $\Gamma$ the set of DSPs starting in the position in which
they are evaluated.
The definition for the upward counterparts is, again,
obtained by substituting $\gtrdot$ for $\lessdot$.
In Figure~\ref{fig:potl-example-word},
$\lcduntil{\lcall}{(\lret \land \mathrm{p}_{\mathit{Err}})}$
holds in pos.~1 because of path 1-7-8 and 1-9-10, 
$\lcusince{(\lcall \lor \lthrow)}{\mathrm{p}_B}$ in pos.~7 because of path 3-6-7, 
and $\lcuuntil{(\lcall \lor \lthrow)}{\lret}$ in 3
because of path 3-6-7-8. 

\paragraph{\textbf{Hierarchical operators.}}
A single position may be the left or right context of multiple chains.
The operators seen so far cannot keep this fact into account,
since they ``forget'' about a left context when they jump to the right one.
Thus, we introduce the \emph{hierarchical} next and back operators.
The \emph{upward} hierarchical next (resp.\ back),
$\lhunext \psi$ (resp.\ $\lhuback \psi$), is true iff the current position $j$
is the right context of a chain whose left context is $i$,
and $\psi$ holds in the next (resp.\ previous) pos.\ $j'$ that is a right context of $i$,
with $i \lessdot j, j'$.
So, $\lhunext \mathrm{p}_{\mathit{Err}}$ holds in pos.\ 7 of Figure~\ref{fig:potl-example-word}
because $\mathrm{p}_{\mathit{Err}}$ holds in 9,
and $\lhuback \mathrm{p}_{\mathit{Err}}$ in 9 because $\mathrm{p}_{\mathit{Err}}$ holds in 7.
In the ST, $\lhunext$ goes \emph{up} between $\lcall$s to $\mathrm{p}_{\mathit{Err}}$,
while $\lhuback$ goes down.
Their \emph{downward} counterparts behave symmetrically,
and consider multiple inner chains sharing their right context.
They are formally defined as:
\begin{itemize}
\item $(w,i) \models \lhunext \varphi$ iff
  there exist a position $h < i$ s.t.\ $\chain(h,i)$ and $h \lessdot i$
  and a position $j = \min\{ k \mid i < k \land \chain(h,k) \land h \lessdot k \}$
  and $(w,j) \models \varphi$;
\item $(w,i) \models \lhuback \varphi$ iff
  there exist a position $h < i$ s.t.\ $\chain(h,i)$ and $h \lessdot i$
  and a position $j = \max\{ k \mid k < i \land \chain(h,k) \land h \lessdot k \}$
  and $(w,j) \models \varphi$;
\item $(w,i) \models \lhdnext \varphi$ iff
  there exist a position $h > i$ s.t.\ $\chain(i,h)$ and $i \gtrdot h$
  and a position $j = \min\{ k \mid i < k \land \chain(k,h) \land k \gtrdot h \}$
  and $(w,j) \models \varphi$;
\item $(w,i) \models \lhdback \varphi$ iff
  there exist a position $h > i$ s.t.\ $\chain(i,h)$ and $i \gtrdot h$
  and a position $j = \max\{ k \mid k < i \land \chain(k,h) \land k \gtrdot h \}$
  and $(w,j) \models \varphi$.
\end{itemize}
In the ST of Figure~\ref{fig:sttree-example}, $\lhdnext$ and $\lhdback$ go \emph{down} and up among
$\lcall$s terminated by the same $\lthrow$.
For example, in pos.~3 $\lhdnext \mathrm{p}_C$ holds,
because both pos.~3 and 4 are in the chain relation with 6.
Similarly, in pos.~4 $\lhdback \mathrm{p}_B$ holds.
Note that these operators do not consider leftmost/rightmost contexts,
so $\lhunext \lret$ is false in pos.\ 9, as $\lcall \doteq \lret$,
and pos.\ 11 is the rightmost context of pos.\ 1.

The hierarchical until and since operators are defined by iterating these next and back operators.
The upward hierarchical path (UHP) between $i$ and $j$ is a sequence of positions
$i = i_1 < i_2 < \dots < i_n = j$ such that there exists a position $h < i$ such that
for each $1 \leq p \leq n$ we have $\chain(h,i_p)$ and $h \lessdot i_p$,
and for each $1 \leq q < n$ there exists no position $k$
such that $i_q < k < i_{q+1}$ and $\chain(h,k)$.
The until and since operators based on the set of UHPs
starting in the position in which they are evaluated are denoted
as $\lhuuntil{}{}$ and $\lhusince{}{}$.
E.g., $\lhuuntil{\lcall}{\mathrm{p}_{\mathit{Err}}}$ holds in pos.~7
because of the singleton path 7 and path 7-9, 
and $\lhusince{\lcall}{\mathrm{p}_{\mathit{Err}}}$ in pos.~9
because of paths 9 and 7-9. 

The downward hierarchical path (DHP) between $i$ and $j$ is a sequence of positions
$i = i_1 < i_2 < \dots < i_n = j$ such that there exists a position $h > j$ such that
for each $1 \leq p \leq n$ we have $\chain(i_p,h)$ and $i_p \gtrdot h$,
and for each $1 \leq q < n$ there exists no position $k$
such that $i_q < k < i_{q+1}$ and $\chain(k,h)$.
The until and since operators based on the set of DHPs
starting in the position in which they are evaluated are denoted
as $\lhduntil{}{}$ and $\lhdsince{}{}$.
In Figure~\ref{fig:potl-example-word}, $\lhduntil{\lcall}{\mathrm{p}_C}$ holds in pos.~3,
and $\lhdsince{\lcall}{\mathrm{p}_B}$ in pos.~4, both because of path 3-4. 

\paragraph{\textbf{Equivalences}}
The POTL until and since operators enjoy expansion laws similar to those of LTL\@.
Here we give those for two until operators,
those for their since and downward counterparts being symmetric.
All such laws are proved in Appendix~\ref{sec:expansion-proofs}.
\begin{align*}
  \lguntil{t}{\chi}{\varphi}{\psi} &\equiv
    \psi \lor \Big(\varphi \land \big(\lnext^t (\lguntil{t}{\chi}{\varphi}{\psi})
      \lor \lcnext{t} (\lguntil{t}{\chi}{\varphi}{\psi})\big)\Big) \\
  \lhuuntil{\varphi}{\psi} &\equiv
    (\psi \land \lcdback \top \land \neg \lcuback \top) \lor
     \big(\varphi \land \lhunext (\lhuuntil{\varphi}{\psi})\big)
\end{align*}

As in LTL, it is worth defining some derived operators.
For $t \in \{d, u\}$, we define the downward/upward summary \emph{eventually} as
$\leven{t} \varphi := \lguntil{t}{\chi}{\top}{\varphi}$,
and the downward/upward summary \emph{globally} as
$\lglob{t} \varphi := \neg \leven{t} (\neg \varphi)$.
$\leven{u} \varphi$ and $\lglob{u} \varphi$ respectively say that $\varphi$
holds in one or all positions in the path from the current position to the root of the ST\@.
Their downward counterparts are more interesting: they consider all positions
in the current rhs and its subtrees, starting from the current position.
$\leven{d} \varphi$ says that $\varphi$ holds in at least one of such positions,
and $\lglob{d} \varphi$ in all of them.
E.g., if $\lglob{d} (\neg \mathrm{p}_A)$ holds in a $\lcall$,
it means that $\mathrm{p}_A$ never holds in its whole function body,
which is the subtree rooted next to the $\lcall$.

We anticipate that preventing downward paths from crossing the boundaries of the current subtrees and conversely imposing upward ones to exit it without entering any inner one adds, rather than limiting, generality w.r.t.\ paths that can cross both such boundaries.

\subsection{POTL on \texorpdfstring{$\omega$}{Omega}-Words}%
\label{sec:optl-omega}

Since applications in model checking usually require temporal logics on
infinite words, we now extend POTL to $\omega$-words.

To define OP $\omega$-words, it suffices to replace the finite set of positions $U$
with the set of natural numbers $\mathbb{N}$ in the definition of OP words.
Then, the formal semantics of all POTL operators remains the same as
in Section~\ref{sec:potl-syntax-semantics}.
The only difference in the intuitive meaning of operators occurs in $\omega$-words
with open chains.
In fact, chain next operators ($\lcdnext$ and $\lcunext$)
do not hold on the left contexts of open chains,
as the $\chain$ relation is undefined on them.
The same can be said for downward hierarchical operators,
when evaluated on left contexts of open chains.

Also, property~\ref{item:upward-prop} of the $\chain$ relation does not hold
if a position $i$ is the left context of an open chain.
In this case, there may be positions $j_1 < j_2 < \dots < j_n$ s.t.\ $\chain(i, j_p)$
and $i \lessdot j_p$ for all $1 \leq p \leq n$, but no position $k$ s.t.\ $\chain(i,k)$
and $i \gtrdot k$ or $i \doteq k$.

\subsection{Motivating Examples}%
\label{sec:motivating-examples}

POTL can express many useful requirements of procedural programs. To emphasize the potential practical applications in automatic verification, we supply a few examples of typical program properties expressed as POTL formulas.

Let $\llglob \psi$ be the LTL \emph{globally} operator,
which can be expressed in POTL as in Section~\ref{sec:ltl-comp}.
POTL can express Hoare-style pre/postconditions
with formulas such as $\llglob (\lcall \land \rho \implies \lcdnext (\lret \land \theta))$,
where $\rho$ is the precondition, and $\theta$ is the postcondition.

Unlike NWTL, POTL can easily express properties related to exception handling
and interrupt management.
E.g., the shortcut
$\lthrnext(\psi) := \lunext (\lthrow \land \psi) \lor \lcunext (\lthrow \land \psi)$,
evaluated in a $\lcall$, states that the procedure currently started
is terminated by an $\lthrow$ in which $\psi$ holds.
So, $\llglob (\lcall \land \rho \land \lthrnext(\top) \implies \lthrnext(\theta))$
means that if precondition $\rho$ holds when a procedure is called,
then postcondition $\theta$ must hold if that procedure is terminated by an exception.
In object oriented programming languages,
if $\rho \equiv \theta$ is a class invariant asserting that a class instance's state is valid,
this formula expresses \emph{weak (or basic) exception safety}~\cite{Abrahams00},
and \emph{strong exception safety} if $\rho$ and $\theta$
express particular states of the class instance.
The \emph{no-throw guarantee} can be stated with
$\llglob (\lcall \land \mathrm{p}_A \implies \neg \lthrnext(\top))$,
meaning procedure $\mathrm{p}_A$ is never interrupted by an exception.

\emph{Stack inspection}~\cite{EsparzaKS03,JensenLT99},
i.e.\ properties regarding the sequence of procedures
active in the program's stack at a certain point of its execution,
is an important class of requirements that can be expressed with shortcut
$\lcallsince(\varphi, \psi) := \lcdsince{(\lcall \implies \varphi)}{(\lcall \land \psi)}$,
which subsumes the \emph{call since} of CaRet, and works with exceptions too.
E.g.,
\(\llglob \big((\lcall \land \mathrm{p}_B \land
                \lcallsince(\top, \mathrm{p}_A))
               \implies \lthrnext(\top)
          \big)
\)
means that whenever $\mathrm{p}_B$ is executed and
at least one instance of $\mathrm{p}_A$ is on the stack,
$\mathrm{p}_B$ is terminated by an exception.
The OPA of Figure~\ref{fig:example-prog} satisfies this formula,
because $\mathrm{p}_B$ is always called by $\mathrm{p}_A$,
and $\mathrm{p}_C$ always throws.
If the OPA was an $\omega$OPBA, it would not satisfy such formula
because of computations where $\mathrm{p}_C$ does not terminate.

\subsection{Comparison with the state of the art}

\subsubsection{Linear Temporal Logic (LTL)}%
\label{sec:ltl-comp}

The main limitation of LTL is that the algebraic structure it is defined on
only contains a linear order on word positions.
Thus, it fails to model systems that require an additional binary relation,
such as the $\chain$ relation of POTL\@.
LTL is in fact expressively equivalent to the first-order fragment of regular languages,
and it cannot represent context-free languages, as POTL does.
LTL model checking on pushdown formalisms has been investigated extensively~\cite{EsparzaHRS00,AlurBEGRY05}.

On the other hand, POTL can express all LTL operators, so that POTL is strictly more expressive than LTL\@.
Any LTL Next formula $\lnext \varphi$ is in fact equivalent to the POTL formula $\ldnext \varphi' \lor \lunext \varphi'$, where $\varphi'$ is the translation of $\varphi$ into POTL, and the LTL Back can be translated symmetrically.

The \emph{Globally} operator can be translated as
$\llglob \psi := \neg \leven{u} (\leven{d} \neg \psi)$.
This formula contains an upward summary eventually followed by a downward one,
and it can be explained by thinking to a word's ST\@.
The upward eventually evaluates its argument on all positions from the current one to the root.
Its argument considers paths from each one of such positions to all the leaves of the
subtrees rooted at their right, in the same rhs.
Together, the two eventually operators consider paths from the current position to
all subsequent positions in the word.
Thus, with the initial negation this formula means that $\neg \psi$ never holds
in such positions, which is the meaning of the LTL Globally operator.

The translation for LTL Until and Since is much more involved.
We need to define some shortcuts, that will be used again in Section~\ref{sec:cxpath-translation}.
For any $a \subseteq AP$,
$\sigma_a := \bigwedge_{\mathrm{p} \in a} \mathrm{p} \land \bigwedge_{\mathrm{q} \not\in a} \neg \mathrm{q}$
holds in a position $i$ iff $a$ is the set of atomic propositions holding in $i$.
For any POTL formula $\gamma$, let
$\lganext{\lessdot} \gamma := \bigvee_{a,b \subseteq AP, \, a \lessdot b} (\sigma_a \land \lcdnext (\sigma_b \land \gamma))$
be the restriction of $\lcdnext \gamma$ to chains with contexts in the $\lessdot$ PR;\@
$\lgaback{\gtrdot} \gamma$ is analogous.

The translation for $\lluntil{\varphi}{\psi}$ follows, that for LTL Since being symmetric:
\[
  \psi' \lor
  \lcuuntil
    {\big(\varphi' \land \alpha(\varphi')\big)}
    {\big(\psi' \lor
       \lcduntil
         {(\varphi' \land \beta(\varphi'))}
         {(\psi' \land \beta(\varphi'))}\big)}
\]
where $\varphi'$ and $\psi'$ are the translations of $\varphi$ and $\psi$ into POTL, and
\begin{align*}
\alpha(\varphi') &:=
  \lcunext \top \implies
  \neg \big(\ldnext (\lcduntil{\top}{\neg \varphi'})
        \lor \lganext{\lessdot} (\lcduntil{\top}{\neg \varphi'})\big) \\
\beta(\varphi') &:=
  \lcdback \top \implies
    \neg \big(\luback (\lcusince{\top}{\neg \varphi'})
          \lor \lgaback{\gtrdot} (\lcusince{\top}{\neg \varphi'})\big)
\end{align*}

The main formula is the concatenation of a US until and a DS until,
and it can be explained similarly to the translation for LTL Globally.
Let $i$ be the word position in which the formula is evaluated, and $j$ the last one of the linear path, in which $\psi'$ holds.
The outermost US until is witnessed by a path from $i$ to a position $i'$ which, in the ST, is part of the rhs which is the closest common ancestor of $i$ and $j$.
In all positions $i < k < i'$ in this path, formula $\alpha(\varphi)$ holds.
It means $\varphi'$ holds in all positions contained in the subtree rooted at the non-terminal to the immediate right of $k$ (i.e., in the body of the chain whose left context is $k$).

Then, the DS until is witnessed by a downward path from $i'$ to $j$.
Here $\beta(\varphi)$ has a role symmetric to $\alpha(\varphi)$.
It forces $\varphi'$ to hold in all subtrees rooted at the non-terminal to the left of positions in the DS path.

Thus, formulas $\alpha(\varphi)$ and $\beta(\varphi)$ make sure $\varphi'$ holds in all chain bodies skipped by the summary paths, and $\psi'$ holds in $j$.

\subsubsection{Logics on Nested Words}%
\label{sec:nw-comp}

The first temporal logics with explicit context-free aware modalities
were based on nested words (cf.\ Section~\ref{sec:structured-langs}).

CaRet~\cite{AlurEM04} was the first temporal logic on nested words to be introduced,
and it focuses on expressing properties on procedural programs,
which explains its choice of modalities.
The \emph{Abstract} Next and Until operators are defined on paths of positions
in the frame of the same function, skipping frames of nested calls.
The \emph{Caller} Next and Until are actually past modalities,
and they operate on paths made of the calls of function frames
containing the current position.
LTL Next and Until are also present.
The Caller operators enable upward movement in the ST of a nested word,
and abstract operators enable movement in the same rhs.
However, no CaRet operator allows pure downward movement in the ST,
which is needed to express properties limited to a single subtree.
While the LTL until can go downward, it can also go beyond the rightmost leaf of a subtree,
thus effectively jumping upwards.

This seems to be the main expressive limitation of CaRet,
which is conjectured not to be FO-complete~\cite{lmcs/AlurABEIL08}.
In fact, FO-complete temporal logics were introduced in~\cite{lmcs/AlurABEIL08}
by adding various kinds of \emph{Within} modalities to CaRet.
Such operators limit their operands to span only positions within the same
call-return pair, and hence the same subtree of the ST, at the cost of an exponential
jump in the complexity of model checking.

Another approach to FO-completeness is that of NWTL~\cite{lmcs/AlurABEIL08},
which is based on \emph{Summary} Until and Since operators.
Summary paths are made of either consecutive positions, or matched call-return pairs.
Thus, they can skip function bodies, and enter or exit them.
Summary-up and down paths, and the respective operators, can be obtained
from summary paths, enabling exclusive upward or downward movement in the ST\@.
In particular, summary-down operators may express properties limited to a single subtree.

In Corollary~\ref{cor:optl-in-potl}, we show that
CaRet $\subseteq$ NWTL $\subset$ POTL\@.

\subsubsection{Logics on OPLs}
The only way to overcome the limitations of nested words is to base a temporal logic
on a more general algebraic structure.
OPTL~\cite{ChiariMP20a} was introduced with this aim,
but it shares some of the limitations that CaRet has on nested words.
It features all LTL past and future operators,
plus the \emph{Matching} Next ($\lanext$) and Back ($\laback$) operators,
resp.\ equivalent to POTL $\lcunext$ and $\lcuback$, OP \emph{Summary} Until and Since,
and \emph{Hierarchical} Until and Since.
POTL has several advantages over OPTL in its ease of expressing requirements.

Given a set of PRs $\Pi$, an OPTL Summary Until $\luntil{\Pi}{}{}$
evaluated on a position $i$ considers paths $i = i_1 < i_2 < \dots < i_n = j$,
with $i \leq j$, such that for any $1 \leq p < n$,
\[
i_{p+1} =
\begin{cases}
  h & \text{if there exists $h$ s.t.\ $\chain(i_p, h)$ and $(i_p \doteq h$ or $i_p \gtrdot h)$, and $h \leq j$;}\\
  i_p + 1 & \text{if $i_p \pr i_{p+1}$ for some $\pi \in \Pi$, otherwise.}
\end{cases}
\]
The Summary Since $\lsince{\Pi}{}{}$ is symmetric: positions in the $\chain$ relation must be in the $\lessdot$ or $\doteq$ PRs.
Since the PRs checked on chain contexts are fixed, the user can control whether such paths go up or down in the ST only partially.
OPTL's $\luntil{\doteq \gtrdot}{\varphi}{\psi}$ is equivalent to POTL's $\lcuuntil{\varphi}{\psi}$,
and $\lsince{\lessdot \doteq}{\varphi}{\psi}$ to $\lcdsince{\varphi}{\psi}$: both operators only go upwards in the ST\@.
However, there is no OPTL operator equivalent to POTL's $\lcduntil{}{}$ or $\lcusince{}{}$, which go downward.
This makes it difficult to express function-local properties limited to a single subtree in OPTL\@.

\begin{figure}[tb]
\centering
\includegraphics{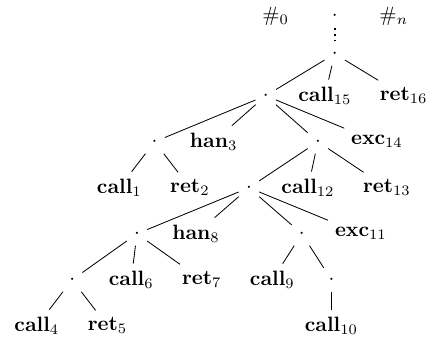}
\caption{Example OP word on OPM $M_\lcall$.}%
\label{fig:optl-example}
\end{figure}

E.g., consider POTL formula $\alpha := \llglob (\lthrow \implies \lcdback (\lhandle \land \leven{d} \mathrm{p}_A))$, which means that if an exception is thrown, it is always caught, and procedure $\mathrm{p}_A$ is called at some point inside the $\lhandle$-$\lthrow$ block.
One could try to translate it into OPTL with a formula such as $\beta := \llglob (\lthrow \implies \laback (\lhandle \land \luntil{\lessdot \doteq}{\top}{\mathrm{p}_A}))$.
Consider the OP word of Figure~\ref{fig:optl-example}.
When the until in $\beta$ is evaluated in the $\lhandle$ of pos.\ 3, its paths can only consider positions 4 and 5, because paths touching such positions cannot pass the $\gtrdot$ relation between 5 and 6.
Its paths, unlike those of its POTL counterpart, cannot jump between chain contexts in the $\lessdot$ relation, and cannot reach positions 6, 8, and 12 in this way.
If $\mathrm{p}_A$ held in pos.\ 6 and 10, $\alpha$ would be true in the word, but $\beta$ would be false.
Replacing the until with $\luntil{\lessdot \doteq \gtrdot}{\top}{\mathrm{p}_A}$ overcomes such issues, but it introduces another one: its paths would go past pos.\ 14, going outside of the subtree.
Thus, if $\mathrm{p}_A$ held only in pos.\ 15, $\alpha$ would be false, but this variant of $\beta$ would hold.

OPTL has Hierarchical operators, too.
Its yield-precedence Hierarchical Until ($\lhyuntil{}{}$) and since ($\lhysince{}{}$) operators,
when evaluated on a position $i$, consider paths $i < j_1 < j_2 < \dots < j_n$ s.t.\ $\chain(i,j_p)$ and $i \lessdot j_p$ for all $1 \leq p \leq n$, and there is no $k < j_1$ s.t.\ $\chain(i,k)$.
Their take-precedence counterparts ($\lhtuntil{}{}$ and $\lhtsince{}{}$) are symmetric.
Thus, such paths do not start in the position where until and since operators are evaluated, but always in a future position: this is another limitation of OPTL\@.
In fact, it is not possible to concatenate them to express complex properties on right (resp.\ left) contexts of chains sharing their left (resp.\ right) context, such as several function calls issued by the same function, or multiple function calls terminated by the same exception.
POTL has both Hierarchical Next/Back and Until/Since pairs which are composable, making it expressively complete on such positions.
For example, POTL formula
\(
\gamma := \llglob (\lcall \land \mathrm{p}_A \implies \lcdsince{(\lhusince{\top}{(\lcall \land \mathrm{p}_B)})}{(\ldback \# \lor \lcdback \#)})
\)
means that whenever procedure $\mathrm{p}_A$ is called, all procedures in the stack have previously invoked $\mathrm{p}_B$ (possibly excluding the one directly calling $\mathrm{p}_A$).
While $\lcdsince{}{}$ can be replaced with OPTL's $\lsince{\lessdot \doteq}{}{}$, POTL's $\lhusince{}{}$ cannot be easily translated.
In fact, OPTL's Hierarchical operators would only allow us to state that $\mathrm{p}_B$ is invoked by the procedures in the stack, but not necessarily before the call to $\mathrm{p}_A$.

The above intuition about OPTL's weaknesses are made formal in the next subsection.

\subsubsection{OPTL is not expressively complete}

In this section, we prove that no OPTL formula is equivalent to POTL formula
$\leven{d} \mathrm{p}_A$.
The proof is quite elaborate, which is unsurprising,
since the analogous problem of the comparison between CaRet and NWTL is still open.

First, we prove the following
\begin{lem}[Pumping Lemma for OPTL]%
\label{lemma:pumping}
Let $\varphi$ be an OPTL formula and $L$ an OPL, both defined on
a set of atomic propositions $AP$ and an OPM $M_{AP}$.
Then, for some positive integer $n$, for each $w \in L$, $|w| \geq n$,
there exist strings $u, v, x, y, z \in \powset{AP}^*$ such that
$w = uvxyz$, $|vy| > 1$, $|vxy| \leq n$
and for any $k > 0$ we have $w' = u v^k x y^k z \in L$;
for any $0 \leq j \leq k$ and $0 \leq i < |v|$ we have
$(w, |u| + i) \models \varphi$ iff $(w', |u| + j |v| + i) \models \varphi$,
and for any $0 \leq i < |y|$ we have
$(w, |u v^k x| + i) \models \varphi$ iff $(w', |u v^k x| + j |y| + i) \models \varphi$.
\end{lem}
\begin{proof}
Given a word $w \in L$, we define $\lambda(w)$ as the word of length $|w|$ such that,
if position $i$ of $w$ is labeled with $a$, then the same position in $\lambda(w)$
is labeled with $(a, 1)$ if $(w, i) \models \varphi$, and with $(a, 0)$ otherwise.
Let $\lambda(L) = \{\lambda(w) \mid w \in L\}$,
and $\lambda^{-1}$ is such that $\lambda^{-1}(\lambda(w)) = w$.
If we prove that $\lambda(L)$ is context-free,
from the classic Pumping Lemma~\cite{Harrison78} follows that, for some $n > 0$,
for all $\hat{w} \in \lambda(L)$ there exist strings
$\hat{u}, \hat{v}, \hat{x}, \hat{y}, \hat{z} \in (\powset{AP} \times \{0, 1\})^*$ such that
$\hat{w} = \hat{u} \hat{v} \hat{x} \hat{y} \hat{z}$, $|\hat{v} \hat{y}| > 1$,
$|\hat{v} \hat{x} \hat{y}| \leq n$
and for any $k > 0$ we have $\hat{w}' = \hat{u} \hat{v}^k \hat{x} \hat{y}^k \hat{z} \in \lambda(L)$.
The claim follows by applying $\lambda^{-1}$ to such strings,
and the word positions in which $\varphi$ holds in $\lambda^{-1}(\hat{w})$ are those labeled with $1$.

To prove that $\lambda(L)$ is context-free, we use the OPTL model checking construction
given in~\cite{ChiariMP20a}, which yields an OPA
$\mathcal{A}_\varphi = (\powset{AP}, \allowbreak M_{AP}, \allowbreak Q, \allowbreak I, \allowbreak F, \allowbreak \delta)$
accepting models of $\varphi$.
The states of $\mathcal{A}_\varphi$ are elements of the set $\clos{\varphi}$,
which contains $\varphi$ and all its subformulas.
Given a word $w$ compatible with $M_{AP}$, the accepting computations of $\mathcal{A}_\varphi$
are such that, for each $0 \leq i < |w|$, the state of $\mathcal{A}_\varphi$
prior to reading position $i$ contains $\psi \in \clos{\varphi}$ iff $(w, i) \models \psi$.

Thus, we build OPA
$\mathcal{A}_{\lambda(\powset{AP}^*)} = (\powset{AP} \times \{0, 1\}, \allowbreak M_{AP}, \allowbreak Q, \allowbreak I', \allowbreak F, \allowbreak \delta')$
that reads words on $(\powset{AP} \times \{0, 1\})^*$ and accepts $\lambda(\powset{AP}^*)$ as follows:
\begin{itemize}
\item $I'$ is the set of all states in $Q$ not containing past operators (and possibly $\varphi$);
\item $\delta'_\mathit{push}$ is such that if
  $(\Phi, a, \Theta) \in \delta_\mathit{push}$, then
  $(\Phi, (a, 1), \Theta) \in \delta_\mathit{push}$ if $\varphi \in \Phi$,
  and $(\Phi, (a, 0), \Theta) \in \delta_\mathit{push}$ otherwise;
\item $\delta'_\mathit{shift}$ is derived from $\delta_\mathit{shift}$ similarly;
\item $\delta'_\mathit{pop} = \delta_\mathit{pop}$.
\end{itemize}

\noindent
Since $L$ is an OPL, there exists an OPA $\mathcal{A}_L$ accepting it.
$\mathcal{A}_L$ can be easily modified to obtain $\mathcal{A}'_L$,
an OPA accepting all words $\hat{w} \in (\powset{AP} \times \{0, 1\})^*$ such that
the underlying word $w \in \powset{AP}^*$ is in $L$.
Language $\lambda(L)$ is the intersection between the language accepted
by $\mathcal{A}'_L$, and $\lambda(\powset{AP}^*)$.
Since OPLs are closed under intersection, $\lambda(L)$ is also an OPL\@.
\end{proof}

Let $L_\lcall$ be the max-language generated by OPM $M_\lcall$,
with the addition that $\mathrm{p}_A$ may appear in any word position.
We prove the following:

\begin{figure}[tb]
\centering
\includegraphics{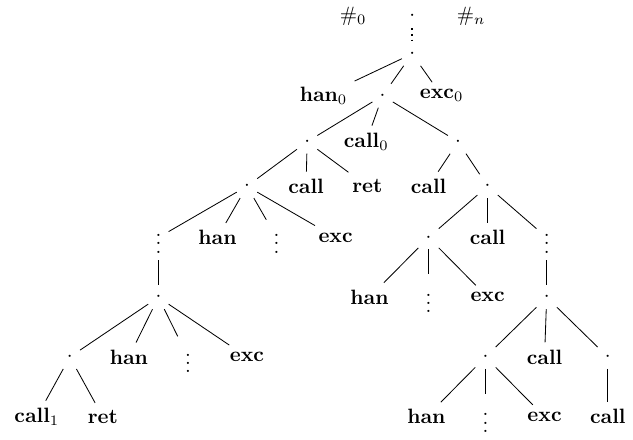}
\caption{Structure of a word in $L_\lcall$.}%
\label{fig:exc-tree}
\end{figure}

\begin{thm}%
\label{thm:optl-vs-potl}
Given the POTL formula $\leven{d} \mathrm{p}_A$,
for every OPTL formula $\varphi$ there exist a word $w \in L_\lcall$ and an integer $0 \leq i < |w|$
such that either $(w, i) \models \leven{d} \mathrm{p}_A$ and $(w, i) \not\models \varphi$,
or $(w, i) \not\models \leven{d} \mathrm{p}_A$ and $(w, i) \models \varphi$.
\end{thm}
\begin{proof}
Figure~\ref{fig:exc-tree} shows the structure of the syntax trees of words in a subset of $L_\lcall$.
Dots between $\lhandle$--$\lthrow$ pairs can be replaced with repetitions of the whole tree structure,
and other dots with the repetition of surrounding tree fragments
(e.g., $\lcall$--$\lret$ or $\lhandle$--$\lthrow$).
$\lhandle_0$ is the position in which $\varphi$ is evaluated, and $\lthrow_0$ is its matched $\lthrow$.
$\lcall_1$ is the $\lcall$ right after $\lhandle_0$, and $\lcall_0$ is the one at the highest
level of the subtree between $\lhandle_0$ and $\lthrow_0$.
$\lcall$s between $\lcall_0$ and $\lthrow_0$ do not have a corresponding $\lret$,
and are terminated by $\lthrow_0$.
The word delimited by $\lhandle_0$ and $\lthrow_0$ is itself part of a larger tree with the same structure.
For $\varphi$ to be equivalent to $\leven{d} \mathrm{p}_A$, it must be able to
\begin{enumerate}
\item look for the symbol $\mathrm{p}_A$ in all positions between $\lhandle_0$ and $\lthrow_0$; and
\item not consider any positions before $\lhandle_0$ or after $\lthrow_0$.
\end{enumerate}
In the following, we show that any OPTL formula $\varphi$ fails to satisfy both requirements, thus
\begin{enumerate}
\item one can hide $\mathrm{p}_A$ in one of the positions not covered,
  so that $\varphi$ is false in $\lhandle_0$, but $\leven{d} \mathrm{p}_A$ is true; or
\item put $\mathrm{p}_A$ in one of the positions outside $\lhandle_0$--$\lthrow_0$ reached by $\varphi$,
  so that it is true in $\lhandle_0$, but $\leven{d} \mathrm{p}_A$ is not.
\end{enumerate}

\noindent
If $\varphi$ is evaluated on position $\lhandle_0$, it must contain some modal operator
based on paths that reach each position between $\lhandle_0$ and $\lthrow_0$.
The length of the word between $\lhandle_0$ and $\lthrow_0$ has no limit,
so $\varphi$ must contain at least an until operator, which may be a LTL Until,
an OPTL Hierarchical Until, or an OPTL Summary Until $\luntil{\Pi}{}{}$.
For $\luntil{\Pi}{}{}$, we must have $\gtrdot \in \Pi$,
or the path would not be able to reach past position $\lret_1$
(remember that an OPTL summary path cannot skip chains with contexts in the $\lessdot$ relation).
The presence of $\gtrdot$ allows the Summary Until to reach positions past $\lthrow_0$:
the formula could be true if $\mathrm{p}_A$ appears after $\lthrow_0$,
but not between $\lhandle_0$ and $\lthrow_0$, unlike $\leven{d} \mathrm{p}_A$.
To avoid this, the path must be stopped earlier, by embedding an appropriate subformula
as the left operand of the until.

Suppose there exists a formula $\psi$ that is true in $\lthrow_0$, and false in all positions
between $\lhandle_0$ and $\lthrow_0$.
By Lemma~\ref{lemma:pumping}, there exists an integer $n$ such that for any $w \in L_\lcall$
longer than $n$ there is $w' = u v^k x y^k z \in L_\lcall$, for some $k > 0$,
such that either (a) $\psi$ never holds in $v^k x y^k$, or (b) it holds at least $k$ times in there.
We can take $w$ such that $\lhandle_0$ and $\lthrow_0$ both appear after position $n$,
and they contain nested $\lhandle$--$\lthrow$ pairs.
In case (a), $\psi$ cannot distinguish $\lthrow_0$ from nested $\lthrow$s, so a $\varphi$ based on
$\psi$ is not equivalent to $\leven{d} \mathrm{p}_A$ in $\lhandle_0$.
The same can be said in case (b), by evaluating $\varphi$ in a $\lhandle$ from $v^i$ with $i < k$.
In this case, also chaining multiple untils, each one ending in a position in which $\psi$ holds,
does not work, as $k$ can be increased beyond the finite length of any OPTL formula.

The above argument holds verbatim for LTL Until, and does not change if we prepend
LTL or abstract next operators to the until, because the length of the branch
between $\lcall_1$ and $\lcall_0$ is unlimited.
The argument for using since operators starting from $\lthrow_0$ is symmetric.
If both until and since operators are used, it suffices to apply the Pumping Lemma twice
(one for until and one for since), and take a value of $k$ large enough that
a part of the string cannot be reached by the number of until and since operators in the formula.
If Hierarchical operators are used, the argument does not change,
as they still need nested until or since operators to cover the whole subtree.

This argument also holds when the formula contains (possibly nested) negated until operators.
This is trivial if their paths cannot reach part of the subtree between $\lhandle_0$ and $\lthrow_0$.
If, instead, they can reach positions past $\lthrow_0$, we can build a word with $\mathrm{p}_A$
in one of such positions, but not between $\lhandle_0$ and $\lthrow_0$.
To distinguish it from a word with $\mathrm{p}_A$ between $\lhandle_0$ and $\lthrow_0$,
the formula would need a subformula that can distinguish positions between $\lhandle_0$ and $\lthrow_0$
from those outside, which would contradict Lemma~\ref{lemma:pumping}.

Until now, we have proved that a formula equivalent to $\leven{d} \mathrm{p}_A$ cannot
contain until or since operators that stop exactly at $\lthrow_0$.
However, they could be stopped earlier.
In the following, we show that any such OPTL formula can only work for words of a limited length.
Hence, no OPTL formula is equivalent to $\leven{d} \mathrm{p}_A$ on all words in $L_\lcall$.

Let $w \in L_\lcall$, and $x = \lhandle \, y \, \lthrow$ a subword of $w$
with the structure of Figure~\ref{fig:exc-tree}, in which each $\lhandle$ has a matched $\lthrow$, and conversely.
We define $\Theight(x) = 0$ if $y$ contains no positions labeled with $\lhandle$ or $\lthrow$
(hence, only $\lcall$s and $\lret$s).
Otherwise, let $x' = \lhandle \, y' \, \lthrow$ be the proper subword of $y$
with the maximum value of $\Theight(x')$: we set $\Theight(x) = \Theight(x') + 1$.

We prove by induction on $\Theight(x)$ that any OPTL formula evaluated in the first position of $x$
must contain nested until or since operators with a nesting depth
of at least $2 \cdot \Theight(x) + 1$ to be equivalent to $\leven{d} \mathrm{p}_A$.

If $\Theight(x) = 0$, at least one until or since operator is needed, as the length of $x$ is not fixed.
E.g., OPTL formula $\luntil{\lessdot \doteq \gtrdot}{\neg \lthrow}{\mathrm{p}_A}$ suffices.

If $\Theight(x) = n > 0$, Figure~\ref{fig:exc-tree} shows a possible structure of $x$.
Any Summary Until in the formula must be nested into another operator,
or its paths would jump to, and go past, the last position of $x$ ($\lthrow_0$).
A Summary Until could be, instead, nested into any number of nested next operators,
to be evaluated in one of the positions shown in Figure~\ref{fig:exc-tree}
between $\lhandle_0$ and $\lcall_0$.
(The tree fragments between $\lcall_1$ and $\lcall_0$ can be repeated enough times
so that the next operators alone cannot reach $\lcall_0$.)
As noted earlier, any such summary until must allow for paths
with consecutive positions in the $\gtrdot$ relation.
It may also jump to $\lthrow_0$ by following the chain relation,
because $\lcall \gtrdot \lthrow$.
Hence, the until must be stopped earlier by choosing appropriate operands
(e.g., $\neg \lanext \lthrow$ as the left operand).
However, this leaves the subword between $\lcall_0$ and $\lthrow_0$ unreached,
so any of its positions containing (or not) $\mathrm{p}_A$ would be ignored.
This can only be solved with another until operator, so at least two are needed.
If it is a Summary Until, then it must not allow the $\gtrdot$ relation,
or it could, again, escape $\lthrow_0$ (e.g.\ by skipping chains between $\lcall$s and $\lthrow_0$).
The argument can be extended by considering an LTL Until which stops anywhere before $\lthrow_0$,
or since operators evaluated in $\lthrow$ (e.g., nested in a $\lanext$ operator).
The same can be said for Hierarchical operators, which can cover only a part
of the subtree if used alone.

Let $x' = \lhandle \, y' \, \lthrow$ be a proper subword of $y$ with $\Theight(x') = n-1$.
Suppose it appears before $\lcall_0$ (the other case is symmetric).
It needs at least an until or since operator to be covered,
which must not escape $\lhandle_0$ or $\lthrow_0$.
The Pumping Lemma can be used to show that no OPTL formula can distinguish positions
in $x$ or $x'$ from those outside.
Thus, a formula with until or since operators that do not exit $y'$ is needed.
By the inductive hypothesis, it consists of at least $2 (n-1) + 1$ until or since operators,
thus $x$ needs $2 n + 1$ of them.

Note that the argument also holds if the until formulas are negated,
because negation cannot change the type of paths considered by an operator,
and cannot decrease the number of nested untils needed to cover the whole subtree.
\end{proof}

\section{First-Order Completeness on Finite Words}%
\label{sec:fo-completeness}

To show that POTL $\subseteq$ FOL on finite OP words,
we give a direct translation of POTL into FOL\@.
Proving that FOL $\subseteq$ POTL is more involved:
we translate \xuntil{}~\cite{Marx2004}, a logic on trees, into POTL\@.
\xuntil{} (defined in Section~\ref{sec:cxpath-translation})
is a logic on trees introduced to prove the expressive completeness of Conditional XPath,
and from its being equivalent to FOL on trees~\cite{Marx2005,LibkinS10}
we derive a FO-completeness result for POTL\@.

\subsection{First-Order Semantics of POTL}%
\label{sec:fo-semantics}

We show that POTL can be expressed with FOL
equipped with monadic relations for atomic propositions,
a total order on positions, and the chain relation between pairs of positions.
We define below the translation function $\nu$,
such that for any POTL formula $\varphi$, word $w$ and position $i$,
$(w, i) \models \nu_\varphi(x)$ iff $(w, i) \models \varphi$.
The translation for propositional operators is trivial.

For temporal operators, we first need to define a few auxiliary formulas.
We define the successor relation as the FO formula
\[
  \lsucc(x, y) := x < y \land \neg \exists z (x < z \land z < y).
\]
The PRs between positions can be expressed by means of propositional combinations
of monadic atomic relations only.
Given a set of atomic propositions $a \subseteq AP$, we define formula $\sigma_a(x)$,
stating that all and only propositions in $a$ hold in position $x$, as follows:
\begin{equation}
  \sigma_a(x) :=
    \bigwedge_{\mathrm{p} \in a} \mathrm{p}(x)
    \land
    \bigwedge_{\mathrm{p} \in AP \setminus a} \neg \mathrm{p}(x)%
\label{eq:sigmaa}
\end{equation}
For any pair of FO variables $x, y$ and $\prf \in \{\lessdot, \doteq, \gtrdot\}$,
we can build formula
\[
  x \pr y :=
    \bigvee_{a, b \subseteq AP \mid a \pr b} (\sigma_a(x) \land \sigma_b(y)).
\]

The following translations employ the three FO variables $x, y, z$, only.
This, in addition to the FO-completeness result for POTL,
proves that FOL on OP words retains the three-variable property, which holds in regular words.

\subsubsection{Next and Back Operators}
\[
  \nu_{\ldnext \varphi}(x) :=
    \exists y \Big(\lsucc(x, y) \land (x \lessdot y \lor x \doteq y) \land \exists x \big(x = y \land \nu_\varphi(x)\big)\Big)
\]
$\nu_{\ldback \varphi}(x)$ is defined similarly, and $\nu_{\lunext \varphi}(x)$ and $\nu_{\luback \varphi}(x)$
by replacing $\lessdot$ with $\gtrdot$.
\[
  \nu_{\lcdnext \varphi}(x) :=
    \exists y \big(\chain(x,y) \land (x \lessdot y \lor x \doteq y) \land \exists x (x = y \land \nu_\varphi(x))\big)
\]
$\nu_{\lcdback \varphi}(x)$, $\nu_{\lcunext \varphi}(x)$ and $\nu_{\lcuback \varphi}(x)$
are defined similarly.

\subsubsection{Downward/Upward Summary Until/Since}
The translation for the DS until operator can be obtained by noting that,
given two positions $x$ and $y$,
the DSP between them, if it exists, is the one that skips all chain bodies
entirely contained between them, among those with contexts
in the $\lessdot$ or $\doteq$ relations.
A position $z$ being part of such a path can be expressed with
formula $\neg \gamma(x,y,z)$ as follows:
\begin{align*}
  \gamma(x,y,z) &:=
    \gamma_L(x,z) \land \gamma_R(y,z) \\
  \gamma_L(x,z) &:=
    \exists y \Big(x \leq y \land y < z \land \exists x \big(z < x \land \chain(y,x) \land (y \lessdot x \lor y \doteq x)\big)\Big) \\
  \gamma_R(y,z) &:=
    \exists x \Big(z < x \land x \leq y \land \exists y \big(y < z \land \chain(y,x) \land (y \lessdot x \lor y \doteq x)\big)\Big)
\end{align*}
$\gamma(x,y,z)$ is true iff $z$ is not part of the DSP between $x$ and $y$, while $x \leq z \leq y$.
In particular, $\gamma_L(x,z)$ asserts that $z$ is part of the body of a chain whose left context is after $x$,
and $\gamma_R(y,z)$ states that $z$ is part of the body of a chain whose right context is before $y$.
Since chain bodies cannot cross, either the two chain bodies are actually the same one,
or one of them is a sub-chain nested into the other.
In both cases, $z$ is part of a chain body entirely contained between $x$ and $y$,
and is thus not part of the path.

Moreover, for such a path to exist, each one of its positions must be
in one of the admitted PRs with the next one.
Formula
\[
\delta(y,z) :=
  \exists x \big(z < x \land x \leq y
  \land (z \lessdot x \lor z \doteq x)
  \land \neg \gamma(z, y, x)
  \land (\lsucc(z, x) \lor \chain(z, x))\big)
\]
asserts this for position $z$, with the path ending in $y$.
(Note that by exchanging $x$ and $z$ in the definition of $\gamma(x,y,z)$ above,
one can obtain $\gamma(z, y, x)$ without using any additional variable.)
Finally, $\lcduntil{\varphi}{\psi}$ can be translated as follows:
\begin{align*}
  \nu_{\lcduntil{\varphi}{\psi}}(x) :=
    \exists y \Big(&x \leq y \land \exists x (x = y \land \nu_\psi(x)) \\
      &\land \forall z \big(x \leq z \land z < y \land \neg \gamma(x, y, z)
        \implies \exists x (x = z \land \nu_\varphi(x)) \land \delta(y,z) \big)\Big)
\end{align*}

The translation for the DS since operator is similar:
\begin{align*}
  \nu_{\lcdsince{\varphi}{\psi}}(x) :=
    \exists y \Big(&y \leq x \land \exists x (x = y \land \nu_\psi(x)) \\
      &\land \forall z \big(y < z \land z \leq x \land \neg \gamma(y, x, z)
        \implies \exists x (x = z \land \nu_\varphi(x)) \land \delta(x,z) \big)\Big)
\end{align*}
$\nu_{\lcuuntil{\varphi}{\psi}}(x)$ and $\nu_{\lcusince{\varphi}{\psi}}(x)$
are defined as above, substituting $\gtrdot$ for $\lessdot$.

\subsubsection{Hierarchical Operators}
Finally, below are the translations for two hierarchical operators,
the others being symmetric.
\begin{align*}
  \nu_{\lhunext \varphi}(x) :=
    &\exists y \Bigg(y < x \land \chain(y,x) \land y \lessdot x \land \\
        &\qquad \exists z \Big(x < z \land \chain(y,z) \land y \lessdot z \land \exists x (x = z \land \nu_\varphi(x)) \\
        &\qquad \quad \land \forall y \big(x < y \land y < z \implies \forall z (\chain(z,x) \land z \lessdot x \implies \neg \chain(z,y))\big)\Big)\Bigg)
\end{align*}
\begin{align*}
  \nu_{\lhuuntil{\varphi}{\psi}}(x) :=
    \exists z \bigg(&z < x \land z \lessdot x \land \chain(z,x) \land \\
      &\exists y \Big(x \leq y \land \chain(z,y) \land z \lessdot y \land \exists x (x = y \land \nu_\psi(x)) \land \\
        &\quad \forall z \big(x \leq z \land z < y \land \exists y (y < x \land y \lessdot x \land \chain(y,x) \land \chain(y,z)) \\
          &\qquad\qquad \implies \exists x (x = z \land \nu_\varphi(x))\big)\Big)\bigg)
\end{align*}

\subsection{Expressing \texorpdfstring{\xuntil{}}{Xuntil} in POTL}

To translate \xuntil{} to POTL, we give an isomorphism between OP words
and (a subset of) unranked ordered trees (UOT),
the structures on which \xuntil{} is defined.
First, we show how to translate OP words into UOTs, and then the reverse.

\subsubsection{OPM-compatible Unranked Ordered Trees}

\begin{defi}[Unranked Ordered Trees]
A UOT is a tuple $T = \langle S, \rchild, \rsibl, L \rangle$.
Each node is a sequence of child numbers, representing the path from the root to it.
$S$ is a finite set of finite sequences of natural numbers closed under the prefix operation,
and for any sequence $s \in S$,
if $s \cdot k \in S$, $k \in \mathbb{N}$, then either $k = 0$ or $s \cdot (k-1) \in S$
(by $\cdot$ we denote concatenation).
$\rchild$ and $\rsibl$ are two binary relations called the \emph{descendant} and
\emph{following sibling} relation, respectively.
For $s, t \in S$, $s \rchild t$ iff $t$ is any child of $s$
($t = s \cdot k$, $k \in \mathbb{N}$, i.e.\ $t$ is the $k$-th child of $s$),
and $s \rsibl t$ iff $t$ is the immediate sibling to the right of $s$
($s = r \cdot h$ and $t = r \cdot (h+1)$, for $r \in S$ and $h \in \mathbb{N}$).
$L \colon AP \rightarrow \powset{S}$ is a function that maps each
atomic proposition to the set of nodes labeled with it.
We denote as $\uotrees$ the set of all UOTs.
\end{defi}

Given an OP word $w = \langle U, <, M_{\powset{AP}}, P \rangle$,
it is possible to build a UOT
$T_w = \langle S_w, \allowbreak \rchild, \allowbreak \rsibl, \allowbreak L_w \rangle \in \uotrees$
with labels in $\powset{AP}$ isomorphic to $w$.
To do so, we define a function $\tau \colon U \rightarrow S_w$,
which maps positions of $w$ into nodes of $T_w$.
First, $\tau(0) = 0$: position 0 is the root node, and the last $\#$ is its rightmost child.
Given any position $i \in U$:
\begin{itemize}
\item
  if $i \doteq i+1$, then $\tau(i+1) = \tau(i) \cdot 0$ is the only child of $i$;
\item
  if $i \gtrdot i+1$, then $i$ has no children;
\item
  if $i \lessdot i+1$, then the leftmost child of $i$ is $i+1$
  ($\tau(i+1) = \tau(i) \cdot 0$).
  Moreover, if $j_1 < j_2 < \dots < j_n$ is the largest set of positions such that $\chain(i,j_k)$
  and either $i \lessdot j_k$ or $i \doteq j_k$ for $1 \leq k \leq n$,
  then $\tau(j_k) = \tau(i) \cdot k$.
\end{itemize}
In general, $i$ is in the $\lessdot$ relation with all of its children,
except possibly the rightmost one, with which $i$ may be in the $\doteq$ relation
(cf.\ property~\ref{item:upward-prop} of the $\chain$ relation).

This way, every position $i$ in $w$ appears in the UOT exactly once.
Indeed, if either $(i-1) \lessdot i$ or $(i-1) \doteq i$, then $i$ is a child of $i-1$.
Conversely, $(i-1) \gtrdot i$ iff $i$ is the right context of at least one chain.
Thus, consider $j$ s.t.\ $\chain(j,i)$, and for no $j' < j$ we have $\chain(j',i)$:
by property~\ref{item:downward-prop} of $\chain$, either $j \doteq i$ or $j \lessdot i$.
So, $i$ is a child of $j$ by the third rule above, and of no other node,
because if $(i-1) \gtrdot i$, then no other rule applies.

Finally, $\tau(i) \in L_w(\mathrm{a})$ iff $i \in P(\mathrm{a})$ for all $\mathrm{a} \in AP$,
so each node in $T_w$ is labeled with the set of atomic
propositions that hold in the corresponding word position.
We denote as $T_w = \tau(w)$ the UOT obtained by applying $\tau$
to every position of an OP word $w$.
Figure~\ref{fig:uotree-example} shows the translation of the word of Figure~\ref{fig:potl-example-word}
into an UOT\@.

\begin{figure}
\includegraphics{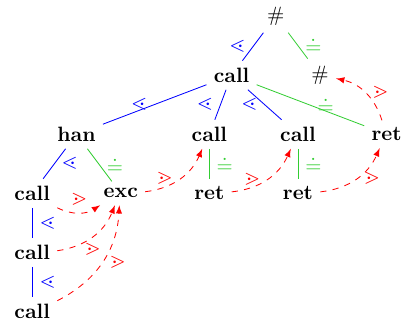}
\caption{The UOT corresponding to the word of Figure~\ref{fig:potl-example-word},
  and the ST of Figure~\ref{fig:sttree-example}.
  PRs between adjacent nodes are highlighted:
  $\mathord{\lessdot}$ is blue and $\mathord{\doteq}$ is green.
  A dashed red arrow connects each node to its $\ra$, if any.}%
\label{fig:uotree-example}
\end{figure}

As for the other way of the isomorphism, notice that we are considering
only a subset of UOTs.
In fact, we only consider UOTs whose node labels are compatible with a given OPM $M_{\powset{AP}}$.
In order to define the notion of OPM compatibility for UOTs,
we need to introduce the \emph{right context} ($\ra$) of a node.
Given a UOT $T$ and a node $s \in T$, the $\ra$ of $s$ is denoted $\ra(s)$.
If $s$ has a child $s'$ such that $s \doteq s'$, then $\ra(s)$ is undefined.
Otherwise, if $r$ is the leftmost right sibling of $s$, then $\ra(s) = r$;
if $s$ has no right siblings, then $\ra(s) = \ra(p)$, where $p$ is the parent of $s$.
In Figure~\ref{fig:uotree-example}, nodes are linked to their $\ra$ by a dashed red arrow.

In the following, for any nodes $s, s'$ and $\prf \in \{ \lessdot, \doteq, \gtrdot \}$,
we write $s \pr s'$ meaning that $a \pr b$,
where $a = \{\mathrm{p} \mid s \in L(\mathrm{p})\}$, and
$b = \{\mathrm{p} \mid s' \in L(\mathrm{p})\}$.
\begin{defi}[OPM-compatible UOTs]
We denote the set of UOTs compatible with an OPM $M$ as $\uotrees_M$.
A UOT $T$ is in $\uotrees_M$ iff the following properties hold.
The root node and its rightmost child are the only ones labeled with $\#$.
For any node $s \in T$, its rightmost child $r$, if any,
is such that either $s \lessdot r$ or $s \doteq r$.
For any other child $s' \neq r$ of $s$, we have $s \lessdot s'$.
If $\ra(s)$ exists, then $s \gtrdot \ra(s)$.
\end{defi}

Given a tree $T \in \uotrees_M$ with labels on $\powset{AP}$,
it is possible to build an OP word $w_T$ isomorphic to $T$.
Indeed,
\begin{lem}%
\label{lemma:tau-isomorphism}
Given an OP word $w$ and the UOT $T_w = \tau(w)$,
function $\tau$ is an isomorphism between positions of $w$ and nodes of $T_w$.
\end{lem}
\begin{proof}
We define function $\invtau : S \rightarrow \powset{AP}^+$,
which maps a UOT node to the subword corresponding to the subtree rooted in it.
For any node $s \in T$, let its label $a = \{\mathrm{p} \mid s \in L(\mathrm{p})\}$,
and let $c_0, c_1 \dots c_n$ be its children, if any.
$\invtau(s)$ is defined as $\invtau(s) = a$ if $s$ has no children,
and $\invtau(s) = a \cdot \invtau(c_0) \cdot \invtau(c_1) \cdots \invtau(c_n)$ otherwise.
We prove $\invtau(s)$ is an OP word.

We need to prove by induction on the tree structure
that for any tree node $s$, $\invtau(s)$ is of the form
$a_0 x_0 a_1 x_1 \dots a_n x_n$, with $n \geq 0$,
and such that for $0 \leq k < n$, $a_k \doteq a_{k+1}$
and either $x_k = \varepsilon$ or $\ochain{a_k}{x_k}{a_{k+1}}$.
In the following, we denote as $\first(x)$ the first position of a string $x$,
and as $\last(x)$ the last one.
Indeed, for each $0 \leq i < n$ we have $a \lessdot \first(\invtau(c_i))$, and
the rightmost leaf $f_i$ of the tree rooted in $c_i$ is such that $\ra(f_i) = c_{i+1}$.
Since $f_i = \tau(\last(\invtau(c_i)))$ and $c_{i+1} = \tau(\first(\invtau(c_{i+1})))$,
we have $\last(\invtau(c_i)) \gtrdot \first(\invtau(c_{i+1}))$.
So, $\ochain{a}{\invtau(c_i)}{\first(\invtau(c_{i+1}))}$.
As for $\invtau(c_n)$, if $a \lessdot c_n$ then $\invtau(s) = a_0 x_0$
(and $a_0 \lessdot \first(x_0)$),
with $a_0 = a$ and $x_0 = \invtau(c_0) \cdot \invtau(c_1) \cdots \invtau(c_n)$.
If $a \doteq c_n$, consider that, by hypothesis, $\invtau(c_n)$ is of the form $a_1 x_1 a_2 \dots a_n x_n$.
So $\invtau(s) = a_0 x_0 a_1 x_1 a_2 \dots a_n x_n$,
with $a_0 = a$ and $x_0 = \invtau(c_0) \cdot \invtau(c_1) \cdots \invtau(c_{n-1})$.

The root $0$ of $T$ and its rightmost child $c_\#$ are labeled with $\#$.
So, $\invtau(c_\#) = \#$,
and $\invtau(0) = \# x_0 \#$, with $\ochain{\#}{x_0}{\#}$,
which is a finite OP word.

$\tau^{-1} : S \rightarrow U$ can be derived from $\invtau$.
By construction, we have $\tau^{-1}(\tau(i)) = i$ for any word $w$ and position $i$.
\end{proof}

From Lemma~\ref{lemma:tau-isomorphism} follows:
\begin{prop}
Let $M_{AP}$ be an OPM on $\powset{AP}$.
For any FO formula $\varphi(x)$ on OP words compatible with $M_{AP}$,
there exists a FO formula $\varphi'(x)$ on trees in $\uotrees_{M_{AP}}$ such that for any
OP word $w$ and position $i$ in it, $w \models \varphi(i)$ iff $T_w \models \varphi'(\tau(i))$,
with $T_w = \tau(w)$.
\end{prop}

\subsubsection{POTL Translation of \texorpdfstring{\xuntil{}}{Xuntil}}%
\label{sec:cxpath-translation}

We now give the full translation of the logic \xuntil{} from~\cite{Marx2004} into POTL\@.

The syntax of \xuntil{} formulas is
\(
\varphi ::=
\mathrm{a} \mid
\top \mid
\neg \varphi \mid
\varphi \land \varphi \mid
\rho(\varphi, \varphi),
\)
with $\mathrm{a} \in AP$ and $\rho \in \{\Downarrow, \Uparrow, \Rightarrow, \Leftarrow\}$.
The semantics of propositional operators is the usual one,
while $\rho(\varphi, \varphi)$ is an until/since operator on the child and sibling relations.
Let $T \in \uotrees$ be a UOT with nodes in $S$.
For any $r, s \in S$, $\rparen$ and $\rlsibl$ are s.t.\ $r \rparen s$ iff $s \rchild r$,
and $r \rlsibl s$ iff $s \rsibl r$.
We denote as $R_\rho^+$ the transitive (but not reflexive) closure of relation $R_\rho$,
and by $R_\rho^*$ its transitive \emph{and} reflexive closure.
For $s \in S$,
$(T,s) \models \rho(\varphi, \psi)$ iff there exists a node $t \in S$ s.t.\ $s R_\rho^+ t$
and $(T,t) \models \psi$, and for any $r \in S$ s.t.\ $s R_\rho^+ r$ and $r R_\rho^+ t$ we have
$(T,r) \models \varphi$.
Notice that $t \neq s$ and $r \neq s$, so $s$ is not included in the paths:
we call this semantics \emph{strict}.
Conversely, in POTL paths always start from the position where an until/since operator is evaluated.

\cite{LibkinS10} proved the equivalence of \xuntil{} to the logic Conditional XPath, 
which was proved equivalent to FOL on finite UOTs in~\cite{Marx2005}.
This result is valid for any labeling of tree nodes, and so is on OPM-compatible UOTs.
\begin{thmC}[\cite{Marx2005,LibkinS10}]%
\label{thm:xuntil-fo-completeness}
Let $M_{AP}$ be an OPM on $AP$.
For any FO formula $\varphi(x)$ on trees in $\uotrees_{M_{AP}}$,
there exists a \xuntil{} formula $\varphi'$ such that,
for any $T \in \uotrees_{M_{AP}}$ and node $t \in T$, we have
$T \models \varphi(t)$ iff $(T,t) \models \varphi'$.
\end{thmC}

We define function $\iota_\mathcal{X}$, which translates any \xuntil{} formula $\varphi$
into a POTL formula s.t.\ $\varphi$ holds on a UOT $T$ iff $\iota_\mathcal{X}(\varphi)$
holds on the isomorphic word $w_T$.
$\iota_\mathcal{X}$ is defined as the identity for the propositional operators,
and with the equivalences below for the other \xuntil{} operators.
Recall from Section~\ref{sec:ltl-comp} that for any $a \subseteq AP$,
$\sigma_a := \bigwedge_{\mathrm{p} \in a} \mathrm{p} \land \bigwedge_{\mathrm{q} \not\in a} \neg \mathrm{q}$
holds in a position $i$ iff $a$ is the set of atomic propositions holding in $i$.
For any POTL formula $\gamma$, let
$\lganext{\lessdot} \gamma := \bigvee_{a,b \subseteq AP, \, a \lessdot b} (\sigma_a \land \lcdnext (\sigma_b \land \gamma))$
be the restriction of $\lcdnext \gamma$ to chains with contexts in the $\lessdot$ PR;\@
operators $\lganext{\doteq} \gamma$, $\lgaback{\lessdot} \gamma$, $\lgaback{\doteq} \gamma$,
$\lgnext{\lessdot} \gamma$, $\lgback{\lessdot} \gamma$ are defined analogously.

For any \xuntil{} formulas $\varphi, \psi$,
let $\varphi' = \iota_\mathcal{X}(\varphi)$ and $\psi' = \iota_\mathcal{X}(\psi)$.
We define $\iota_\mathcal{X}$ as follows:
\begin{align}
\iota_\mathcal{X}(\Downarrow(\varphi, \psi)) &:=
\ldnext (\lcduntil{\varphi'}{\psi'})
\lor \lcdnext (\lcduntil{\varphi'}{\psi'})\label{eq:iota-downarrow} \\
\iota_\mathcal{X}(\Uparrow(\varphi, \psi)) &:=
\ldback (\lcdsince{\varphi'}{\psi'})
\lor \lcdback (\lcdsince{\varphi'}{\psi'})\label{eq:iota-uparrow}
\end{align}
\begin{align}
\iota_\mathcal{X}(\Rightarrow(\varphi, \psi)) :=
&\lhunext (\lhuuntil{\varphi'}{\psi'})\label{eq:iota-rightarrow-1} \\
&\lor \big( \neg \lhunext (\lhuuntil{\top}{\neg \varphi'})
            \land \lgaback{\lessdot} (\lganext{\doteq} \psi')\big)%
      \label{eq:iota-rightarrow-2} \\
&\lor \lgback{\lessdot} \Big(
      \lganext{\lessdot} \big(\psi' \land \neg \lhuback (\lhusince{\top}{\neg \varphi'}) \big) \Big)%
      \label{eq:iota-rightarrow-3} \\
&\lor \lgback{\lessdot} (\lganext{\doteq} \psi' \land \neg \lganext{\lessdot} \neg \varphi')%
      \label{eq:iota-rightarrow-4}
\end{align}
\begin{align}
\iota_\mathcal{X}(\Leftarrow(\varphi, \psi)) :=
&\lhuback (\lhusince{\varphi'}{\psi'})%
      \label{eq:iota-leftarrow-1} \\
&\lor \lgaback{\doteq} \big(
      \lganext{\lessdot} (\neg \lhunext \top \land \lhusince{\varphi'}{\psi'}) \big)%
      \label{eq:iota-leftarrow-2} \\
&\lor \big( \lgaback{\lessdot} (\lgnext{\lessdot} \psi')
      \land \neg \lhuback (\lhusince{\top}{\neg \varphi'}) \big)%
      \label{eq:iota-leftarrow-3} \\
&\lor \lgaback{\doteq} (
      \lgnext{\lessdot} \psi' \land \neg \lganext{\lessdot} \neg \varphi' )%
      \label{eq:iota-leftarrow-4}
\end{align}

We prove the correctness of this translation in the following lemmas.
\begin{lem}%
\label{lemma:iota-downarrow}
  Given an OP alphabet $(AP, M_{AP})$,
  for every \xuntil{} formula $\Downarrow(\varphi, \psi)$,
  and for any OP word $w$ and position $i$ in $w$, we have
  \[
    (T_w,\tau(i)) \models \Downarrow(\varphi, \psi) \text{ iff }
    (w,i) \models \iota_\mathcal{X}(\Downarrow(\varphi, \psi)).
  \]
  $T_w \in \uotrees_{M_{AP}}$ is the UOT obtained by applying function $\tau$
  to every position in $w$, such that for any position $i'$ in $w$
  $(T_w,\tau(i')) \models \varphi \text{ iff } (w,i') \models \iota_\mathcal{X}(\varphi)$,
  and likewise for $\psi$.
\end{lem}
\begin{proof}
  Let $\varphi' = \iota_\mathcal{X}(\varphi)$ and $\psi' = \iota_\mathcal{X}(\psi)$.

  \textbf{[Only if]}
  Suppose $(T_w,\tau(i)) \models \Downarrow(\varphi, \psi)$.
  Let $r = \tau(i)$, and $s = \tau(j)$ s.t.\ $r \rchild s$
  and $s$ is the first tree node of the path witnessing $\Downarrow(\varphi, \psi)$.

  We inductively prove that $\lcduntil{\varphi'}{\psi'}$ holds in $j$.
  If $s$ is the last node of the path, then $\psi'$ holds in $j$ and so does, trivially,
  $\lcduntil{\varphi'}{\psi'}$.
  Otherwise, consider any node $t = \tau(k)$ in the path, except the last one,
  and suppose $\lcduntil{\varphi'}{\psi'}$ holds in $k'$ s.t.\ $t' = \tau(k')$
  is the next node in the path.
  If $t'$ is the leftmost child of $t$, then $k' = k+1$ and either $k \lessdot k'$ or $k \doteq k'$:
  in both cases $\ldnext (\lcduntil{\varphi'}{\psi'})$ holds in $k$.
  If $t'$ is not the leftmost child, then $\chain(k,k')$ and $k \lessdot k'$ or $k \doteq k'$:
  so $\lcdnext (\lcduntil{\varphi'}{\psi'})$ holds in $k$.
  Thus, by expansion law
  $\lcduntil{\varphi'}{\psi'} \equiv \psi'
  \lor (\varphi' \land (\ldnext (\lcduntil{\varphi'}{\psi'})
         \lor \lcdnext (\lcduntil{\varphi'}{\psi'})))$,
  $\lcduntil{\varphi'}{\psi'}$ holds in $k$ and, by induction, also in $j$.

  Suppose $s$ is the leftmost child of $r$: $j = i+1$, and either $i \lessdot j$ or $i \doteq j$,
  so $\ldnext (\lcduntil{\varphi'}{\psi'})$ holds in $i$.
  Otherwise, $\chain(i,j)$ and either $i \lessdot j$ or $i \doteq j$.
  In both cases, $\lcdnext (\lcduntil{\varphi'}{\psi'})$ holds in $i$.

  \textbf{[If]}
  Suppose~\eqref{eq:iota-downarrow} holds in $i$.
  If $\ldnext (\lcduntil{\varphi'}{\psi'})$ holds in $i$,
  then $\lcduntil{\varphi'}{\psi'}$ holds in $j = i+1$,
  and either $i \lessdot j$ or $i \doteq j$: then $s = \tau(j)$ is the leftmost child of $\tau(i)$.
  If $\lcdnext (\lcduntil{\varphi'}{\psi'})$ holds in $i$,
  then $\lcduntil{\varphi'}{\psi'}$ holds in $j$ s.t.\ $\chain(i,j)$
  and $i \lessdot j$ or $i \doteq j$: $s = \tau(j)$ is a child of $\tau(i)$ in this case as well.

  We prove that if $\lcduntil{\varphi'}{\psi'}$ holds in a position
  $j$ s.t.\ $\tau(i) \rchild \tau(j)$, then $\Downarrow(\varphi, \psi)$ holds in $\tau(i)$.
  If $\lcduntil{\varphi'}{\psi'}$ holds in $j$, then there exists a DSP of minimal length
  from $j$ to $h > j$ s.t.\ $(w,h) \models \psi'$ and $\varphi'$ holds
  in all positions $j \leq k < h$ of the path, and $(T_w,\tau(k)) \models \varphi$.
  In any such $k$,
  $\lcduntil{\varphi'}{\psi'} \equiv \psi'
  \lor (\varphi' \land (\ldnext (\lcduntil{\varphi'}{\psi'})
         \lor \lcdnext (\lcduntil{\varphi'}{\psi'})))$
  holds.
  Since this DSP is the minimal one, $\psi'$ does not hold in $k$.
  Either $\ldnext (\lcduntil{\varphi'}{\psi'})$
  or $\lcdnext (\lcduntil{\varphi'}{\psi'})$ hold in it.
  Therefore, the next position in the path is $k'$ s.t.\ either $k' = k+1$ or $\chain(k,k')$,
  and either $k \lessdot k'$ or $k \doteq k'$, and
  $(w,k') \models \lcduntil{\varphi'}{\psi'}$.
  Therefore, $\tau(k')$ is a child of $\tau(k)$.
  So, there is a sequence of nodes $s_0, s_1, \dots, s_n$ in $T_w$ s.t.\ $\tau(i) \rchild s_0$,
  and $s_i \rchild s_{i+1}$ and $(T_w,s_i) \models \varphi$ for $0 \leq i < n$,
  and $(T_w,s_n) \models \psi$.
  This is a path making $\Downarrow(\varphi, \psi)$ true in $\tau(i)$.
\end{proof}

The proof for $\iota_\mathcal{X}(\Uparrow(\varphi, \psi))$ (\ref{eq:iota-uparrow})
is analogous to Lemma~\ref{lemma:iota-downarrow}, and is therefore omitted.

\begin{lem}%
\label{lemma:iota-rightarrow}
  Given an OP alphabet $(AP, M_{AP})$,
  for every \xuntil{} formula $\Rightarrow(\varphi, \psi)$,
  and for any OP word $w$ and position $i$ in $w$, we have
  \[
    (T_w,\tau(i)) \models \Rightarrow(\varphi, \psi) \text{ iff }
    (w,i) \models \iota_\mathcal{X}(\Rightarrow(\varphi, \psi)).
  \]
  $T_w \in \uotrees_{M_{AP}}$ is the UOT obtained by applying function $\tau$
  to every position in $w$, such that for any position $i'$ in $w$
  $(T_w,\tau(i')) \models \varphi \text{ iff } (w,i') \models \iota_\mathcal{X}(\varphi)$,
  and likewise for $\psi$.
\end{lem}
\begin{proof}
  Let $\varphi' = \iota_\mathcal{X}(\varphi)$ and $\psi' = \iota_\mathcal{X}(\psi)$.

  \textbf{[Only if]}
  Suppose $\Rightarrow(\varphi, \psi)$ holds in $s = \tau(i)$.
  Then, node $r = \tau(h)$ s.t.\ $r \rchild s$ has at least two children,
  and $\Rightarrow(\varphi, \psi)$ is witnessed by a path starting in $t = \tau(j)$
  s.t.\ $s \rsibl t$, and ending in $v = \tau(k)$.
  We have the following cases:
  \begin{enumerate}
  \item\label{item:proof-iota-ra-1} \emph{$s$ is not the leftmost child of $r$.}
    \begin{enumerate}
    \item\label{item:proof-iota-ra-1a} \emph{$h \lessdot k$.}
      By the construction of $T_w$, for any node $t'$ in the path, there exists a position
      $j' \in w$ s.t.\ $t' = \tau(j')$, $\chain(h,j')$ and $h \lessdot j'$.
      The path made by such positions is a UHP, and $\lhuuntil{\varphi'}{\psi'}$ is true in $j$.
      Since $s$ is not the leftmost child of $r$, we have $\chain(h,i)$, and $h \lessdot i$,
      so~\eqref{eq:iota-rightarrow-1}
      holds in $i$.
    \item\label{item:proof-iota-ra-1b} \emph{$h \doteq k$, so $v$ is the rightmost child of $r$.}
      $\varphi$ holds in all siblings between $s$ and $v$ (excluded),
      and $\varphi'$ holds in the corresponding positions of $w$.
      All such positions $j$, if any, are s.t.\ $\chain(h,j)$ and $h \lessdot j$,
      and they form a UHP, so $\lhunext (\lhuuntil{\top}{\neg \varphi'})$
      never holds in $i$.
      Moreover, since $\psi$ holds in $v$, $\psi'$ holds in $k$.
      Note that $\lgaback{\lessdot}$ in $i$ uniquely identifies position $h$,
      and $\lganext{\doteq}$ evaluated in $h$ identifies $k$.
      So,~\eqref{eq:iota-rightarrow-2} holds in $i$.
    \end{enumerate}
  \item\label{item:proof-iota-ra-2} \emph{$s$ is the leftmost child of $r$.}
    In this case, we have $i = h+1$ and $h \lessdot i$
    (if $h \doteq i$, then $r$ would have only one child).
    \begin{enumerate}
    \item\label{item:proof-iota-ra-2a} \emph{$h \lessdot k$.}
      $\lgback{\lessdot}$ evaluated in $i$ identifies position $h$.
      $\psi'$ holds in $k$, and $\lhuback (\lhusince{\top}{\neg \varphi'})$ does not,
      because in all positions between $i$ and $k$ (excluded) corresponding to children of $r$,
      $\varphi'$ holds. Note that all such positions form a UHP, but $i$ is not part of it
      ($i = h+1$, so $\neg \chain(h,i)$), and is not considered by $\lhusince{\top}{\neg \varphi'}$.
      So,~\eqref{eq:iota-rightarrow-3} holds in $i$.
      \item\label{item:proof-iota-ra-2b} \emph{$h \doteq k$, so $v$ is the rightmost child of $r$.}
        $\psi$ holds in $v$, and $\varphi$ holds in all children of $r$,
        except possibly the first ($s$) and the last one ($v$).
        These are exactly all positions s.t.\ $\chain(h,j)$ and $h \lessdot j$.
        Since $\varphi'$ holds in all of them by hypothesis,
        $\neg \lganext{\lessdot} \neg \varphi'$ holds in $h$.
        Since $\psi$ holds in $v$, $\psi'$ holds in $k$, and $\lganext{\doteq} \psi'$ in $h$.
        So,~\eqref{eq:iota-rightarrow-4} holds in $i$.
    \end{enumerate}
  \end{enumerate}

\noindent
  \textbf{[If]}
  We separately consider cases~\eqref{eq:iota-rightarrow-1}--\eqref{eq:iota-rightarrow-4}.

  \noindent~\eqref{eq:iota-rightarrow-1}: $\lhunext (\lhuuntil{\varphi'}{\psi'})$
  holds in a position $i$ in $w$.
  Then, there exists a position $h$ s.t.\ $\chain(h,i)$ and $h \lessdot i$,
  and a position $j$ s.t.\ $\chain(h,j)$ and $h \lessdot j$
  that is the hierarchical successor of $i$, and $\lhuuntil{\varphi'}{\psi'}$ holds in $j$.
  So, $i$ and $j$ are consecutive children of $r = \tau(h)$.
  Moreover, there exists a UHP between $j$ and a position $k \geq j$.
  The tree nodes corresponding to all positions in the path are consecutive children of $r$,
  so we fall in case~\ref{item:proof-iota-ra-1a} of the \emph{only if} part of the proof.
  In $T_w$, a path between $t = \tau(j)$ and $v = \tau(k)$ witnesses the truth of
  $\Rightarrow(\varphi, \psi)$ in $s$.

  \noindent~\eqref{eq:iota-rightarrow-2}:
  $\neg \lhunext (\lhuuntil{\top}{\neg \varphi'}
   \land \lgaback{\lessdot} (\lganext{\doteq} \psi')$)
  holds in position $i \in w$ (this corresponds to case~\ref{item:proof-iota-ra-1b}).
  If $\lgaback{\lessdot} (\lganext{\doteq} \psi')$  holds in $i$,
  then there exists a position $h$ s.t.\ $\chain(h,i)$ and $h \lessdot i$,
  and a position $k$ s.t.\ $\chain(h,k)$ and $h \doteq k$, and $\psi'$ holds in $k$.
  $v = \tau(k)$ is the rightmost child of $r = \tau(h)$, parent of $s = \tau(i)$.
  Moreover, if $\neg \lhunext (\lhuuntil{\top}{\neg \varphi'})$ holds in $i$, then either:
  \begin{itemize}
  \item $\neg \lhunext \top$ holds, i.e.\ there is no position $j > i$
    s.t.\ $\chain(h,j)$ and $h \lessdot j$,
    so $v$ is the immediate right sibling of $s$.
    In this case $\Rightarrow(\varphi, \psi)$ holds in $s$ because $\psi$ holds in $v$.
  \item $\neg (\lhuuntil{\top}{\neg \varphi'})$ holds in $j > i$,
    the first position after $i$ s.t.\ $\chain(h,j)$ and $h \lessdot j$.
    This means $\varphi'$ holds in all positions $j' \geq j$
    s.t.\ $\chain(h,j')$ and $h \lessdot j'$.
    Consequently, the tree nodes corresponding to these positions plus $v = \tau(k)$ form
    a path witnessing $\Rightarrow(\varphi, \psi)$, which holds in $s = \tau(i)$.
  \end{itemize}

  \noindent~\eqref{eq:iota-rightarrow-3}:
  $\lgback{\lessdot} \Big(
      \lganext{\lessdot} \big(\psi' \land
      \neg \lhuback (\lhusince{\top}{\neg \varphi'}) \big) \Big)$
  holds in $i$.
  Let $h = i-1$, with $h \lessdot i$ (it exists because $\lgback{\lessdot}$ is true).
  There exists a position $k$, $\chain(h,k)$ and $h \lessdot k$,
  in which $\psi'$ holds, so $\psi$ does in $v = \tau(k)$,
  and $\lhuback (\lhusince{\top}{\neg \varphi'})$ is false in it.
  If it is false because $\neg \lhuback \top$ holds,
  there is no position $j < k$ s.t.\ $\chain(h,j)$ and $h \lessdot j$,
  so $v$ is the second child of $r = \tau(h)$, $s = \tau(i)$ being the first one.
  So, $\Rightarrow(\varphi, \psi)$ trivially holds in $s$ because $\psi$ holds in the next sibling.
  Otherwise, let $j < k$ be the rightmost position lower than $k$
  s.t.\ $\chain(h,j)$ and $h \lessdot j$.
  $\neg (\lhusince{\top}{\neg \varphi'})$ holds in it, so $\varphi'$ holds
  in all positions $j'$ between $i$ and $k$ that are part of the hierarchical path,
  i.e.\ s.t.\ $\chain(h,j')$ and $h \lessdot j'$.
  The corresponding tree nodes form a path ending in $v = \tau(k)$
  that witnesses the truth of $\Rightarrow(\varphi, \psi)$ in $s$
  (case~\ref{item:proof-iota-ra-2a}).

  \noindent~\eqref{eq:iota-rightarrow-4}:
  $\lgback{\lessdot} (\lganext{\doteq} \psi' \land \neg \lganext{\lessdot} \neg \varphi')$
  holds in $i$. Then let $h = i-1$, $h \lessdot i$, and $s = \tau(i)$ is the leftmost
  child of $r = \tau(h)$.
  Since $\lganext{\doteq} \psi'$ holds in $h$, there exists a position $k$,
  s.t.\ $\chain(h,k)$ and $h \doteq k$, in which $\psi'$ holds.
  So, $\psi$ holds in $v = \tau(k)$, which is the rightmost child of $r$, by construction.
  Moreover, $\varphi'$ holds in all positions s.t.\ $\chain(h,j)$ and $h \lessdot j$.
  Hence, $\varphi$ holds in all corresponding nodes $t = \tau(j)$,
  which are all nodes between $s$ and $v$, excluded.
  This, together with $\psi$ holding in $v$, makes a path that verifies
  $\Rightarrow(\varphi, \psi)$ in $s$
  (case~\ref{item:proof-iota-ra-2b}).
\end{proof}

\begin{lem}%
\label{lemma:iota-leftarrow}
  Given an OP alphabet $(AP, M_{AP})$,
  for every \xuntil{} formula $\Leftarrow(\varphi, \psi)$,
  and for any OP word $w$ and position $i$ in $w$, we have
  \[
    (T_w,\tau(i)) \models \Leftarrow(\varphi, \psi) \text{ iff }
    (w,i) \models \iota_\mathcal{X}(\Leftarrow(\varphi, \psi)).
  \]
  $T_w \in \uotrees_{M_{AP}}$ is the UOT obtained by applying function $\tau$
  to every position in $w$, such that for any position $i'$ in $w$
  $(T_w,\tau(i')) \models \varphi \text{ iff } (w,i') \models \iota_\mathcal{X}(\varphi)$,
  and likewise for $\psi$.
\end{lem}
\begin{proof}
  Let $\varphi' = \iota_\mathcal{X}(\varphi)$ and $\psi' = \iota_\mathcal{X}(\psi)$.

  \textbf{[Only if]}
  Suppose $\Leftarrow(\varphi, \psi)$ holds in $s = \tau(i)$.
  Then node $r = \tau(h)$ s.t.\ $r \rchild s$ has at least two children,
  and $\Leftarrow(\varphi, \psi)$ is true because of a path starting in $v = \tau(k)$,
  s.t.\ $r \rchild v$ and $(T_w,v) \models \psi$
  and ending in $t = \tau(j)$ s.t.\ $t \rsibl s$.
  We distinguish between the following cases:
  \begin{enumerate}
  \item \emph{$v$ is not the leftmost child of $r$.}
    \begin{enumerate}
    \item\label{item:proof-iota-la-1a} \emph{$h \lessdot i$.}
      By construction, all nodes in the path correspond to positions $j' \in w$
      s.t.\ $\chain(h,j')$ and $h \lessdot j'$, so they form a UHP\@.
      Hence, $\lhusince{\varphi'}{\psi'}$ holds in $j$, and~\eqref{eq:iota-leftarrow-1}
      holds in $i$.
    \item\label{item:proof-iota-la-1b} \emph{$h \doteq i$.}
      In this case, $s$ is the rightmost child of $r$, and $\chain(h,i)$.
      The path made of positions between $k$ and $j$ corresponding to nodes between $v$ and $t$
      (included) is a UHP\@.
      So $\lhusince{\varphi'}{\psi'}$ holds in $j$, which is the rightmost position of any
      possible such UHP:\@ so $\neg \lhunext \top$ also holds in $j$.
      Hence,~\eqref{eq:iota-leftarrow-2}
      holds in $i$.
    \end{enumerate}
  \item \emph{$v$ is the leftmost child of $r$.}
    \begin{enumerate}
    \item\label{item:proof-iota-la-2a} \emph{$h \lessdot i$.}
      In this case, $k = h+1$ and $\psi'$ holds in $k$.
      So, $\lgnext{\lessdot} \psi'$ holds in $h$,
      and $\lgaback{\lessdot} (\lgnext{\lessdot} \psi')$ holds in $i$.
      Moreover, in all word positions $j'$ with $k < j' < j$ corresponding to children of $r$,
      $\varphi'$ holds.
      Such positions form a UHP\@.
      So $\neg \lhuback (\lhusince{\top}{\neg \varphi'})$ holds in $i$.
      Note that this is true even if $s$ is the first right sibling of $v$.
      Thus,~\eqref{eq:iota-leftarrow-3} holds in $i$.
    \item\label{item:proof-iota-la-2b} \emph{$h \doteq i$.}
      $\psi'$ holds in $k = h+1$, so $\lgnext{\lessdot} \psi'$ holds in $h$.
      Since $\chain(h,i)$ and $h \doteq i$,
      $\lgaback{\doteq} (\lgnext{\lessdot} \psi')$ holds in $i$.
      Moreover, $\varphi$ holds in all children of $r$ except the first and last one,
      so $\varphi'$ holds in all positions $j'$ s.t.\ $\chain(h,j')$ and $h \lessdot j'$.
      So $\neg \lganext{\lessdot} \neg \varphi'$ holds in $h$, and~\eqref{eq:iota-leftarrow-4}
      in $i$.
    \end{enumerate}
  \end{enumerate}

\noindent
  \textbf{[If]}
  We separately consider cases~\eqref{eq:iota-leftarrow-1}--\eqref{eq:iota-leftarrow-4}.

  \noindent~\eqref{eq:iota-leftarrow-1}:
  $\lhuback (\lhusince{\varphi'}{\psi'})$ holds in $i$.
  Then, there exists a position $h$ s.t.\ $\chain(h,i)$ and $h \lessdot i$,
  and a position $j < i$ s.t.\ $\chain(h,j)$ and $h \lessdot j$.
  Since $j \neq h+1$, the corresponding tree node is not the leftmost one.
  So, this corresponds to case~\ref{item:proof-iota-la-1a},
  and $\Leftarrow(\varphi, \psi)$ holds in $s = \tau(i)$.

  \noindent~\eqref{eq:iota-leftarrow-2}:
  $\lgaback{\doteq} \big(\lganext{\lessdot} (\neg \lhunext \top \land \lhusince{\varphi'}{\psi'}) \big)$
  holds in $i$.
  Then, there exists a position $h$ s.t.\ $\chain(h,i)$ and $h \doteq i$.
  Moreover, at least a position $j'$ s.t.\ $\chain(h,j')$ and $h \lessdot j'$ exists.
  Let $j$ be the rightmost one, i.e.\ the only one in which $\neg \lhunext \top$ holds.
  The corresponding tree node $t = \tau(j)$ is s.t.\ $t \rsibl s$, with $s = \tau(i)$.
  Since $\lhusince{\varphi'}{\psi'}$ holds in $j$, a UHP starts from it,
  and $\psi$ and $\varphi$ hold in the tree nodes corresponding to, respectively,
  the first and all other positions in the path.
  This is case~\ref{item:proof-iota-la-1b}, and $\Leftarrow(\varphi, \psi)$ holds in $s$.

  \noindent~\eqref{eq:iota-leftarrow-3}:
  $\lgaback{\lessdot} (\lgnext{\lessdot} \psi')
    \land \neg \lhuback (\lhusince{\top}{\neg \varphi'})$
  holds in $i$.
  Then, there exists a position $h$ s.t.\ $\chain(h,i)$ and $h \lessdot i$.
  $\psi'$ holds in $k = h+1$, so $\psi$ holds in the leftmost child of $r = \tau(h)$.
  Moreover, $\varphi'$ holds in all positions $j' < i$ s.t.\ $\chain(h,j')$ and $h \lessdot j'$,
  so $\varphi$ holds in all children of $r$ between $v = \tau(k)$ and $s = \tau(i)$, excluded.
  This is case~\ref{item:proof-iota-la-2a}, and $\Leftarrow(\varphi, \psi)$ holds in $s$.

  \noindent~\eqref{eq:iota-leftarrow-4}:
  $\lgaback{\doteq} (\lgnext{\lessdot} \psi' \land \neg \lganext{\lessdot} \neg \varphi' )$
  holds in $i$.
  Then, there exists a position $h$ s.t.\ $\chain(h,i)$ and $h \doteq i$.
  $\lgnext{\lessdot} \psi'$ holds in $h$, so $\psi$ holds in node $v = \tau(h+1)$,
  which is the leftmost child of $r = \tau(h)$.
  Since $\neg \lganext{\lessdot} \neg \varphi'$ holds in $h$,
  $\psi'$ holds in all positions $j'$ s.t.\ $\chain(h,j')$ and $h \lessdot j$.
  So, $\psi$ holds in all children of $r$ except (possibly) the leftmost ($v$)
  and the rightmost ($s = \tau(i)$) ones.
  This is case~\ref{item:proof-iota-la-2b}, and $\Leftarrow(\varphi, \psi)$ holds in $s$.
\end{proof}

It is possible to express all POTL operators in FOL, as per Section~\ref{sec:fo-semantics}.
From this, and Lemmas~\ref{lemma:iota-downarrow},~\ref{lemma:iota-rightarrow}, and~\ref{lemma:iota-leftarrow}
together with Theorem~\ref{thm:xuntil-fo-completeness}, we derive
\begin{thm}%
\label{thm:potl-completeness}
POTL = FOL with one free variable on finite OP words.
\end{thm}
\begin{cor}%
\label{cor:complete-subset}
The propositional operators plus
\(
  \ldnext,
  \ldback,
  \lcdnext,
  \lcdback,
  \lcduntil{}{},
  \lcdsince{}{},
  \lhunext, \allowbreak
  \lhuback, \allowbreak
  \lhuuntil{}{}, \allowbreak
  \lhusince{}{}
\)
are expressively complete on OP words.
\end{cor}
\begin{cor}%
\label{cor:optl-in-potl}
NWTL $\subset$ OPTL $\subset$ POTL over finite OP words.
\end{cor}
\begin{cor}%
\label{cor:3var}
Every FO formula with at most one free variable is equivalent
to one using at most three distinct variables on finite OP words.
\end{cor}

Corollary~\ref{cor:complete-subset} follows from the definition of $\iota_\mathcal{X}$
and Theorem~\ref{thm:potl-completeness}
(note that all other operators are shortcuts for formulas expressible with those listed).
In Corollary~\ref{cor:optl-in-potl}, NWTL $\subset$ OPTL was proved in~\cite{ChiariMP20a},
and OPTL $\subseteq$ POTL comes from Theorem~\ref{thm:potl-completeness} and
the semantics of OPTL being expressible in FOL similarly to POTL,
while OPTL $\subsetneq$ POTL comes from Theorem~\ref{thm:optl-vs-potl}.
Corollary~\ref{cor:3var}, stating that OP words have the three-variable property,
follows from the FOL semantics of POTL being expressible with just three variables.

\section{First-Order Completeness on \texorpdfstring{$\omega$}{Omega}-Words}%
\label{sec:omega-fo-completeness}

To prove the FO-completeness of the translation of \xuntil{} into POTL
also on OP $\omega$-words, we must prove that \xuntil{} is FO-complete on the OPM-compatible
UOTs resulting from $\omega$-words.
In Section~\ref{sec:opm-comp-omega-trees} we show that infinite OPM-compatible UOTs
can be divided in two classes, depending on their shape.
Then, after introducing some new notation in Section~\ref{sec:omega-notation}, we show how to translate a given FO formula into \xuntil{} separately for each UOT class
in Sections~\ref{sec:rr-trees} and~\ref{sec:lr-trees},
and how such translations can be combined to work on any infinite OPM-compatible UOT
in Section~\ref{sec:synthesis}.
Our proofs exploit composition arguments on trees from~\cite{Libkin09,HaferT87,MollerR99} but introduce new techniques to deal with the peculiarities of UOTs derived from $\omega$-OPLs.

\subsection{OPM-compatible \texorpdfstring{$\omega$}{omega}-UOTs}%
\label{sec:opm-comp-omega-trees}
The application of function $\tau$ from Section~\ref{sec:fo-completeness}
to OP $\omega$-words results in two classes of infinite UOTs,
depending on the shape of the underlying ST\@.
In both cases, in the UOT the rightmost child of the root is not labeled with $\#$,
and nodes in the rightmost branch do not have a \emph{right context}.
$\tau^{-1}$, which can be defined in the same way, converts such UOTs into words
with open chains.

\begin{figure}
\centering
\begin{tabular}{m{0.5\textwidth} m{0.3\textwidth}}
\centering
\includegraphics{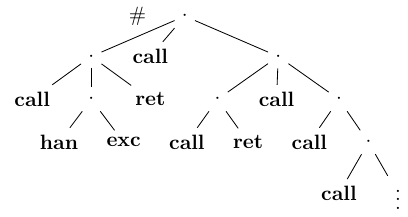}
&
\centering
\includegraphics{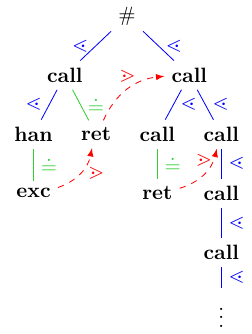}
\end{tabular}
\caption{ST (left) and UOT (right) of the RR OP $\omega$-word
  $\# \; \lcall \; \lhandle \; \lthrow \; \lret \; \allowbreak \lcall \; \allowbreak \lcall \; \lret \; \lcall \; \lcall \; \lcall \; \dots$
  on OPM $M_\lcall$ (Figure~\ref{fig:opm-mcall}).}%
\label{fig:rr-example}
\end{figure}

If a word $w$ reaches infinity through right recursion, then it contains an infinite
number of chains that have a left context, which we call a \emph{pending} position,
but no right context.
An $\omega$OPBA reading such a word does an infinite number of push moves,
and its stack grows to infinity.
The corresponding UOT $T_w = \tau(w)$
presents a single infinite branch, made of the rightmost nodes of each level.
Such nodes, which we call \emph{pending}, are in the $\lessdot$ PR
when they correspond to a right recursion step,
and in the $\doteq$ PR when they are siblings in the ST\@.
Pending nodes may have left non-terminal (dot) siblings, which correspond to bodies of inner chains
between two consecutive right-recursion steps.
So, $\omega$-words in which left and right recursion are alternated also
fall into this class, which we call right-recursive (RR) UOTs (Figure~\ref{fig:rr-example}).

\begin{figure}
\centering
\begin{tabular}{m{0.45\textwidth} m{0.5\textwidth}}
\centering
\includegraphics{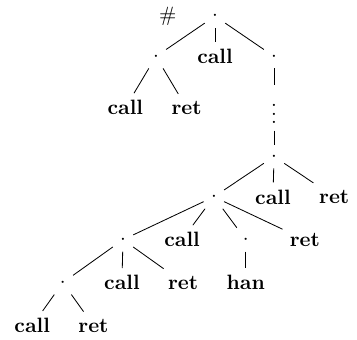}
&
\centering
\includegraphics{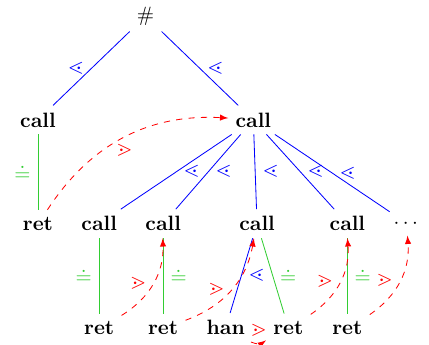}
\end{tabular}
\caption{ST (left) and UOT (right) of the LR OP $\omega$-word
  $\# \; \lcall \; \lret \; \lcall \; \allowbreak \lcall \; \allowbreak \lret \; \allowbreak \lcall \; \lret \; \lcall \; \lhandle \; \lret \; \lcall \; \lret \; \dots$
  on OPM $M_\lcall$ (Figure~\ref{fig:opm-mcall}).}%
\label{fig:lr-example}
\end{figure}

If $w$ reaches infinity by left recursion, then it contains an infinite number of chains
sharing the same left context.
An $\omega$OPBA reading $w$ performs an infinite sequence of pop moves, each one followed by a push,
and its stack size is ``ultimately bounded'', i.e., the stack symbol related to such a left context remains in the stack indefinitely,
and other symbols are repeatedly pushed on top of it, and popped.
Thus, the stack reaches the same size infinitely many times.
The rightmost branch of $T_w = \tau(w)$ ends with a node $r_\infty$
with an infinite number of children (cf.\ Figure~\ref{fig:lr-example}).
Node $r_\infty$ is in the $\lessdot$ relation with all of its children,
otherwise it would violate property~\ref{item:downward-prop} of the $\chain$ relation.
This is the left-recursive (LR) class of UOTs.

An exception to the above classification may occur if the OPM is such that the transitive closure of the $\doteq$ relation is reflexive ---in other words the OPM contains $\doteq$-circularities---.
In this case the ST \emph{may} contain just one node with an infinite number of children,
all in the $\doteq$ PR\@.
As a result, such nodes form \emph{a unique infinite branch} in the corresponding UOT whose nodes are in the $\doteq$ relation unlike the case of Figure~\ref{fig:rr-example}.
This is the distinguishing feature of RR UOTs, despite the fact that in this case the stack of the $\omega$OPBA remains ``ultimately bounded'' as in the case of LR UOTs. Thus, this exceptional case is attributed to the RR class.

In the following, by RR (resp.\ LR) word or ST we mean an OP $\omega$-word or ST that translates to a RR (resp.\ LR) UOT\@.
Next, we separately give a translation of FOL into \xuntil{} for RR and LR UOTs,
and then show how to combine them to obtain completeness.

\begin{rems}
There cannot be RR UOTs containing a node with infinitely many children,
or LR UOTs where more than one node has infinite children.
If this was the case, then the infinite children would appear as consecutive
infinite subsets of positions in the OP $\omega$-word isomorphic to the UOT\@.
But this is impossible, because the set of word positions is $\mathbb{N}$,
which is not dense and does not contain dense subsets.

The FO-completeness proof of the logic NWTL~\cite{lmcs/AlurABEIL08}
is also based on a translation of Nested Words to UOTs.
However, Nested Words result in only one kind of UOT,
because VPL grammars can be transformed so that words grow in only one direction.
Thus, that proof does not deal with the issue of combining two separate translations.
\end{rems}

\subsection{Notation}\label{sec:omega-notation}
Given a FO formula $\varphi$, we call its \emph{quantifier rank} (q.r.)
the maximum nesting level of its quantifiers.
Let $\mathcal{M}$ be a structure on a relational signature $\langle D, R_1, \dots, R_n \rangle$,
with domain $D$ and relations $R_1, \dots, R_n$.
The \emph{rank-$k$ type} of $\mathcal{M}$ is the set
\[
\sigma_k(\mathcal{M}) =
\{\varphi \mid \text{$\varphi \in$ FOL, $\mathcal{M} \models \varphi$
  and the q.r.\ of $\varphi$ is $k$}\},
\]
while the rank-$k$ type of $\mathcal{M}$ with a distinguished element $d \in D$ is
\[
\sigma_k(\mathcal{M}, d) =
\{\varphi(x) \mid \text{$\varphi(x) \in$ FOL, $(\mathcal{M}, d) \models \varphi(x)$
  and the q.r.\ of $\varphi(x)$ is $k$}\}.
\]

Since the set of nonequivalent FO formulas with at most $k$ quantifiers on a relational signature is finite,
there are only finitely many rank-$k$ types.
For each rank-$k$ type $\sigma_k$ it is possible to define the Hintikka
formula $H_{\sigma_k}$~\cite{Hintikka53}, s.t.\ $\mathcal{M} \models H_{\sigma_k}$
iff the rank-$k$ type of $\mathcal{M}$ is $\sigma_k$.

The compositional argument used here is based on Ehrenfeucht-Fra\"{\i}ss\'e (EF) games
(see e.g.,~\cite{GradelKLMSVVW07,Immerman12})
between two players, $\forall$ (the \emph{Spoiler}) and $\exists$ (the \emph{Duplicator}).
A round of an EF game between two structures $\mathcal{M}$ and $\mathcal{M}'$ with the same signature
starts with $\forall$ picking an element of the domain
of either one of $\mathcal{M}$ and $\mathcal{M}'$,
followed by $\exists$ answering by picking an element of the other structure.
$\exists$ wins the $k$-round game on two structures if the map
between the elements picked by $\forall$ and those picked by $\exists$ in each of the first $k$ rounds
is a partial isomorphism between $\mathcal{M}$ and $\mathcal{M}'$.
We write $\mathcal{M} \sim_k \mathcal{M}'$ if $\exists$ has a winning strategy for the
$k$-round game between $\mathcal{M}$ and $\mathcal{M}'$ and,
given $a \in D$ and $a' \in D'$, we write $(\mathcal{M}, a) \sim_k (\mathcal{M}', a')$
iff $\exists$ wins the game on $\mathcal{M}$ and $\mathcal{M}'$ in whose first round $\forall$
picks $a$, and $\exists$ answers with $a'$.
We write $\mathcal{M} \equiv_k \mathcal{M}'$ to state that $\mathcal{M}$ and $\mathcal{M}'$
have the same rank-$k$ type, i.e.\ they satisfy exactly the same FO formulas of q.r.\ at most $k$.
By the EF Theorem, $\mathcal{M} \sim_k \mathcal{M}'$ iff $\mathcal{M} \equiv_k \mathcal{M}'$, for all $k$.

We refer to the syntax and semantics of \xuntil{} presented in Section~\ref{sec:cxpath-translation}.
The semantics of the until and since operators is strict,
i.e.\ the position where they are evaluated is not part of their paths,
which start with the next one.
Thus, next and back operators are not needed, but we define them as shortcuts,
for any \xuntil{} formula $\varphi$:
\begin{align*}
\xcnext \varphi &:= \xcuntil{\neg \top}{\varphi} &
\xcback \varphi &:= \xcsince{\neg \top}{\varphi} &
\xsnext \varphi &:= \xsuntil{\neg \top}{\varphi} &
\xsback \varphi &:= \xssince{\neg \top}{\varphi}
\end{align*}

We call the structure $\langle U, <, P \rangle$ a finite LTL word if $U \subseteq \mathbb{N}$
is finite, and an LTL $\omega$-word if $U = \mathbb{N}$.
$<$ is a linear order on $U$, and $P \colon AP \to \powset{U}$ is a labeling function.
The syntax of an LTL formula $\varphi$ is, for $\mathrm{a} \in AP$,
\(
\varphi ::= \mathrm{a}
\mid \neg \varphi
\mid \varphi \lor \varphi
\mid \lluntil{\varphi}{\varphi}
\mid \llsince{\varphi}{\varphi}
\),
where propositional operators have the usual meaning.
Given an LTL word $w$, and a position $i$ in it,
\begin{itemize}
\item $(w, i) \models \mathrm{a}$ iff $i \in P(\mathrm{a})$;
\item $(w, i) \models \lluntil{\psi}{\theta}$
  iff there is $j > i$ s.t.\ $(w, j) \models \theta$ and
  for any $i < j' < j$ we have $(w, j') \models \psi$;
\item $(w, i) \models \llsince{\psi}{\theta}$
  iff there is $j < i$ s.t.\ $(w, j) \models \theta$ and
  for any $j < j' < i$ we have $(w, j') \models \psi$.
\end{itemize}
A \emph{future} LTL formula only contains the $\lluntil{}{}$ modality,
and a \emph{past} one only contains $\llsince{}{}$.

\subsection{RR UOTs}%
\label{sec:rr-trees}

\begin{figure}[bt]
  \centering
  \includegraphics{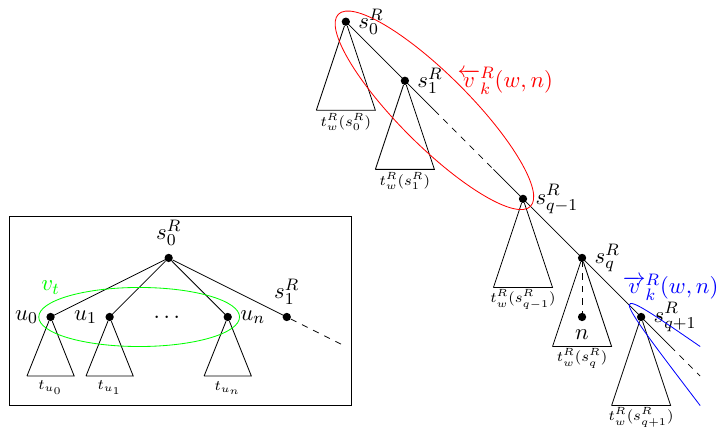}
  \caption{Parts in which we divide a RR UOT for Lemma~\ref{lemma:composition-rr} (right)
    and Lemma~\ref{lemma:composition-subtree} (left).}%
  \label{fig:rr-tree-comp}
\end{figure}

We prove the following:
\begin{lem}%
\label{lemma:transl-rr}
Given a FO formula on UOTs $\bar{\varphi}(x)$ of q.r.\ $k \geq 1$,
there exists a \xuntil{} formula $\varphi_R$
s.t.\ for any OP $\omega$-word $w$ s.t.\ $T_w = \tau(w)$ is a RR UOT,
and for any node $n \in T_w$ we have
$(T_w, n) \models \bar{\varphi}(x)$ iff $(T_w, n) \models \varphi_R$.
\end{lem}

To prove Lemma~\ref{lemma:transl-rr}, we show that $\varphi_R$
can be built by combining formulas
describing only parts of a UOT, as hinted in Figure~\ref{fig:rr-tree-comp}.
First, we prove that the rank-$k$ type of such parts determines
the rank-$k$ type of the whole tree (Lemma~\ref{lemma:composition-rr});
then, since such subdivision includes finite subtrees, we show how to express
their rank-$k$ types in \xuntil{} (Lemma~\ref{lemma:composition-subtree});
finally we use such results to translate $\bar{\varphi}(x)$.

This part of the proof partially resembles the one for the FO-completeness of NWTL
in~\cite{lmcs/AlurABEIL08}, because of the similarity between the shape of RR UOTs
and those resulting from nested $\omega$-words.
However, the two proofs diverge significantly, because Nested Words are isomorphic to binary trees,
while RR UOTs are unranked.

Consider the RR UOT $T_w$ of the statement.
We denote with $s^R_0, s^R_1, \dots$ the infinite sequence of \emph{pending} nodes obtained
by starting from the root of $T_w$, and always descending through the rightmost child.
We call $t^R_w(s^R_p)$ the finite subtree obtained by removing the rightmost child of $s^R_p$
and its descendants from the subtree rooted in $s^R_p$.
If $s^R_p$ has one single child, $t^R_w(s^R_p)$ is made of $s^R_p$ only.
Let $n$ be any node in $T_w$, and let $s^R_q$ be s.t.\ $n$ is part of $t^R_w(s^R_q)$,
and $s^R_q = n$ if $n$ is a pending node.
Let $\Gamma_k$ be the (finite) set of all rank-$k$ types of finite UOTs,
and $\sigma^R_k(w, n) \in \Gamma_k$ be the rank-$k$ type of $t^R_w(s^R_q)$.
We define $\vla^R_k(w, n)$ as a finite LTL word of length $q$ on alphabet $\Gamma_k$,
s.t.\ each position $p$, $0 \leq p \leq q-1$, is labeled with $\sigma^R_k(w, s^R_p)$.
Also, let $\vra^R_k(w, n)$ be a LTL $\omega$-word on $\Gamma_k$ having each position
labeled with $\sigma^R_k(w, s^R_{q+j+1})$ for all $j \geq 0$.
We give the following compositional argument:
\begin{lem}%
\label{lemma:composition-rr}
Let $w_1$ and $w_2$ be two OP $\omega$-words, such that $T_{w_1} = \tau(w_1)$
and $T_{w_2} = \tau(w_2)$ are two RR UOTs.
Let $i_1$ and $i_2$ be two positions in, resp., $w_1$ and $w_2$,
and let $s_1 = \tau(i_1)$ and $s_2 = \tau(i_2)$.
For any $k \geq 1$, if
\begin{enumerate}
\item\label{item:foc-rr-vla}
  $\vla^R_k(w_1, s_1) \equiv_k \vla^R_k(w_2, s_2)$,
\item\label{item:foc-rr-vra}
  $\vra^R_k(w_1, s_1) \equiv_k \vra^R_k(w_2, s_2)$ and
\item\label{item:foc-rr-sametype}
  $\sigma^R_k(w_1, s_1) = \sigma^R_k(w_2, s_2)$,
\end{enumerate}
then $(T_{w_1}, s_1) \equiv_k (T_{w_2}, s_2)$.
\end{lem}

\begin{proof}
By the EF Theorem, equivalences~\ref{item:foc-rr-vla}--\ref{item:foc-rr-vra}--\ref{item:foc-rr-sametype}
imply that the respective games have a winning strategy; we refer to them by game~\ref{item:foc-rr-vla},~\ref{item:foc-rr-vra} and~\ref{item:foc-rr-sametype}.
We prove that $(T_{w_1}, s_1) \sim_k (T_{w_2}, s_2)$,
i.e.\ that $\exists$ has a winning strategy in the EF game on $(T_{w_1}, s_1)$ and $(T_{w_2}, s_2)$ so that the thesis follows.
In round 0 of the game, partial isomorphism is ensured by $s_1$ and $s_2$ having the same labels,
due to hypothesis~\ref{item:foc-rr-sametype}.
In any subsequent round, suppose w.l.o.g.\ that $\forall$ picks a node $s^\forall$ from $T_{w_1}$
(the converse is symmetric).
Let $s^R_{q_1}$ be the pending node s.t.\ $s_1$ is part of $t^R_w(s^R_{q_1})$,
and $s^R_{q_2}$ be the pending node s.t.\ $s_2$ is part of $t^R_w(s^R_{q_2})$.
We have the following cases:

\begin{itemize}
\item
$s^\forall = s^R_{q^\forall}$ is one of the pending nodes that are ancestors of $s^R_{q_1}$,
and $q^\forall$ is the corresponding position in $\vla^R_k(w_1, s_1)$.
Then, $\exists$ selects $q^\exists$ in $\vla^R_k(w_2, s_2)$ in response to $q^\forall$
according to her winning strategy for game~\ref{item:foc-rr-vla}.
Her answer to $s^\forall$ is $s^\exists = s^R_{q^\exists}$
(i.e.\ the pending node with the same index).

\item
$s^\forall$ is part of a subtree $t^R_{w_1}(s^R_{q^\forall})$
such that $s^R_{q^\forall}$ is an ancestor of $s^R_{q_1}$.
Then $\exists$ chooses position $q^\exists$ in $\vla^R_k(w_2, s_2)$ as before.
According to game~\ref{item:foc-rr-vla}, $q^\forall$ and $q^\exists$ must be labeled with
the same rank-$k$ type of a subtree
(i.e.\ $\sigma^R_k(w_1, s^R_{q^\forall}) = \sigma^R_k(w_2, s^R_{q^\exists})$).
Hence, game $t^R_{w_1}(s^R_{q^\forall}) \sim_k t^R_{w_2}(s^R_{q^\exists})$ has a winning strategy,
which $\exists$ can use to pick $s^\exists$ in $t^R_{w_2}(s^R_{q^\exists})$.

\item
If the node picked by $\forall$ is the rightmost child of $s^R_{q_1}$
or one of its descendants, then $\exists$ proceeds symmetrically,
but using her winning strategy on game~\ref{item:foc-rr-vra}.

\item
Finally, if $\forall$ picks a node in $t^R_{w_1}(s^R_{q_1})$,
then by~\ref{item:foc-rr-sametype} we have $t^R_{w_1}(s^R_{q_1}) \sim_k t^R_{w_2}(s^R_{q_2})$,
and $\exists$ answers according to her winning strategy in this game.
\end{itemize}

\noindent
This strategy preserves the partial isomorphism
w.r.t.\ the child and sibling relations and monadic predicates,
as a direct consequence of rank-$k$ type equivalences in the hypotheses.
\end{proof}

Lemma~\ref{lemma:composition-rr} shows that the rank-$k$ type of an RR OPM-compatible UOT
is determined by the rank-$k$ types of the parts in which we divide it.
Given FO formula $\bar{\varphi}(x)$, consider the set of all tuples made of
\begin{enumerate*}
\item the rank-$k$ type of $\vla^R_k(w, s)$,
\item the rank-$k$ type of $\vra^R_k(w, s)$, and
\item the type $\sigma^R_k(w, s)$,
\end{enumerate*}
such that $(T_w, s) \models \bar{\varphi}(x)$ for any RR UOT $T_w$ and $s \in T_w$.
This set is finite, because there are only finitely many rank-$k$ types of each component.
So we can translate the formulas expressing the types in each tuple into \xuntil{} separately,
and combine them to obtain one for the whole tree.
Then, \xuntil{} formula $\varphi$ is a disjunction of the resulting translated formulas,
one for each tuple.

Before proceeding with our translation, we need to show how to express in \xuntil{}
the rank-$k$ type of a finite UOT such as $t^R_{w}(s^R_{p})$,
for some $w$ and $p$, in the context of $T_w$.
Since the rank-$k$ type of $t^R_{w}(s^R_{p})$ only contains information about that subtree,
we need to restrict the formula expressing it to such nodes.
In the following lemma we show how to do this, thanks to a formula $\alpha^*$
which holds in the root of the subtree (but may hold in other parts of $T_w$),
and allows us to restrict \xuntil{} operators so that they do not exit the subtree.

\begin{lem}%
\label{lemma:composition-subtree}
Let $\sigma_k(t)$, with $k \geq 1$, be the rank-$k$ type of a finite OPM-compatible UOT $t$.
Let $r$ be a node of a RR or LR OPM-compatible UOT $T_w$ with a finite number of children.
Let $\alpha^*$ be a \xuntil{} formula that holds in $r$, and does not hold in subtrees rooted at children of $r$, except possibly the rightmost one.
Then, there exists a \xuntil{} formula $\beta(t)$ that, if evaluated in $r$,
is true iff $r$ is the root of a subtree of $T_w$ with rank-$k$ type $\sigma_k(t)$,
from which the subtree rooted in $r$'s rightmost child has been erased in case $\alpha^*$ holds in that child.
\end{lem}
\begin{proof}
Let $t_w$ be the subtree rooted at $r$, excluding $r$'s rightmost child and its descendants,
if $\alpha^*$ holds in it.
We provide a formula $\beta(t)$ that holds in $r$ iff $t_w \equiv_k t$.
If $t$ has only one node $s$ with no children,
which can be determined with a FO formula with one quantifier,
then $\beta(t) := \beta_{AP}(s) \land \neg \xcnext (\neg \alpha^*)$,
where $\beta_{AP}(s)$ is a Boolean combination of the atomic propositions holding in $s$.

Otherwise, the rank-$k$ type of $t$ is fully determined by
\begin{enumerate}
\item the propositional symbols holding in its root $s$;
\item the rank-$k$ type of the finite LTL word $v_t$,
  whose positions $0 \leq p \leq m$
  are labeled with the rank-$k$ types of the subtrees $t_{u_p}$
  rooted at the children $u_p$ of $s$ (including the rightmost one).
\end{enumerate}
This can be proved with a simple compositional argument.

Thus, we define
\[
  \beta(t) := \xcnext (\neg \xsback \top \land \beta'(t)) \land \beta_{AP}(s),
\]
where $\beta_{AP}(s)$ is a Boolean combination of the atomic propositions holding in $s$,
and $\beta'(t)$ characterizes $t_{u_0}, \dots, t_{u_m}$.
By this we mean that $\beta'(t)$ is such that when $\beta(t)$ holds on $r$,
its children (except the rightmost one if $\alpha^*$ holds in it)
must be roots of subtrees isomorphic to $t_{u_0}, \dots, t_{u_m}$.
Note that formula $\beta'(t)$ is enforced in the leftmost child
(where $\xsback \top$ is false) of the node where $\beta(t)$ is evaluated.

We now show how to obtain $\beta'(t)$.
Due to Kamp's Theorem~\cite{Kamp68} and the separation property of LTL~\cite{GabbayHR94},
there exists a future LTL formula $\beta''(t)$ that,
evaluated in the first position of $v_t$, completely determines its rank-$k$ type
(it can be obtained by translating into LTL the Hintikka formula
equivalent to  the rank-$k$ type of $v_t$).

We now show how to express the rank-$k$ types of the subtrees $t_{u_p}$.
Since they are finite, by Marx's Theorem~\cite[Corollary 3.3]{Marx2005},
there exists a \xuntil{} formula $\gamma_p$ that,
evaluated in $u_p$, fully determines the rank-$k$ type of $t_{u_p}$
($\gamma_p$ can be obtained by translating the Hintikka formula
for the rank-$k$ type of $t_{u_p}$ into \xuntil{}).
Unfortunately, the separation property does not hold for \xuntil{}~\cite{BenediktL16},
and $\gamma_p$ may contain $\xcsinceo$ operators that, in the context of $T_w$,
consider nodes that are not part of $t_{u_p}$.

There is, however, a way of syntactically transforming $\gamma_p$ so that,
if evaluated on a child $r'$ of $r$, its paths remain constrained to $t_{r'}$,
the subtree rooted in $r'$.
Given a \xuntil{} formula $\psi$, it can be written as a Boolean combination
of atomic propositions and until/since operators (possibly nested).
We denote by $\psi^\Downarrow$ the formula obtained by replacing all subformulas of the form
$\xcsince{\varphi}{\varphi'}$, $\xsuntil{\varphi}{\varphi'}$ and $\xssince{\varphi}{\varphi'}$ at the topmost level with $\neg \top$.
If $\gamma_p$ is evaluated in the root of $t_{u_p}$ outside of $t$,
all such operators evaluate to false.
So, $\gamma_p^\Downarrow$ in $r'$ agrees with $\gamma_p$ in the root of $t_{u_p}$ on such subformulas.
Now, take $\gamma_p^\Downarrow$, and recursively replace all subformulas of the form
$\xcsince{\varphi}{\varphi'}$ with the following:
\[
  \xcsince{\varphi \land \neg \xcback \alpha^*}
    {(\neg \xcback \alpha^* \land \varphi') \lor
     (\xcback \alpha^* \land \varphi^{\prime \Downarrow})}.
\]
We call $\gamma'_p$ the obtained formula.
Note that, since $\alpha^*$ holds in $r$ and at most in its rightmost child,
$\xcback \alpha^*$ only holds in nodes $u_p$, $0 \leq p \leq m$.
This way, we can prove that $\xcsinceo$ paths in $\gamma'_p$ cannot continue past $u_p$,
and those ending in $u_p$ depend on a formula that does not consider nodes outside the subtree
rooted at $u_p$.
Thus, when $\gamma'_p$ is evaluated in $u_p$ in the context of $T_w$,
all its until/since operators consider exactly the same positions as the corresponding ones
in $\gamma_p$ evaluated outside of $T_w$.
Hence, the two formulas are equivalent in their respective contexts,
and $\gamma'_p$ is true in positions that are the root of a subtree equivalent to $t_{u_p}$.

We now show that such transformations do not change the formula's meaning in unwanted ways.
Let $t_1$ be a UOT, and let $t_2$ be a subtree of a larger UOT $T$,
such that $t_1$ and $t_2$ are identical, and $\alpha^*$ holds in the parent of $r_2$,
the root of $t_2$.
We prove that $\gamma'_p$ holds on $r_2$ iff $\gamma_p$ holds on $r_1$, the root of $t_1$.
For any subformula $\psi$ of $\gamma_p$ not at the topmost nesting level,
let $\psi'$ be the corresponding subformula in $\gamma'_p$.
Then, we prove by structural induction on $\psi$ that, for any node $s \neq r_1$ in $t_1$ and $t_2$,
we have $(t_1, s) \models \psi$ iff $(T, s) \models \psi'$.

The base case, where $\psi$ is an atomic proposition, is trivial.
The composition by Boolean operators is also straightforward.

Paths considered by formulas of the form $\xsuntil{\varphi_1}{\varphi_2}$,
$\xssince{\varphi_1}{\varphi_2}$, and $\xcuntil{\varphi_1}{\varphi_2}$ cannot get out of $t_2$,
when evaluated in one of its nodes that is not $r_2$.
Thus, we have e.g.\ $(t_1, s) \models \xsuntil{\varphi_1}{\varphi_2}$
iff $(T, s) \models \xsuntil{\varphi'_1}{\varphi'_2}$ directly from the fact that
$(t_1, s') \models \varphi_1$ iff $(T, s') \models \varphi'_1$ for any $s' \in t_1$
(and the same for $\varphi_2$).

Finally, let $\psi = \xcsince{\varphi_1}{\varphi_2}$.
Suppose $\psi$ is witnessed by a path in $t_1$ that does not include its root $r_1$.
Then, $\neg \xcback \alpha^*$ holds in all positions in such path inside $t_2$,
and $\psi'$ is satisfied by the same path, thanks to the inductive hypothesis.
Otherwise, $\psi$ may be witnessed by a path ending in $r_1$.
In this case, $\neg \xcback \alpha^*$ holds in all positions
in the corresponding path in $t_2$ except the last one.
There, $(\xcback \alpha^* \land \varphi^{\prime \Downarrow})$ holds.
In fact, $\varphi$ holds in the root of $t_1$ where any $\xcsinceo$,
$\xsuntilo$ and $\xssinceo$ operator are false,
so they can correctly be replaced with $\neg \top$ in $\varphi^{\prime \Downarrow}$.
Moreover, $\xcuntilo$ operators are strict, so they must be witnessed by paths
completely contained in $t_1$, excluding its root.
Hence, by the inductive hypothesis they witness
the corresponding operators in $\varphi^{\prime \Downarrow}$.
Conversely, any path that witnesses $\psi'$ in $s \in t_2$ must be, by construction,
entirely inside $t_2$, so it witnesses $\psi$ in $t_1$ as well.
Finally, $\gamma'_p$ holding on $r_2$ can be justified by an argument similar
to the one for $\varphi^{\prime \Downarrow}$.

Now, we can go on with the definition of $\beta(t)$:
we set $\beta'(t) := \neg \alpha^* \land \beta'''(t)$, 
where $\beta'''(t)$ is obtained from $\beta''(t)$ by 
\begin{itemize}
\item recursively replacing all subformulas $\lluntil{\varphi}{\varphi'}$
  with $\xsuntil{\neg \alpha^* \land \varphi}{\neg \alpha^* \land \varphi'}$;
\item replacing all labels, which are rank-$k$ types of UOTs $t_{u_p}$,
  with the corresponding $\gamma'_p$.
\end{itemize}
Thus, $\beta'(t)$ is a \xuntil{} formula that characterizes children of $r$,
without considering the one where $\alpha^*$ holds (if any).
\end{proof}

Lemma~\ref{lemma:composition-subtree} allows us to express in \xuntil{} the rank-$k$ type of
any subtree $t = t^R_w(s^R_p)$, for any pending node $s^R_p$ in $T_w$.
The root of $t$ is a pending node, and so is its rightmost child $s^R_{p+1}$,
which is not part of $t$.
So, we use $\alpha^R_\infty := \neg \xcsince{\top}{\xsnext \top}$, which is true in pending nodes,
where the path from the current node to the root is made of rightmost nodes only.
Thus, formula $\beta(t)$ from Lemma~\ref{lemma:composition-subtree}
is true in a pending node iff it is the root of a subtree equivalent to $t$,
excluding its rightmost child.

Let $n$ be the node where the translated formula is evaluated
(i.e.\ the one corresponding to the free variable in the FO formula $\bar{\varphi}(x)$
to be translated).
Let $s^R_{p_n}$ be its closest pending ancestor.
According to Lemma~\ref{lemma:composition-rr}, we need to express the rank-$k$ types
of $t^R_w(s^R_{p_n})$, $\vra^R_k(w, n)$, and $\vla^R_k(w, n)$.

For $t^R_w(s^R_{p_n})$, simply take formula $\beta(t^R_w(s^R_{p_n}))$.

By Kamp's Theorem and the separation property of LTL, the rank-$k$ type of $\vra^R_k(w, n)$
can be expressed by a future LTL formula $\overrightarrow{\psi}$
to be evaluated in its first position.
First, we recursively take the conjunction of all subformulas of $\overrightarrow{\psi}$
with $\alpha^R_\infty$.
Then, we recursively substitute each LTL operator $\lluntil{\varphi}{\varphi'}$ with
$\xcuntil{\varphi}{\varphi'}$, obtaining a \xuntil{} formula $\overrightarrow{\psi}'$
that evaluates its paths only on pending positions in $T_w$.
We can prove that equivalence is kept despite such transformations by induction
on the formula's structure, as we did in Lemma~\ref{lemma:composition-subtree}.
Since $\vra^R_k(w, n)$ is labeled with rank-$k$ types of UOTs,
we substitute any such atomic proposition $\sigma_k(t)$ in $\overrightarrow{\psi}'$
with $\beta(t)$, obtaining formula $\overrightarrow{\psi}''$,
that captures the part of the tree rooted at the rightmost child of $s^R_{p_n}$.

A formula $\overleftarrow{\psi}''$ for $\vla^R_k(w, n)$ can be obtained similarly,
but using a past LTL formula.
The Since modality can be replaced with $\xcsinceo$,
while other transformations remain the same.

All formulas we built so far are meant to be evaluated in a pending node.
If $n = s^R_{p_n}$, then we are done.
Otherwise, we evaluate in $n$ an appropriate $\xcsinceo$ formula,
that can only be witnessed by a path ending in $s^R_{p_n}$.
Thus, the final translation is
\[
  \varphi_R :=
  (\alpha^R_\infty \land \varphi'_R) \lor
  \xcsince{\neg \alpha^R_\infty}{\alpha^R_\infty \land \varphi'_R}
\]
where
\[
  \varphi'_R :=
  \beta(t^R_w(s^R_{p_n}))
  \land \xcnext (\neg \xsnext \top \land \overrightarrow{\psi}'')
  \land \xcback \overleftarrow{\psi}''.
\]

This concludes the proof of Lemma~\ref{lemma:transl-rr}.
\qed%

\subsection{LR UOTs}%
\label{sec:lr-trees}

\begin{figure}[bt]
  \centering
  \includegraphics{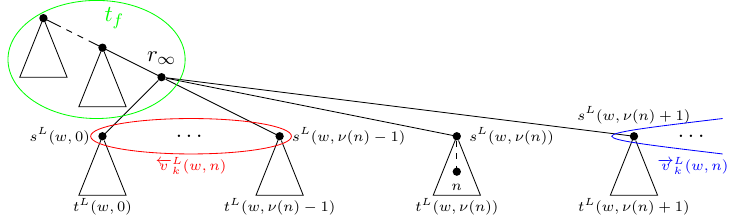}
  \caption{Parts in which we divide a LR UOT for Lemma~\ref{lemma:composition-lr}.}%
  \label{fig:lr-tree-comp}
\end{figure}

Let $w$ be an OP $\omega$-word, and $T_w = \tau(w)$ a LR UOT\@.
We name $r_\infty$ the node with infinite children, and
denote as $t_f(w)$ the finite UOT obtained by removing all children of $r_\infty$ from $T_w$.
We prove the following:
\begin{lem}%
\label{lemma:transl-lr}
Given a FO formula on UOTs $\bar{\varphi}(x)$ of q.r.\ $k \geq 1$,
there are two \xuntil{} formulas $\varphi_{L1}$
and $\varphi_{L2}$ s.t.\ for any OP $\omega$-word s.t.\ $T_w = \tau(w)$ is a LR UOT,
and for any node $n \in T_w$ we have
\begin{itemize}
\item $(T_w, n) \models \bar{\varphi}(x)$ iff $(T_w, n) \models \varphi_{L1}$ when $n \not\in t_f(w)$, and
\item $(T_w, n) \models \bar{\varphi}(x)$ iff $(T_w, n) \models \varphi_{L2}$ when $n \in t_f(w)$.
\end{itemize}
\end{lem}

\noindent
The proof is structured in a way similar to that of Lemma~\ref{lemma:transl-rr},
but differs in the parts in which we divide the tree.
In fact, we divide a LR UOT into its finite part,
and two LTL words which make up the infinite children of $r_\infty$
(see Figure~\ref{fig:lr-tree-comp}).

We name $s^L(w, 0), s^L(w, 1), \dots$ the children of $r_\infty$,
and denote as $t^L(w, 0), t^L(w, 1), \dots$ the finite subtrees rooted at,
resp., $s^L(w, 0), s^L(w, 1), \dots$,
and $\sigma^L_k(w, 0), \allowbreak \sigma^L_k(w, 1), \allowbreak \dots$ their rank-$k$ types.
For all $p \geq 0$, for any node $m$ in $t^L(w, p)$,
we define the map $\nu$ so that $\nu(m) = p$.
Now, let $n$ be any node in $T_w$.
We define $\vra^L_k(w, n)$ as the LTL $\omega$-word on the alphabet of rank-$k$ types
of finite OPM-compatible UOTs, such that each of its positions $i_q$, $q \geq 0$,
is labeled with $\sigma^L_k(w, \nu(n)+q+1)$ if $n \not\in t_f(w)$.
If $n$ is in $t_f(w)$, then each position $i_q$ is labeled with $\sigma^L_k(w, q)$.
We further define $\vla^L_k(w, n)$ as the finite LTL word of length $\nu(n)$ on the same alphabet,
such that each of its positions $i_q$, $0 \leq q \leq \nu(n)-1$, is labeled with $\sigma^L_k(w, q)$.
If $\nu(n) = 0$, or $n$ is part of $t_f(w)$, then $\vla^L_k(w, n)$ is the empty word.
We now prove the following composition argument:
\begin{lem}%
\label{lemma:composition-lr}
Let $w_1$ and $w_2$ be two OP $\omega$-words, such that $T_{w_1} = \tau(w_1)$
and $T_{w_2} = \tau(w_2)$ are two LR UOTs,
and $r_\infty^1$ and $r_\infty^2$ are their nodes with infinitely many children.
Let $i_1$ and $i_2$ be two positions in $w_1$ and $w_2$, such that, by letting $n_1 = \tau(i_1)$ and $n_2 = \tau(i_2)$,
$n_1 \in t_f(w_1)$ iff $n_2 \in t_f(w_2)$.
If
\begin{enumerate}
\item\label{item:foc-lr-tf}
  $(t_f(w_1), n_1) \equiv_k (t_f(w_2), n_2)$ if $n_1 \in t_f(w_1)$;
  $(t_f(w_1), r_\infty^1) \equiv_k (t_f(w_2), r_\infty^2)$ otherwise;
\item\label{item:foc-lr-vla}
  $\vla^L_k(w_1, n_1) \equiv_k \vla^L_k(w_2, n_2)$;
\item\label{item:foc-lr-vra}
  $\vra^L_k(w_1, n_1) \equiv_k \vra^L_k(w_2, n_2)$;
\item\label{item:foc-lr-sametype}
  $n_1 \not\in t_f(w_1)$ implies $t^L(w_1, \nu(n_1)) \equiv_k t^L(w_2, \nu(n_2))$
\end{enumerate}
then $(T_{w_1}, n_1) \equiv_k (T_{w_2}, n_2)$.
\end{lem}
\begin{proof}
The proof is carried out by means of a standard composition argument.
We give a winning strategy for $\exists$ in the $k$-round EF game
between $(T_{w_1}, n_1)$ and $(T_{w_2}, n_2)$.
Suppose w.l.o.g.\ that $\forall$ picks a node $n^\forall$ from $T_{w_1}$.
$n^\forall$ may be part of, either, $t_f(w_1)$ or a subtree $t^L(w_1, \nu(n^\forall))$.
In the former case, $\exists$ answers with a node in $t_f(w_2)$,
according to her winning strategy in game~\ref{item:foc-lr-tf}.
For the latter case, let us first suppose $n_1$ and $n_2$ are not in, resp., $t_f(w_1)$ and $t_f(w_2)$.
Then, $\exists$ proceeds as follows:
\begin{enumerate}[label=(\alph*)]
\item\label{item:foc-lr-vla-proof}
  $\nu(n^\forall) < \nu(n_1)$.
  In this case, $\exists$ plays her game~\ref{item:foc-lr-vla} as if $\forall$ had picked
  position $\nu(n^\forall)$ in $\vla^L_k(w_1, n_1)$.
  Let $q^\exists$ be the position in $\vla^L_k(w_2, n_2)$ chosen according to such strategy.
  Since $\vla^L_k(w_1, n_1) \equiv_k \vla^L_k(w_2, n_2)$ and $k \geq 1$,
  $\nu(n^\forall)$ and $q^\exists$ must have the same label,
  so we have $\sigma^L_k(w_1, \nu(n^\forall)) = \sigma^L_k(w_2, q^\exists)$.
  So, $t^L(w_1, \nu(n^\forall)) \equiv_k t^L(w_2, q^\exists)$, and $\exists$ may pick a node $n^\exists$
  in $t^L(w_2, q^\exists)$ according to her winning strategy for the $k$-round game on $t^L(w_1, \nu(n^\forall))$
  and $t^L(w_2, q^\exists)$, considering $n^\forall$ as $\forall$'s move.
\item $\nu(n^\forall) = \nu(n_1)$.
  In this case, $\exists$ picks $n^\exists$ in $T_{w_2}$ according to her winning strategy
  on the EF game for equivalence~\ref{item:foc-lr-sametype}.
\item\label{item:foc-lr-vra-proof}
  $\nu(n^\forall) > \nu(n_1)$.
  $\exists$ may proceed as in case~\ref{item:foc-lr-vla-proof},
  but picking $q^\exists$ from $\vra^L_k(w_2, n_2)$ according to her winning strategy
  on game~\ref{item:foc-lr-vra}.
\end{enumerate}

\noindent
If, instead, $n_1$ and $n_2$ are in $t_f(w_1)$ and $t_f(w_2)$,
then $\vra^L_k(w_1, n_1)$ (resp.\ $\vra^L_k(w_2, n_2)$)
represents all of the infinite siblings in $T_{w_1}$ (resp.\ $T_{w_2}$),
and $\vla^L_k(w_1, n_1)$ and $\vla^L_k(w_2, n_2)$ are the empty word.
So, $\exists$ may proceed as in case~\ref{item:foc-lr-vla-proof},
but picking $q^\exists$ from $\vra^L_k(w_2, n_2)$ according to her winning strategy
on game~\ref{item:foc-lr-vra}.
\end{proof}

We now show how to express a FO formula with one free variable with a \xuntil{}
formula equivalent to it on a LR UOT\@.
As in the RR case, we show how to represent in \xuntil{} the rank-$k$ types of all
parts in which we divide the tree in Lemma~\ref{lemma:composition-lr}.
Let $n$ be the node in which the \xuntil{} formula is evaluated.
We need to distinguish whether $n \in t_f(w)$ or not.
This is more conveniently done while also discerning such cases from the RR one.
Thus, we now treat the two LR cases separately, and show how to combine them with the RR case
in Section~\ref{sec:synthesis}.

Suppose $n \not\in t_f(w)$.
Let $s_n = s^L(w, \nu(n))$ be the root of the subtree $t_n = t^L(w, \nu(n))$ containing $n$.
Children of $r_\infty$ can be identified by the \xuntil{} formula
$\alpha^L_\infty := \xsnext \top \land \neg \xsuntil{\top}{\neg \xsnext \top}$,
saying that no right sibling without right siblings is reachable from the current node
(i.e.\ there exists no rightmost sibling).
By Lemma~\ref{lemma:composition-subtree} with $\alpha^* = \alpha^L_\infty$,
there exists a \xuntil{} formula $\beta(t_n)$ that fully identifies the rank-$k$ type
of $t_n$ if evaluated in $s_n$.

Moreover, $\vla^L_k(w, n)$ and $\vra^L_k(w, n)$ can be expressed similarly to
$\vla^R_k(w, n)$ and $\vra^R_k(w, n)$ for RR UOTs.
By Kamp's Theorem, there exists a future LTL formula $\overrightarrow{\psi}_k(w, n)$ that,
evaluated in the first position of $\vra^L_k(w, n)$, fully identifies its rank-$k$ type.
Recursively substitute $\lluntil{\varphi}{\varphi'}$ subformulas with $\xsuntil{\varphi}{\varphi'}$
in $\overrightarrow{\psi}_k(w, n)$, obtaining $\overrightarrow{\psi}'_k(w, n)$.
Then, replace all rank-$k$ types of UOTs in $\overrightarrow{\psi}'_k(w, n)$
with the respective formulas obtained from Lemma~\ref{lemma:composition-subtree},
with $\alpha^* = \alpha^L_\infty$, thus obtaining $\overrightarrow{\psi}''_k(w, n)$.
Now, $\xsnext \overrightarrow{\psi}''_k(w, n)$, evaluated in $s_n$,
fully describes the rank-$k$ type of all right siblings of $s_n$, and of the subtrees rooted in them.

Formula $\xsback \overleftarrow{\psi}''_k(w, n)$, which describes left siblings of $s_n$ if evaluated
in it, can be obtained symmetrically, but replacing $\llsince{}{}$ with $\xssinceo$
in the LTL formula.

Finally, we need to describe the rank-$k$ type of $t_f(w)$.
By Marx's Theorem~\cite[Corollary 3.3]{Marx2005}, there exists a formula $\psi_{f1}$ that,
evaluated in $r_\infty$, describes the rank-$k$ type of $t_f(w)$.
Take the conjunction of each subformula of $\psi_{f1}$ with $\neg \alpha^L_\infty$,
and call the obtained formula $\psi'_{f1}$.
Thus, the rank-$k$ type of $t_f(w)$ is described by $\xcback \psi'_{f1}$, evaluated in $s_n$.
Finally, the rank-$k$ type of the whole tree is described by the following formula,
evaluated in $n$:
\[
  \varphi_{L1} :=
  (\alpha^L_\infty \land \varphi'_{L1}) \lor
  \xcsince{\neg \alpha^L_\infty}{\alpha^L_\infty \land \varphi'_{L1}},
\]
where the $\xcsinceo$ operator is needed to reach $s_n$ from $n$ if $s_n \neq n$, and
\[
  \varphi'_{L1} :=
  \xcback \psi'_{f1} \land
  \xsback \overleftarrow{\psi}''_k(w, n) \land
  \xsnext \overrightarrow{\psi}''_k(w, n) \land
  \beta(t_n).
\]

Suppose, instead, $n \in t_f(w)$.
Then, we express the rank-$k$ type of $t_f(w)$ by means of a \xuntil{} formula $\psi_{f2}$,
evaluated in $n$, which exists by Marx's Theorem~\cite[Corollary 3.3]{Marx2005}.
In it, we recursively take the conjunction of subformulas with $\neg \alpha^L_\infty$,
thus obtaining $\psi'_{f2}$.

Next, we need to describe the rank-$k$ type of $\vra^L_k(w, s^L(w, 0))$
(recall that $\vla^L_k(w, s^L(w, 0))$ is the empty word).
This can be done as in case $n \not\in t_f(w)$,
thus obtaining formula $\overrightarrow{\psi}''_k(w, s^L(w, 0))$ which,
evaluated in $s^L(w, 0)$, fully identifies the children of $r_\infty$.
The latter can be identified by formula $\xcnext \alpha^L_\infty$.
Thus, to describe its children from $n$, we use formula
\[
  \varphi'_{L2} :=
  \xcsince{\top}
    {\neg \xcback \top \land \xcuntil{\neg \xsnext \top}
      {\xcnext \big(\alpha^L_\infty \land \neg \xsback \top \land \overrightarrow{\psi}''_k(w, s^L(w, 0))\big)}},
\]
where the outermost $\xcsinceo$ reaches the root of $T_w$ (identified by $\neg \xcback \top$),
the inner $\xcuntilo$ reaches $r_\infty$ (identified by $\xcnext \alpha^L_\infty$)
by descending through rightmost children
($r_\infty$ is the left context of an open chain, hence a pending position),
and $\neg \xsback \top$ identifies $s^L(w, 0)$.
The final formula is
\[
  \varphi_{L2} :=
  \varphi'_{L2} \land \psi'_{f2}.
\]

This concludes the proof of Lemma~\ref{lemma:transl-lr}.
\qed%

\subsection{Synthesis}%
\label{sec:synthesis}
To finish the proof, we must combine the previous cases to obtain a single \xuntil{} formula.
First, note that it is possible to discern the type of UOT by means of FO
formulas of q.r.\ at most 5.
The following formula identifies RR UOTs:
\[
  \gamma_R :=
  \forall x [(0 \rchild^* x \land \forall y (0 \rchild^* y \land y \rchild^* x)
    \implies \mu(y))
      \implies \exists y (x \rchild y \land \mu(y))],
\]
where $\mu(y) := \neg \exists z (y \rsibl z)$.
$\gamma_R$ means that any node $x$ which is on the rightmost branch from the root
must have a rightmost child, i.e.\ the rightmost branch has no end.
LR UOTs are identified by the following formulas:
\[
  \gamma_{L1}(n) :=
  \exists x [x \rchild^+ n \land \exists y (x \rchild y) \land
    \forall y (x \rchild y \implies \exists z (y \rsibl z))]
\]
\[
  \gamma_{L2}(n) :=
  \exists x [\neg (x \rchild^+ n) \land \exists y (x \rchild y) \land
    \forall y (x \rchild y \implies \exists z (y \rsibl z))]
\]
Both $\gamma_{L1}(n)$ and $\gamma_{L2}(n)$ say there exists a node $x = r_\infty$
that has infinite children, but $\gamma_{L2}(n)$ is true iff the free variable $n$ is in $t_f(w)$,
while $\gamma_{L1}(n)$ is true otherwise.

Given a FO formula of q.r.\ $m$ with one free variable $\bar{\varphi}(x)$, let $k = \max(m, 5)$.
Consider the finite set
$\Gamma_k = \{\sigma_k(T_w, n) \mid (T_w, n) \models \bar{\varphi}(x),\ n \in T_w\}$
of the rank-$k$ types of OPM-compatible UOTs satisfying $\bar{\varphi}(x)$,
with $n$ as a distinguished node.
For each one of them, take the corresponding Hintikka formula $H(w, n)$.
Since $k \geq 5$, and the UOT type is distinguishable by formulas of q.r.\ at most 5,
for each $\sigma_k(T_w, n)$ it is possible to tell whether it describes RR or LR UOTs.
For each type, by Lemmas~\ref{lemma:transl-rr} and~\ref{lemma:transl-lr},
it is possible to express $H(w, n)$ through formulas describing the rank-$k$
types of the substructures given by the composition arguments,
and translate them into \xuntil{} accordingly.
Then, the translated formula $\varphi$ is the XOR of one of the following formulas,
for each type $\sigma_k(T_w, n) \in \Gamma_k$.
\begin{itemize}
\item If $T_w$ is RR, then
  \(
    \xi_R \land \varphi_R(\sigma_k(T_w, n)).
  \)
\item If $T_w$ is LR, and $n \not\in t_f(w)$:
  \(
    \xi_L \land (\alpha^L_\infty \lor \xcsince{\top}{\alpha^L_\infty})
      \land \varphi_{L1}(\sigma_k(T_w, n)),
  \)
  where we assert one of the infinite siblings is reachable going upwards form $n$,
  and so $n$ is not in $t_f(w)$.
\item If $T_w$ is LR, and $n \in t_f(w)$:
  \(
    \xi_L \land \neg (\alpha^L_\infty \lor \xcsince{\top}{\alpha^L_\infty})
      \land \varphi_{L2}(\sigma_k(T_w, n)).
  \)
\end{itemize}
In the above formulas, $\varphi_R(\sigma_k(T_w, n))$, $\varphi_{L1}(\sigma_k(T_w, n))$,
and $\varphi_{L2}(\sigma_k(T_w, n))$ are the formulas expressing $\sigma_k(T_w, n)$,
obtained as described in the previous paragraphs.
Moreover, we have
\begin{align*}
  \xi_R &:=
    \xcsince{\top}{\neg \xcback \top \land
      \neg \xcuntil{\neg \xsnext \top}
        {\neg \xsnext \top \land \neg \xcnext \top}}, \\
  \xi_L &:=
    \xcsince{\top}{\neg \xcback \top \land \xcuntil{\neg \xsnext \top}{\neg \xsnext \top \land \xcnext \alpha^L_\infty}}.
\end{align*}
$\xi_R$ identifies RR UOTs. In it, the outermost $\xcsinceo$ reaches the root of the tree,
  and the inner $\xcuntilo$ imposes that the rightmost branch has no end.
$\xi_L$ identifies LR UOTs. The outermost $\xcsinceo$ also reaches the root of the tree,
and the inner $\xcuntilo$ is verified by a path in which all nodes are on the rightmost branch,
except the last one, which is one of the infinite siblings.

Boolean queries can be expressed as Boolean combinations of formulas of the form
$\bar{\varphi} := \exists x (\bar{\varphi}'(x))$.
Then, it is possible to translate $\bar{\varphi}'(x)$ as above, to obtain \xuntil{}
formula $\varphi'$, and $\xcuntil{\top}{\varphi'}$ is such that
$(T_w, 0) \models \xcuntil{\top}{\varphi'}$ iff $T_w \models \bar{\varphi}$,
for any OPM-compatible UOT $T_w$.

Thus, we can state
\begin{thm}%
\label{thm:xuntil-opm-trees}
\xuntil{} $=$ FOL with one free variable on OPM-compatible $\omega$-UOTs.
\end{thm}

Thanks to Theorem~\ref{thm:xuntil-opm-trees}, we can extend the results entailed
by the translation of Section~\ref{sec:cxpath-translation} to $\omega$-words.
\begin{thm}%
\label{thm:potl-completeness-omega}
POTL = FOL with one free variable on OP $\omega$-words.
\end{thm}
\begin{cor}%
\label{cor:complete-subset-omega}
The propositional operators plus
\(
  \ldnext,
  \ldback,
  \lcdnext,
  \lcdback,
  \lcduntil{}{},
  \lcdsince{}{},
  \lhunext, \allowbreak
  \lhuback, \allowbreak
  \lhuuntil{}{}, \allowbreak
  \lhusince{}{}
\)
are expressively complete on OP $\omega$-words.
\end{cor}

The containment relations in Corollary~\ref{cor:optl-in-potl} are easily extended to $\omega$-words:
\begin{cor}%
\label{cor:optl-in-potl-omega}
NWTL $\subset$ OPTL $\subset$ POTL over OP $\omega$-words.
\end{cor}
Moreover,
\begin{cor}%
\label{cor:3var-omega}
Every FO formula with at most one free variable is equivalent
to one using at most three distinct variables on OP $\omega$-words.
\end{cor}

\section{Decidability, Satisfiability, and Model Checking}%
\label{sec:mc}

The decidability of POTL is a consequence of its being equivalent to the FO fragment
of the MSO characterization of OPLs~\cite{LonatiEtAl2015}.
More practical algorithms for satisfiability and model checking can be obtained
by building OPAs equivalent to POTL formulas.
Such construction procedure, which is quite involved, is detailed in~\cite{ChiariMP21}.
Thus,
\begin{thmC}[\cite{ChiariMP21}]
Given a POTL formula $\varphi$, we can build an OPA (or an $\omega$OPBA)
$\mathcal{A}_\varphi$ of size $2^{O(|\varphi|)}$ accepting the language defined by $\varphi$.
\end{thmC}

Emptiness of OPA and $\omega$OPBA is decidable in polynomial time by adapting techniques
originally developed for Pushdown Systems and Recursive State Machines,
such as \emph{saturation}~\cite{AlurBE18}, or graph-theoretic algorithms~\cite{AlurCEM05}.
Hence, satisfiability of a formula $\varphi$ can be tested in time exponential in $|\varphi|$
by building $\mathcal{A}_\varphi$ and checking its emptiness.
By Corollaries~\ref{cor:optl-in-potl} and~\ref{cor:optl-in-potl-omega},
we can extend the EXPTIME-hardness result for NWTL~\cite{lmcs/AlurABEIL08} to POTL, obtaining
\begin{cor}
POTL satisfiability is EXPTIME-complete.
\end{cor}

Since there are practical algorithms for computing the intersection of both OPA and $\omega$OPBA~\cite{LonatiEtAl2015}, model checking of OPA (or $\omega$OPBA) models against a formula $\varphi$
can be done by building $\mathcal{A}_{\neg \varphi}$, and checking emptiness of their intersection.
Thus,
\begin{cor}
POTL model checking against OPA and $\omega$OPBA is EXPTIME-complete.
\end{cor}

An experimental evaluation of the POTL model-checking procedure is reported in~\cite{ChiariMP21}.
Just to give an idea of its practicality, we report that formula
\[
  \llglob \big((\lcall \land \mathrm{p}_B \land
    \lcallsince(\top, \mathrm{p}_A))
    \implies \lthrnext(\top) \big)
\]
from Section~\ref{sec:motivating-examples} has been successfully checked
against the OPA of Figure~\ref{fig:example-prog} in just 867 ms,
with a RAM occupancy of 70 MiB, on a laptop with a 2.2 GHz Intel processor and 15 GiB of RAM,
running Ubuntu GNU/Linux 20.04.

\section{Conclusions}%
\label{sec:conclusion}

We introduced the temporal logic POTL, and proved its equivalence to FOL
on both OP finite and $\omega$-words.
Thus, thanks to an independent result~\cite{MPC20-arXiv},
the languages defined through POTL formulas coincide with the class of aperiodic OPLs.
We also proved that POTL is strictly more expressive than temporal logics
with explicit context-free-aware modalities in the literature,
including OPTL, which is also based on OPLs.
Such proofs are technically quite involved, which is unsurprising,
given the difficulties encountered in analogous problems,
even if based on the simpler framework of Nested Words~\cite{lmcs/AlurABEIL08}
(recall that the relationship between CaRet and NWTL remains unknown).
The same hardships, however, do not afflict the algorithmic complexity
of satisfiability and model checking, which are not higher than those of nested-words logics.
Thus, we argue that the strong gain in expressive power w.r.t.\ previous
approaches to model checking CFLs brought by POTL is worth the technicalities needed to
achieve the present ---and future--- results.

On the other hand, POTL shows promising results in its applications:~\cite{ChiariMP21} reports on a complete model-checker for POTL
and on the first encouraging experiments on a benchmark of practical interest.

In our view, POTL is the theoretical foundation on top of which to build a complete,
practical and user-friendly environment to specify and verify properties
of many pushdown-based systems.

While logics such as CaRet and NWTL can be viewed as the extension of LTL to Nested Words,
POTL can be seen as the extension of LTL to OPLs and OP words.
An interesting path for future investigations is the extension of branching-time logics
such as CTL to OPLs.
Something similar has been done for Nested Words in~\cite{DBLP:journals/toplas/AlurCM11},
where a $\mu$-calculus of Nested Trees is introduced.
Nested-words logics have also been augmented in other directions:~\cite{BozzelliS14} introduces a temporal logic capturing the whole class of VPLs,
while timed extensions of CaRet are given in~\cite{BozzelliMP18}.
The same extensions could be attempted for POTL too.

\bibliographystyle{alphaurl}
\bibliography{biblio}

\clearpage
\appendix

\section{Omitted Proofs: Properties of the \texorpdfstring{$\chain$}{Chain} Relation}%
\label{sec:chain-prop-proofs}

In the following lemma, we prove the properties of the chain relation.

\begin{lem}%
\label{lemma:chain-prop}
Given an OP word $w$ and positions $i, j$ in it, the following properties hold.
\begin{enumerate}
\item\label{item:chain-prop-1-proof}
  If $\chain(i,j)$ and $\chain(h,k)$, for any $h,k$ in $w$,
  then we have $i < h < j \implies k \leq j$
  and $i < k < j \implies i \leq h$.
\item If $\chain(i,j)$, then $i \lessdot i+1$ and $j-1 \gtrdot j$.
\item\label{item:chain-prop-3-proof}
  Consider all positions (if any) $i_1 < i_2 < \dots < i_n$ s.t.\ $\chain(i_p, j)$
  for all $1 \leq p \leq n$.
  We have $i_1 \lessdot j$ or $i_1 \doteq j$ and, if $n > 1$,
  $i_q \gtrdot j$ for all $2 \leq q \leq n$.
\item
  Consider all positions (if any) $j_1 < j_2 < \dots < j_n$ s.t.\ $\chain(i, j_p)$
  for all $1 \leq p \leq n$.
  We have $i \gtrdot j_n$ or $i \doteq j_n$ and, if $n > 1$,
  $i \lessdot j_q$ for all $1 \leq q \leq n-1$.
\end{enumerate}
\end{lem}
\begin{proof}
In the following, we denote by $c_p$ the character labeling word position $p$,
and by writing $\ochain{c_{-1}}{x_0 c_0 x_1 \dots x_n c_n x_{n+1}}{c_{n+1}}$ we imply
$c_{-1}$ and $c_{n+1}$ are the context of a simple or composed chain,
in which either $x_p = \varepsilon$, or $\ochain{c_{p-1}}{x_{p}}{c_p}$ is a chain, for each $p$.
\begin{enumerate}
\item Suppose, by contradiction, that $\chain(i,j)$, $\chain(h,k)$, and $i < h < j$, but $k > j$.
  Consider the case in which $\chain(i,j)$ is the innermost chain whose body contains $h$,
  so it is of the form $\ochain{c_i}{x_0 c_0 \dots c_h x_p c_p \dots c_n x_{n+1}}{c_j}$
  or $\ochain{c_i}{x_0 c_0 \dots c_h x_{n+1}}{c_j}$.
  By the definition of chain, we have either $c_h \doteq c_p$ or $c_h \gtrdot c_j$, respectively.

  Since $\chain(h,k)$, this chain must be of the form
  $\ochain{c_h}{x_p c_p \dots }{c_k}$ or $\ochain{c_h}{x_{n+1} c_j \dots }{c_k}$,
  implying $c_h \lessdot c_p$ or $c_h \lessdot c_j$, respectively.
  This means there is a conflict in the OPM, contradicting the hypothesis that $w$ is an OP word.

  In case $\chain(i,j)$ is not the innermost chain whose body contains $h$,
  we can reach the same contradiction by inductively considering the chain between $i$ and $j$
  containing $h$ in its body.
  Moreover, it is possible to reach a symmetric contradiction with the hypothesis
  $\chain(i,j)$, $\chain(h,k)$, and $i < k < j$, but $i > h$.

\item Trivially follows from the definition of chain.

\item
We prove that only $i_1$ can be s.t.\ $i_1 \lessdot j$ or $i_1 \doteq j$.
Suppose, by contradiction, that for some $r > 1$ we have $i_r \lessdot j$ or $i_r \doteq j$.

If $i_r \lessdot j$, by the definition of chain, $j$ must be part of the body
of another composed chain whose left context is $i_r$.
So, $w$ contains a structure of the form $\ochain{c_{i_r}}{x_0 c_j \dots}{c_k}$
where $|x_0| \geq 1$, $\ochain{c_{i_r}}{x_0}{c_j}$, and $k > j$ is s.t.\ $\chain(i_r, k)$.
This contradicts the hypothesis that $\chain(i_1, j)$,
because such a chain would cross $\chain(i_r,k)$,
contradicting property~\eqref{item:chain-prop-1-proof}.

If $i_r \doteq j$,
then $w$ contains a structure $\ochain{c_{i_{r-1}}}{\dots c_{i_r} x_{i_r}}{c_j}$,
with $|x_{i_r}| \geq 1$ and $\ochain{c_{i_r}}{x_{i_r}}{c_j}$.
By the definition of chain, we have $i_r \gtrdot j$, which contradicts the hypothesis.

Thus, the only remaining alternative for $r > 1$ is $i_r \gtrdot j$.

Similarly, if we had $i_1 \gtrdot j$, the definition of chain would lead to the existence
of a position $h < i_1$ s.t.\ $\chain(h, j)$, which contradicts the hypothesis that
$i_1$ is the leftmost of such positions.
$i_1 \lessdot j$ and $i_1 \doteq j$ do not lead to such contradictions.

\item The proof is symmetric to the previous one.
\qedhere
\end{enumerate}
\end{proof}

\section{Omitted Proofs: Expansion Laws}%
\label{sec:expansion-proofs}

We prove the following expansion laws for POTL:\@
\begin{align}
  \lguntil{t}{\chi}{\varphi}{\psi} &\equiv
    \psi \lor \Big(\varphi \land \big(\lnext^t (\lguntil{t}{\chi}{\varphi}{\psi})
      \lor \lcnext{t} (\lguntil{t}{\chi}{\varphi}{\psi})\big)\Big)%
  \label{eq:expansion-opsuntil} \\
  \lgsince{t}{\chi}{\varphi}{\psi} &\equiv
    \psi \lor \Big(\varphi \land \big(\lback^t (\lgsince{t}{\chi}{\varphi}{\psi})
      \lor \lcback{t} (\lgsince{t}{\chi}{\varphi}{\psi})\big)\Big)%
  \label{eq:expansion-opssince} \\
  \lhuuntil{\varphi}{\psi} &\equiv
    (\psi \land \lcdback \top \land \neg \lcuback \top) \lor
     \big(\varphi \land \lhunext (\lhuuntil{\varphi}{\psi})\big) \\
  \lhusince{\varphi}{\psi} &\equiv
    (\psi \land \lcdback \top \land \neg \lcuback \top) \lor
     \big(\varphi \land \lhuback (\lhusince{\varphi}{\psi})\big) \\
  \lhduntil{\varphi}{\psi} &\equiv
    (\psi \land \lcunext \top \land \neg \lcdnext \top) \lor
     \big(\varphi \land \lhdnext (\lhduntil{\varphi}{\psi})\big) \\
  \lhdsince{\varphi}{\psi} &\equiv
    (\psi \land \lcunext \top \land \neg \lcdnext \top) \lor
     \big(\varphi \land \lhdback (\lhdsince{\varphi}{\psi})\big)
\end{align}

\begin{lem}%
\label{lemma:expansion-law-su}
  Given a word $w$ on an OP alphabet $(\powset{AP}, M_{AP})$,
  two POTL formulas $\varphi$ and $\psi$,
  for any position $i$ in $w$ the following equivalence holds:
  \[
    \lcduntil{\varphi}{\psi} \equiv
    \psi \lor \Big(\varphi \land \big(\ldnext (\lcduntil{\varphi}{\psi})
      \lor \lcdnext (\lcduntil{\varphi}{\psi})\big)\Big).
  \]
\end{lem}
\begin{proof}
  \textbf{[Only if]}
  Suppose $\lcduntil{\varphi}{\psi}$ holds in $i$.
  If $\psi$ holds in $i$, the equivalence is trivially verified.
  Otherwise, $\lcduntil{\varphi}{\psi}$ is verified by a DSP
  $i = i_0 < i_1 < \dots < i_n = j$ with $n \geq 1$,
  s.t.\ $(w,i_p) \models \varphi$ for $0 \leq p < n$ and $(w,i_n) \models \psi$.
  Note that, by the definition of DSP, any suffix of that path is also a DSP ending in $j$.
  Consider position $i_1$: $\varphi$ holds in it, and it can be either
  \begin{itemize}
  \item $i_1 = i+1$.
    Then either $i \lessdot (i+1)$ or $i \doteq (i+1)$,
    and path $i_1 < i_2 < \dots < i_n = j$ is the DSP between $i_1$ and $j$,
    and $\varphi$ holds in all $i_p$ with $1 \leq p < n$, and $\psi$ in $j_n$.
    So, $\lcduntil{\varphi}{\psi}$ holds in $i_1$,
    and $\ldnext (\lcduntil{\varphi}{\psi})$ holds in $i$.
  \item $i_1 > i+1$.
    Then, $\chain(i,i_1)$, and $i \lessdot i_1$ or $i \doteq i_1$.
    Since $i_1 < i_2 < \dots < i_n = j$ is the DSP from $i_1$ to $j$,
    $\lcduntil{\varphi}{\psi}$ holds in $i_1$,
    and so does $\lcdnext (\lcduntil{\varphi}{\psi})$ in $i$.
  \end{itemize}

\noindent
  \textbf{[If]}
  Suppose the right-hand side of the equivalence holds in $i$.
  The case $(w,i) \models \psi$ is trivial, so suppose $\psi$ does not hold in $i$.
  Then $\varphi$ holds in $i$, and either:
  \begin{itemize}
  \item $\ldnext (\lcduntil{\varphi}{\psi})$ holds in $i$.
    Then, we have $i \lessdot (i+1)$ or $i \doteq (i+1)$,
    and there is a DSP $i+1 = i_1 < i_2 < \dots < i_n = j$,
    with $\varphi$ holding in all $i_p$ with $1 \leq p < n$, and $\psi$ in $i_n$.
    \begin{itemize}
    \item If $i \doteq (i+1)$, it is not the left context of any chain,
      and $i = i_0 < i_1 < i_2 < \dots < i_n$ is a DSP satisfying $\lcduntil{\varphi}{\psi}$ in $i$.
    \item Otherwise, let $k = \min\{h \mid \chain(i,h)\}$:
      we have $k > j$, because a DSP cannot cross right chain contexts.
      So, adding $i$ to the DSP generates another DSP,
      because there is no position $h$ s.t.\ $\chain(i,h)$ with $h \leq j$,
      and the successor of $i$ in the path can only be $i_1 = i+1$.
    \end{itemize}
  \item $\lcdnext (\lcduntil{\varphi}{\psi})$ holds in $i$.
    Then, there exists a position $k$ s.t.\ $\chain(i,k)$
    and $i \mathrel\lessdot k$ or $i \doteq k$
    and $\lcduntil{\varphi}{\psi}$ holds in $k$,
    because of a DSP $k = i_1 < i_2 < \dots < i_n = j$.
    If $k = \max\{h \mid h \leq j \land \chain(i,h) \land (i \lessdot h \lor i \doteq h)\}$,
    then $i = i_0 < i_1 < i_2 < \dots < i_n$ is a DSP by definition,
    and since $\varphi$ holds in $i$, $\lcduntil{\varphi}{\psi}$ is satisfied in it.
    Otherwise, let
    $k' = \max\{h \mid h \leq j \land \chain(i,h) \land (i \lessdot h \lor i \doteq h)\}$.
    Since $i_1 > i$ and chains cannot cross,
    there exists a value $q$, $1 < q \leq n$, s.t.\ $i_q = k'$.
    Thus $i_q < i_{q+1} < \dots < i_n = j$ is a DSP,
    so $\lcduntil{\varphi}{\psi}$ holds in $i_q$ too.
    The path $i < i_q < \dots < i_n$ is a DSP, and $\lcduntil{\varphi}{\psi}$ holds in $i$.
\qedhere
  \end{itemize}
\end{proof}

The proofs for the summary since and upward summary until operators are symmetric.

\begin{lem}%
\label{lemma:expansion-law-hu}
  Given a word $w$ on an OP alphabet $(\powset{AP}, M_{AP})$,
  and two POTL formulas $\varphi$ and $\psi$,
  for any position $i$ in $w$ the following equivalence holds:
  \[
    \lhuuntil{\varphi}{\psi} \equiv
    (\psi \land \lcdback \top \land \neg \lcuback \top) \lor
     \big(\varphi \land \lhunext (\lhuuntil{\varphi}{\psi})\big).
  \]
\end{lem}
\begin{proof}
  \textbf{[Only if]}
  Suppose $\lhuuntil{\varphi}{\psi}$ holds in $i$.
  Then, there exists a path $i = i_0 < i_1 < \dots < i_n$, $n \geq 0$,
  and a position $h < i$ s.t.\ $\chain(h,i_p)$ and $h \lessdot i_p$ for each $0 \leq p \leq n$,
  $\varphi$ holds in all $i_q$ for $0 \leq q < n$, and $\psi$ holds in $i_n$.
  If $n = 0$, $\psi$ holds in $i = i_0$, and so does $\lcdback \top$, but $\lcuback \top$ does not.
  Otherwise, the path $i_1 < \dots < i_n$ is also a UHP,
  so $\lhuuntil{\varphi}{\psi}$ is true in $i_1$.
  Therefore, $\varphi$ holds in $i$, and so does $\lhunext (\lhuuntil{\varphi}{\psi})$.

  \textbf{[If]}
  If $\lcdback \top$ holds in $i$ but $\lcuback \top$ does not,
  then there exists a position $h < i$ s.t.\ $\chain(h,i)$ and $h \lessdot i$.
  If $\psi$ also holds in $i$, then $\lhuuntil{\varphi}{\psi}$ is trivially satisfied
  in $i$ by the path made of only $i$ itself.
  Otherwise, if $\lhunext (\lhuuntil{\varphi}{\psi})$ holds in $i$,
  then there exist a position $h < i$ s.t.\ $\chain(h,i)$ and $h \lessdot i$,
  and a position $i_1$ which is the minimum one s.t.\ $i_1 > i$, $\chain(h,i_1)$ and $h \lessdot i_1$.
  In $i_1$, $\lhuuntil{\varphi}{\psi}$ holds,
  so it is the first position of a UHP $i_1 < i_2 < \dots < i_n$.
  Since $\varphi$ also holds in $i$, the path $i = i_0 < i_1 < \dots < i_n$ is also a UHP,
  satisfying $\lhuuntil{\varphi}{\psi}$ in $i$.
\end{proof}
The proofs for the other hierarchical operators are analogous.

\end{document}